\documentclass[11pt]{article}
\usepackage{graphicx}
\usepackage[margin=1in]{geometry}
\usepackage{amsmath,amsthm,amssymb}
\usepackage{verbatim}
\usepackage{algorithm}
\usepackage{paralist}
\usepackage{color}

\newtheorem{theorem}{Theorem}
\numberwithin{theorem}{section}
\newtheorem{lemma}[theorem]{Lemma}
\newtheorem{corollary}[theorem]{Corollary}

\newcommand{\stronger}{\succeq}
\newcommand{\tab}{\hspace{5mm}}
\newcommand{\FD}{D}

\newcommand{\set}[1]{\left\{#1\right\}}

\DeclareMathOperator*{\argmin}{arg\,min}

\title{Consensus using Asynchronous Failure Detectors}
\author{Nancy Lynch \\
CSAIL, MIT
\and Srikanth Sastry\footnote{The author is currently affiliated with Google Inc.}\\
CSAIL, MIT}
\date{}

\begin{document}

\maketitle

\begin{abstract}
The FLP result shows that crash-tolerant consensus is impossible to solve in asynchronous systems, and several solutions have been proposed for crash-tolerant consensus under alternative (stronger) models. One popular approach is to augment the asynchronous system with appropriate failure detectors, which provide (potentially unreliable) information about process crashes in the system, to circumvent the FLP impossibility.  

In this paper, we demonstrate the exact mechanism by which (sufficiently powerful) asynchronous failure detectors enable solving crash-tolerant consensus. Our approach, which borrows arguments from the FLP impossibility proof and the famous result from \cite{chan:twfdf}, which shows that $\Omega$ is a weakest failure detector to solve consensus, also yields a natural proof to $\Omega$ as a weakest asynchronous failure detector to solve consensus. The use of I/O automata theory in our approach enables us to model execution in a more detailed fashion than \cite{chan:twfdf} and also addresses the latent assumptions and assertions in the original result in \cite{chan:twfdf}.
\end{abstract}

\section{Introduction}

%
%
%

In~\cite{cornejoetalAFD,cornejoetalAFD-TR}
we introduced a new formulation of failure detectors.  Unlike the
traditional failure detectors of~\cite{chan:ufdfr, chan:twfdf}, ours are modeled as
asynchronous automata, and defined in terms of the general I/O
automata framework for asynchronous concurrent systems.
To distinguish our failure detectors from the traditional ones, we
called ours ``Asynchronous Failure Detectors (AFDs)''.

In terms of our model, we presented many of the standard results of
the field and some new results.
Our model narrowed the scope of failure detectors sufficiently so that
AFDs satisfy several desirable properties, which are not true of the
general class of traditional failure detector. For example,
(1) AFDs are self-implementable; (2) if an AFD $D'$ is strictly stronger
than another AFD $D$, then $D'$ is sufficient to solve a strict superset
of the problems solvable by $D$. See \cite{cornejoetalAFD-TR} for details. 
Working entirely within an asynchronous framework allowed us to take
advantage of the general results about I/O automata and to prove
our results rigorously without too much difficulty.

In this paper, we investigate the role of asynchronous failure detectors in circumventing the impossibility of crash-tolerant consensus in asynchronous systems (FLP) \cite{FLP}. Specifically, we demonstrate exactly how sufficiently strong AFDs circumvent
the FLP impossibility.  We borrow ideas from the important related result by Chandra, Hadzilacos,
and Toueg~\cite{chan:twfdf} that says that the failure detector $\Omega$ is a
``Weakest Failure Detector'' that solves the consensus
problem. Incidentally, the proof in \cite{chan:twfdf} make certain implicit assumptions
and assertions which are entirely reasonable and true, respectively. However, for
the purpose of rigor, it is desirable that these assumptions be explicit
and these assertions be proved. Our demonstration of how sufficiently strong AFDs circumvent FLP dovetails effortlessly with an analogous proof of  ``weakest AFD'' for consensus.

While our proof generally follows the proof in~\cite{chan:twfdf}, we
state the (implicit) assumptions and assertions from \cite{chan:twfdf} explicitly.
Since our framework is entirely asynchronous and all our definitions are based
on an established concurrency theory foundation, we are able to provide
rigorous proofs for the (unproven) assertions from \cite{chan:twfdf}. In order to prove the main result of this paper, we modified certain definitions from \cite{cornejoetalAFD-TR}. However, these modifications do not invalidate any of the results from \cite{cornejoetalAFD,cornejoetalAFD-TR}.

The rest of this paper is organized as follows. Section \ref{sec:approachAndContributions} outlines the approach that we use in this paper and its major contributions. In section \ref{sec:comarisonWithCHT}, we compare our proof with the original CHT proof in \cite{chan:twfdf}. Sections \ref{sec:ioautomata} through~\ref{sec:solvingProblems} introduce I/O automata and the definitions of a problem, of an asynchronous system, and of AFDs; much of the material is summarized
from~\cite{cornejoetalAFD,cornejoetalAFD-TR}. Section \ref{sec:observations} introduces the notion of \emph{observations} of AFD behavior, which are a key part of showing that $\Omega$ is a weakest AFD to solve consensus; this section proves several useful properties of observations which are central to the understanding of the proof and are a contribution of our work. In Section \ref{subsec:treeOfExec}, we introduce \emph{execution trees} for any asynchronous system that uses an AFD; we construct such trees from observations introduced in Section \ref{sec:observations}. We also prove several properties of such execution trees, which may be of independent interest and useful in analysis of executions in any AFD-based system. In Section \ref{sec:consensusAndAFD}, we formally define the consensus problem and use the notions of observations and execution trees to demonstrate how sufficiently strong AFDs enable asynchronous systems to circumvent the impossibility of fault tolerant consensus in asynchronous systems \cite{FLP}; Section \ref{sec:consensusAndAFD} defines and uses \emph{decision gadgets} in an execution tree to demonstrate this; it also shows that the set of such decision gadgets is countable, and therefore, any such execution tree contains a ``first'' decision gadget. Furthermore, Section \ref{sec:consensusAndAFD} also shows that each decision gadget is associated with a location that is live and never crashes; we call it the \emph{critical location} of the decision gadget. In Section \ref{sec: wfd}, we show that $\Omega$ is a weakest AFD to solve consensus by presenting a distributed algorithm that simulates the output of $\Omega$. The algorithm constructs observations and execution trees, and it eventually identifies the ``first'' decision gadget and its corresponding critical location; the algorithm outputs this critical location as the output of the simulated $\Omega$ AFD, thus showing that $\Omega$ is a weakest AFD for consensus.

\section{Approach and contributions}\label{sec:approachAndContributions}
To demonstrate our results, we start with a complete definition of asynchronous systems and AFDs. Here, we modified the definitions of AFD from \cite{cornejoetalAFD,cornejoetalAFD-TR}, but we did so without invalidating earlier results. We argue that the resulting definition of AFDs is more natural and models a richer class of behaviors in crash-prone asynchronous systems. Next, we introduce the notion of \emph{observations} of AFD behavior (Section \ref{sec:observations}), which are DAGs that model a partial ordering AFD outputs are different processes; importantly, the knowledge of this partial order can be gained by any process through asynchronous message passing alone. \emph{Observations as a tool for modeling AFD behavior is of independent interest}, and we prove several important properties of observations that are used in our later results.

From such observations, we construct trees of executions of arbitrary AFD-based systems; again, \emph{such trees are of independent interest}, and we prove several important properties of such trees that are used later.

Next, we define the consensus problem and the notion \emph{valence}. Roughly speaking, a finite execution of a system is univalent if all its fair extensions result in the same decision value and the execution is bivalent if some fair extension results in a decision value $1$ and another fair extension results in a decision value $0$. We present our first important result using observations and execution trees; \emph{we show that a sufficiently powerful AFD guarantees that in the execution tree constructed from any viable\footnote{Informally, an observation is viable if it can be constructed from an AFD trace.} observation of AFD outputs, the events responsible for the transition from a bivalent execution to a univalent execution must occur at location that does not crash}. Such transitions to univalent executions correspond to so-called ``decision gadgets'', and the live location corresponding to such transitions is called the ``critical location'' of the decision gadgets.

Next, we use the aforementioned result to show that $\Omega$ is a weakest AFD to solve consensus. In order to do so, we first define a metric function that orders all the decision gadgets. This metric function satisfies an important stability property which guarantees the following. Given the decision gadget with the smallest metric value in a given infinite execution tree, for any sufficiently large, but finite, subtree, the same decision gadget will have the smallest metric value within that subtree. Note that the original proof in \cite{chan:twfdf} did not provide such a metric function, and we contend that this is an essential compoenent for completing this proof. We then construct an emulation algorithm (similar to the one in \cite{chan:twfdf}) that uses an AFD sufficiently powerful to solve consensus and simulates the output of $\Omega$. In this algorithm processes exchange AFD outputs and construct finite observations and corresponding finite execution trees. The aforementioned stability property ensures that eventually forever, each process that does not crash identifies the same decision gadget as the one with the smallest metric value. Recall that the critical location of any decision gadget is guaranteed to not crash. Therefore, eventually forever, each process that does not crash identifies the same correct process and outputs that correct process as the output of the simulated $\Omega$ AFD.

\section{Comparisons with the original CHT proof}\label{sec:comarisonWithCHT}
Our proof has elements that are very similar to the the original CHT proof from~\cite{chan:twfdf}. However, despite the similarity in our arguments, our proof deviates from the CHT proof in some subtle, but significant ways.

\subsection{Observations}
In \cite{chan:twfdf}, the authors introduce DAGs with special properties that model the outputs of a failure detector at different processes and establishes partial ordering of these outputs. In our proof, the analogous structure is an observation (See Section \ref{sec:observations}). However, our notion of an observation is much more general than the DAG introduced in \cite{chan:twfdf}.

First, the DAG in \cite{chan:twfdf} is an infinite graph and cannot model failure detector outputs in finite executions. In contrast, observations may be finite or infinite. Second, we also introduce the notion of a sequence of finite observations that can be constructed from progressively longer finite executions that enable us to model the evolution of observations and execution trees as failure detector outputs become available. Such detailed modeling and analysis does not appear in \cite{chan:twfdf}.

\subsection{Execution trees}
In \cite{chan:twfdf}, each possible input to consensus gives rise to a unique execution tree from the DAG. Thus, for $n$ processes, there are $2^n$ possible trees that constitute a forest a trees. In contrast, our proof constructs exactly one tree that models the executions of all possible inputs to consensus. This change is not merely cosmetic. It simplifies analysis and makes the proof technique more general in the following sense.

The original proof in \cite{chan:twfdf} cannot be extended to understanding long-lived problems such as iterative consensus or mutual exclusion. The simple reason for this is that the number of possible inputs for such problems can be uncountably infinite, and so the number of trees generated by the proof technique in \cite{chan:twfdf} is also uncountably infinite. This introduces significant challenges in extracting any structures within these trees by a distributed algorithm. In contrast, in our approach, the execution tree will remain virtually the same; only the rules for determining the action tag values at various edges change.

\subsection{Determining the ``first'' decision gadget}
In \cite{chan:twfdf} and in our proof, a significant result is that there are infinite, but countable number of decision gadgets, and therefore there exists a unique enumeration of the decision gadgets such that one of them is the ``first'' one. This result is then used in \cite{chan:twfdf}  to claim that all the emulation algorithms converge to the same decision gadget. However, \cite{chan:twfdf} does not provide any proof of this claim. Furthermore, we show that this proving this claim in non-trivial.

The significant gap in the original proof in \cite{chan:twfdf} is the following. During the emulation, each process constructs only finite DAGs, that are subgraphs of some infinite DAG with the required special properties. However, since the DAGs are finite, the trees of executions constructed from this DAG could incorrectly detect certain parts of the trees as being decision gadgets, when in the execution tree of the infinite DAG, these are not decision gadgets. Each such pseudo decision gadget, is eventually deemed to not be a decision gadget, as the emulation progresses. However, there can be infinitely many such pseudo gadgets. Thus, given any arbitrary enumeration of decision gadgets, it is possible that such pseudo decision gadgets appears infinitely often, and are enumerated ahead of the ``first'' decision gadget. Consequently, the emulation never stabilizes to the first decision gadget.

In our proof, we address is gap by carefully defining metric functions for nodes and decision gadgets so that eventually, all the pseudo decision gadgets are ordered after the eventual ``first'' decision gadget.

\section{I/O Automata}\label{sec:ioautomata}

We use the I/O Automata framework~\cite{lynch:hapf,lynch:aiti,lync:da}
%
for specifying the system model and failure detectors. 
Briefly, an I/O automaton models a component of a distributed system
as a (possibly infinite) state machine that interacts with other state
machines through discrete actions. 
This section summarizes the I/O-Automata-related definitions that we
use in this paper. 
See \cite[Chapter 8]{lync:da} for a thorough description of I/O
Automata.

\subsection{Automata Definitions}
\label{sec: IOA: automata-defs}

An \emph{I/O automaton}, which we will usually refer to as simply an
``automaton'', consists of five components:
a signature, a set of states, a set of initial states, a
state-transition relation, and a set of tasks. 
We describe these components next.

\paragraph{Actions, Signature, and Tasks.} 
The state transitions of an automaton are associated with named
\emph{actions}; we denote the set of actions of an automaton $A$ by
$act(A)$.
Actions are classified as \emph{input}, \emph{output}, or
\emph{internal}, and this classification constitutes the
\emph{signature} of the automaton. 
We denote the sets of input, output, and internal actions of an
automaton $A$ by $input(A)$, $output(A)$, and $internal(A)$,
respectively. 
Input and output actions are collectively called the \emph{external}
actions, denoted $external(A)$, and output and internal actions are
collectively called the \emph{locally controlled} actions. 
The locally controlled actions of an automaton are partitioned into
\emph{tasks}. 
Tasks are used in defining fairness conditions on executions of the
automaton, as we describe in Section~\ref{sec: IOA: fairness}.

Internal actions of an automaton are local to the automaton itself
whereas external (input and output) actions are available for
interaction with other automata.
Locally controlled actions are initiated by the automaton itself,
whereas input actions simply arrive at the automaton from the outside, 
without any control by the automaton.

\paragraph{States.} 
The states of an automaton $A$ are denoted by $states(A)$; some
(non-empty) subset $init(A) \subseteq states(A)$ is designated as the
set of \emph{initial states}. 

\paragraph{Transition Relation.}
The state transitions of an automaton $A$ are defined by a
\emph{state-transition relation} $trans(A)$, which is a set
of tuples of the form $(s,a,s')$ where $s,s'\in states(A)$ and $a \in
act(A)$. 
Each such tuple $(s,a,s')$ is a \emph{transition}, or a \emph{step},
of $A$.
Informally speaking, each step $(s,a,s')$ denotes the following
behavior:  automaton $A$, in state $s$, performs action $a$ and
changes its state to $s'$. 

For a given state $s$ and action $a$, if $trans(A)$ contains some step
of the form $(s,a,s')$, then $a$ is said to be \emph{enabled} in
$s$. 
We assume that every input action in $A$ is enabled in every state of
$A$; that is, for every input action $a$ and every state $s$,
$trans(A)$ contains a step of the form $(s,a,s')$.
A task $C$, which is a set of locally controlled actions, is said to
be \emph{enabled} in a state $s$ iff some action in $C$ is enabled in
$s$.

\paragraph{Deterministic Automata.}
The general definition of an I/O automaton permits multiple locally
controlled actions to be enabled in any given state. 
It also allows the resulting state after performing a given action to
be chosen nondeterministically. 
For our purposes, it is convenient to consider a class of I/O automata
whose behavior is more restricted. 

We define an action $a$ (of an automaton $A$) to be
\emph{deterministic} provided that, for every state $s$, $trans(A)$
contains at most one transition of the form $(s,a,s')$.
We define an automaton $A$ to be \emph{task deterministic} iff (1) for
every task $C$ and every state $s$ of $A$, at most one action in $C$
is enabled in $s$, and (2) all the actions in $A$ are deterministic. 
An automaton is said to be \emph{deterministic} iff it
is task deterministic, has exactly one task, and has a unique start
state.

\subsection{Executions, Traces, and Schedules}

Now we define how an automaton executes. 
An \emph{execution fragment} of an automaton $A$ is a finite sequence
$s_0,a_1,s_1,a_2,\ldots,s_{k-1},a_k,s_{k}$, or an infinite sequence
$s_0,a_1,s_1,a_2,\ldots,s_{k-1},a_k,s_{k},\ldots$, of alternating
states and actions of $A$ such that for every $k \geq 0$, $(s_k,a_{k+1},s_{k+1})$ is in $trans(A)$.
%
%
A sequence consisting of just a state is a special case of an
execution fragment and is called a \emph{null execution fragment}. 
Each occurrence of an action in an execution fragment is called an 
\emph{event}. 

An execution fragment that starts with an initial state (that is, $s_0
\in init(A)$) is called an \emph{execution}. 
A null execution fragment consisting of an initial state is called a
\emph{null execution}.
A state $s$ is said to be \emph{reachable} if there exists a finite
execution that ends with $s$. 
By definition, any initial state is reachable.

We define concatenation of execution fragments.
Let $\alpha_1$ and $\alpha_2$ be two execution fragments of an I/O
automaton such that $\alpha_1$ is finite and the final state of
$\alpha_1$ is also the starting state of $\alpha_2$, and let
$\alpha_2'$ denote the sequence obtained by deleting the first state
in $\alpha_2$. 
Then the expression $\alpha_1 \cdot \alpha_2$ denotes the execution
fragment formed by appending $\alpha_2'$ after $\alpha_1$.

It is sometimes useful to consider just the sequence of events that
occur in an execution, ignoring the states.
Thus, given an execution $\alpha$, the \emph{schedule} of $\alpha$ is
the subsequence of $\alpha$ that consists of all the events in
$\alpha$, both internal and external.
The \emph{trace} of an execution includes only the externally
observable behavior;
formally, the trace $t$ of an execution $\alpha$ is the subsequence of
$\alpha$ consisting of all the external actions. 

More generally, we define the \emph{projection} of any sequence on a
set of actions as follows.
Given a sequence $t$ (which may be an execution fragment, schedule, or
trace) and a set $B$ of actions, the projection of $t$ on $B$, denoted
by $t|_B$, is the subsequence of $t$ consisting of all the events from
$B$.

We define concatenation of schedules and traces. 
Let $t_1$ and $t_2$ be two sequences of actions of some I/O automaton
where $t_1$ is finite; then $t_1 \cdot t_2$ denotes the sequence
formed by appending $t_2$ after $t_1$.

To designate specific events in a schedule or trace, we use the
following notation:
if a sequence $t$ (which may be a schedule or a trace) contains at
least $x$ events, then $t[x]$ denotes the $x^{th}$ event in the
sequence $t$, and otherwise, $t[x] = \bot$. 
Here, $\bot$ is a special symbol that we assume is different from the
names of all actions. 

\subsection{Operations on I/O Automata}

\paragraph{Composition.} 
A collection of I/O automata may be composed by matching output
actions of some automata with the same-named input actions of
others.\footnote{Not all collections of I/O automata may be
  composed. For instance, in order to compose a collection of I/O
  automata, we require that no two automata have a common output
  action.  See \cite[chapter 8]{lync:da} for details.}
Each output of an automaton may be matched with inputs of any number
of other automata. 
Upon composition, all the actions with the same name are performed
together.

Let $\alpha = s_0,a_1,s_1,a_2,\ldots$ be an execution of the
composition of automata $A_1, \dots, A_N$. 
The \emph{projection} of $\alpha$ on automaton $A_i$, where $i \in
[1,N]$, is denoted by $\alpha|A_i$ and is defined to be the
subsequence of $\alpha$ obtained by deleting each pair $a_k,s_k$
for which $a_k$ is not an action of $A_i$ and replacing each remaining
state $s_k$ by automaton $A_i$'s part of $s_k$. 
Theorem~8.1 in~\cite{lync:da} states that if $\alpha$ is an execution
of the composition $A_1, \dots, A_N$, then for each $i \in [1,N]$,
$\alpha|A_i$ is an execution of $A_i$. 
Similarly, if $t$ is a trace of of $A_1,\dots, A_N$, then for
each $i$, $t|A_i$ is an trace of $A_i$.


\paragraph{Hiding.} 
In an automaton $A$, an output action may be ``hidden'' by
reclassifying it as an internal action. 
A hidden action no longer appears in the traces of the automaton.

\subsection{Fairness}
\label{sec: IOA: fairness}

When considering executions of an I/O automaton, we will often be
interested in those executions in which every task of the automaton
gets infinitely many turns to take steps; we call such executions
``fair''.
When the automaton represents a distributed systems, the notion of
fairness can be used to express the idea that all system components
continue to get turns to perform their activities.

Formally, an execution fragment $\alpha$ of an automaton $A$ is said
to be \emph{fair} iff the following two conditions hold for
every task $C$ in $A$. (1) If $\alpha$ is finite, then no action in
$C$ is enabled in the final state of $\alpha$. (2) If $\alpha$ is
infinite, then either (a) $\alpha$ contains infinitely many events
from $C$, or (b) $\alpha$ contains infinitely many occurrences of
states in which $C$ is not enabled.

A schedule $\sigma$ of $A$ is said to be \emph{fair} if it is the
schedule of a fair execution of $A$. 
Similarly, a trace $t$ of $A$ is said to be \emph{fair} if it is the
trace of a fair execution of $A$.


\section{Crash Problems}

In this section, we define problems, distributed problems,
crash problems, and failure-detector problems.
We also define a particular failure-detector problem corresponding to
the leader election oracle $\Omega$ of~\cite{chan:twfdf}.

\subsection{Problems} 
\label{subsec:problems}

We define a \emph{problem} $P$ to be a tuple $(I_P,O_P,T_P)$, where
$I_P$ and $O_P$ are disjoint sets of actions and $T_P$ is a set of
(finite or infinite) sequences over these actions such that there
exists an automaton $A$ where $input(A) = I_P$, $output(A) = O_P$, and
the set of fair traces of $A$ is a subset of $T_P$.  
In this case we state that $A$ \emph{solves} $P$.
We include the aforementioned assumption of solvability to satisfy a non-triviality
property, which we explain in Section~\ref{sec:solvingProblems}.

\paragraph{Distributed Problems.}
Here and for the rest of the paper, we introduce a fixed finite set
$\Pi$ of $n$ location IDs; we assume that $\Pi$ does not contain the
special symbol $\bot$. We assume a fixed total ordering $<_\Pi$ on $\Pi$.
We also assume a fixed mapping $loc$ from actions to $\Pi \cup \{ \bot
\}$; for an action $a$, if $loc(a) = i \in \Pi$, then we say that $a$
\emph{occurs at} $i$. 
A problem $P$ is said to be \emph{distributed} over $\Pi$ if, for every
action $a \in I_P \cup O_P$, $loc(a) \in \Pi$. 
We extend the definition of $loc$ by defining $loc(\bot) = \bot$.

Given a problem $P$ that is distributed over $\Pi$, and a location $i
\in \Pi$, $I_{P,i}$ and $O_{P,i}$ denote the set of actions in $I_P$
and $O_P$, respectively, that occur at location $i$;
that is, $I_{P,i} = \set{a|(a\in I_P) \wedge (loc(a)=i)}$ and $O_{P,i}
= \set{a|(a\in O_P) \wedge (loc(a)=i)}$. 


\paragraph{Crash Problems.}
We assume a set $\hat{I} = \set{crash_i \vert i\in \Pi}$ of crash
events, where $loc(crash_i) = i$. 
That is, $crash_i$ represents a crash that occurs at location $i$.
A problem $P = (I_P,O_P,T_P)$ that is distributed over $\Pi$ is said
to be a \emph{crash problem} iff $\hat{I} \subseteq I_P$. 
That is, $crash_i \in I_{P,i}$ for every $i \in \Pi$.

Given a (finite or infinite) sequence $t \in T_P$, $faulty(t)$ denotes
the set of locations at which a $crash$ event occurs in $t$. 
Similarly, $live(t) = \Pi \setminus faulty(t)$ denotes the set of
locations at which a $crash$ event does not occur in $t$. 
A location in $faulty(t)$ is said to be \emph{faulty} in $t$, and
a location in $live(t)$ is said to be \emph{live} in $t$. 

%

\subsection{Failure-Detector Problems}
\label{subsec:FDproblems}

Recall that a failure detector is an oracle that provides information
about crash failures. 
In our modeling framework, we view a failure detector as a special
type of crash problem.
A necessary condition for a crash problem $P = (I_P,O_P,T_P)$ to
be an asynchronous failure detector (AFD) is \emph{crash exclusivity},
which states that $I_P = \hat{I}$; that is, the actions $I_P$ are
exactly the $crash$ actions. 
Crash exclusivity guarantees that the only inputs to a failure
detector are the $crash$ events, and hence, failure detectors
provide information only about crashes. 
An AFD must also satisfy additional properties, which we describe next.

Let $\FD = (\hat{I},O_{\FD},T_{\FD})$ be a crash problem satisfying
crash exclusivity. 
We begin by defining a few terms that will be used in the definition of
an AFD. Let $t$ be an arbitrary sequence over $\hat{I}\cup O_{\FD}$.

\paragraph{Valid sequence.} 
The sequence $t$ is said to be \emph{valid} iff 
(1) for every $i \in \Pi$, no event in $O_{{\FD},i}$ (the set of
actions in $O_{\FD}$ at location $i$) occurs after a $crash_i$ event
in $t$, and 
(2) if no $crash_i$ event occurs in $t$, then $t$ contains infinitely
many events in $O_{{\FD},i}$.

Thus, a valid sequence contains no output events at a location $i$
after a $crash_i$ event, and contains infinitely many output events at
each live location.

\paragraph{Sampling.} 
A sequence $t'$ is a \emph{sampling} of $t$ iff (1) $t'$ is a
subsequence of $t$, (2) for every location $i \in \Pi$, (a) if $i$ is
live in $t$, then $t'|_{O_{{\FD},i}} = t|_{O_{{\FD},i}}$, and (b) if
$i$ is faulty in $t$, then  $t'$ contains the first $crash_i$ event in
$t$, and $t'|_{O_{{\FD},i}}$ is a prefix of $t|_{O_{{\FD},i}}$.

A sampling of sequence $t$ retains all events at live locations.
For each faulty location $i$, it may remove a suffix of the outputs at
location $i$.  It may also remove some crash events, but must retain
the first crash event.

\paragraph{Constrained Reordering.}
Let $t'$ be a valid permutation of events in $t$; $t'$ is a
\emph{constrained reordering} of $t$ iff the following is true. For
every pair of events $e$ and $e'$, if (1) $e$ precedes $e'$ in $t$, and
(2) either (a) $e,e' \in O_D$ and $loc(e) = loc(e')$, or (b) $e \in \hat{I}$ and $e' \in O_D$, then $e$ precedes
$e'$ in $t'$ as well.\footnote{Note that the definition of constrained reordering is less restrictive than the definition in \cite{cornejoetalAFD,cornejoetalAFD-TR}; specifically, unlike in \cite{cornejoetalAFD,cornejoetalAFD-TR}, this definition allow crashes to be reordered with respect to each other. However, this definition is ``compatible'' with the earlier definition in the sense that the results presented in \cite{cornejoetalAFD,cornejoetalAFD-TR} continue to be true under this new definition.}

A constrained reordering of sequence $t$ maintains the relative
ordering of events that occur at the same location and maintains the
relative order between any $crash$ event and any subsequent event.

\paragraph{Crash Extension.} Assume that $t$ is a finite sequence. A
crash extension of $t$ is a (possibly infinite) sequence $t'$ such
that $t$ is a prefix of $t'$ and the suffix of $t'$ following $t$ is a
sequence over $\hat{I}$.

In other words, a crash extension of $t$ is obtained by extending $t$
with $crash$ events.

\paragraph{Extra Crashes.} 
An \emph{extra crash} event in $t$ is a $crash_i$ event in $t$, for
some $i$, such that $t$ contains a preceding $crash_i$. 

An extra crash is a crash event at a location that has already
crashed.

\paragraph{Minimal-Crash Sequence.} 
Let $mincrash(t)$ denote the subsequence of $t$ that contains all the
events in $t$, except for the extra crashes; $mincrash(t)$ is called
the \emph{minimal-crash sequence} of $t$.

\paragraph{Asynchronous Failure Detector.}
Now we are ready to define asynchronous failure detectors.
A crash problem of the form $\FD = (\hat{I},O_{\FD},T_{\FD})$
(which satisfies crash exclusivity) is an \emph{asynchronous failure
  detector} (AFD, for short) iff $\FD$ satisfies the following
properties. 
\begin{enumerate}
\item \textbf{Validity.}
Every sequence $t \in T_{\FD}$ is valid.

\item \textbf{Closure Under Sampling}. For every sequence $t \in
  T_{\FD}$, every sampling of $t$ is also in $T_{\FD}$.

\item \textbf{Closure Under Constrained Reordering.} For every
  sequence $t \in T_{\FD}$, every constrained reordering $t$ is also
  in $T_{\FD}$. 

\item \textbf{Closure Under Crash Extension.} For every sequence $t
  \in T_{\FD}$, for every prefix $t_{pre}$ of $t$, for every crash
  extension $t'$ of $t_{pre}$, the following are true. (a) If $t'$ is finite, then $t'$ is a
  prefix of some sequence in $T_{\FD}$. (b) If $faulty(t') = \Pi$,
  then $t'$ is in $T_{\FD}$.

\item \textbf{Closure Under Extra Crashes.} 
For every sequence $t \in T_{\FD}$, every sequence $t'$ such that
$mincrash(t) = mincrash(t')$ is also in $T_{\FD}$.
\end{enumerate}

Of the properties given here, the first three---validity and closure
under sampling and constrained reordering---were also used in our
earlier papers~\cite{cornejoetalAFD,cornejoetalAFD-TR}.
The other two closure properties---closure under crash extension and
extra crashes---are new here.  

A brief motivation for the above properties is in order. The validity
property ensures that (1) after a location crashes, no outputs occur
at that location, and (2) if a location does not crash, outputs occur
infinitely often at that location. Closure under sampling permits a
failure detector to ``skip'' or ``miss'' any suffix of outputs at a faulty
location. Closure under constrained reordering permits ``delaying''
output events at any location. Closure under crash extension permits a
crash event to occur at any time. Finally, closure under extra crashes
captures the notion that once a location is crashed, the occurrence of
additional crash events (or lack thereof) at that location has no
effect.

We define one additional constraint, below.
This contraint is a formalization of an implicit assumption made in 
\cite{chan:twfdf}; namely, for any AFD $D$, any ``sampling'' (as
defined in \cite{char:isolt}) of a failure detector sequence in $T_D$
is also in $T_D$. 
%

\paragraph{Strong-Sampling AFDs.}
Let $\FD$ be an AFD, $t \in T_{\FD}$. 
A subsequence $t'$ of $t$ is said to be a \emph{strong sampling} of
$t$ if $t'$ is a valid sequence.
AFD $\FD$ is said to satisfy \emph{closure under strong sampling} if,
for every trace $t \in T_{\FD}$, every strong sampling of $t$ is also in
$T_{\FD}$. Any AFD that satisfies closure under strong sampling is said to
be a \emph{strong-sampling} AFD.

Although the set of strong-sampling AFDs are a strict subset of all AFDs, we conjecture that restricting our discussion to strong sampling AFDs does not weaken our result. Specifically, we assert without proof that for any AFD $D$, we can construct an ``equivalent'' strong-sampling AFD $D'$. This notion of equivalence is formally discussed in  Section \ref{subsec:solvingFDproblemsWithAnother}.


\subsection{The Leader Election Oracle.} 
\label{subset:omegaDef}

An example of a strong-sampling AFD is the leader election oracle
$\Omega$~\cite{chan:twfdf}.
Informally speaking, $\Omega$ continually outputs a location ID at
each live location; eventually and permanently, $\Omega$ outputs the
ID of a unique live location at all the live locations.
The $\Omega$ failure detector was shown in \cite{chan:twfdf} to be a
``weakest'' failure detector to solve crash-tolerant consensus, in a
certain sense. 
We will present a version of this proof in this paper.

We specify our version  of $\Omega = (\hat{I},O_\Omega,T_\Omega)$
as follows.
The action set $O_\Omega = \cup_{i \in \Pi}O_{{\Omega},i}$, where, for
each $i \in \Pi$, $O_{\Omega,i} = \{ FD\text{-}\Omega(j)_i|j \in
\Pi\}$. 
$T_\Omega$ is the set of all valid sequences $t$ over $\hat{I} \cup
O_\Omega$ that satisfy the following property:  
if $live(t) \neq \emptyset$, then there exists a location $l \in
live(t)$ and a suffix $t_{suff}$ of $t$ such that
$t_{suff}|_{O_\Omega}$ is a sequence over the set $\{FD\text{-}\Omega(l)_i|i \in live(t)\}$.

\begin{algorithm}\footnotesize
\caption{Automaton that implements the $\Omega$ AFD}
\label{alg:OmegaAutomaton}
The automaton $FD\text{-}\Omega$

\textbf{Signature:}

\tab input $crash_i$, $i \in \Pi$ 

\tab output $FD\text{-}\Omega(j)_i$, $i,j \in \Pi$

\textbf{State variables:}

\tab $crashset$, a subset of $\Pi$, initially $\emptyset$

\textbf{Transitions:}

\tab input $crash_i$

\tab effect

\tab \tab $crashset := crashset \cup \set{i}$

\tab

\tab output $FD\text{-}\Omega(j)_i$

\tab precondition

\tab \tab $(i\notin crashset) \wedge (j = \min (\Pi \setminus crashset))$

\tab effect

\tab \tab none

\textbf{Tasks:}

\tab One task per location $i \in \Pi$ defined as follows:

\tab \tab $\set{FD\text{-}\Omega(j)_i|j \in \Pi}$

\end{algorithm}
 
Algorithm \ref{alg:OmegaAutomaton} shows an automaton whose set of
fair traces is a subset of $T_{\Omega}$; it follows that $\Omega$
satisfies our formal definition of a ``problem''.
It is easy to see that $\Omega = (\hat{I}, O_\Omega,T_\Omega)$
satisfies all the properties of an AFD, and furthermore, note that
$\Omega$ also satisfies closure under strong sampling. 
The proofs of these observations are left as an exercise.

\paragraph{AFD $\Omega_f$.}  Here, we introduce $\Omega_f$, where $f\leq n$ is a natural number, as a generalization of $\Omega$. In this paper, we will show that $\Omega_f$ is a weakest strong-sampling AFD that solves fault-tolerant consensus if at most $f$ locations are faulty.
Informally speaking, $\Omega_f$ denotes the AFD
that behaves exactly like $\Omega$ in traces that have at most $f$
faulty locations. Thus,  $\Omega_{n}$ is the AFD $\Omega$.

Precisely, $\Omega_f = (\hat{I},O_\Omega, T_{\Omega_f})$, where
$T_{\Omega_f}$ is the set of all valid sequences $t$ over $\hat{I}
\cup O_{\Omega}$ such that, if $|faulty(t)| \leq f$, then $t \in
T_{\Omega}$.
This definition implies that $T_{\Omega_f}$ contains all the valid
sequences over $\hat{I} \cup O_{\Omega}$  such that $|faulty(t)|> f$.

It is easy to see that $\Omega_f$ is a strong-sampling AFD.

\section{System Model and Definitions}
\label{sec:systemModel}

We model an asynchronous system as the composition of a collection of
I/O automata of the following kinds:  process automata, channel
automata, a crash automaton, and an environment automaton.
%
%
The external signature of each automaton and the interaction among
them are described in Section \ref{subsec:systemInteraction}. 
The behavior of these automata is described in Sections
\ref{subsec:processAutomata}---\ref{subsec:environmentAutomaton}. 

For the definitions that follow, we assume an alphabet $\mathcal{M}$
of messages.

\subsection{System Structure}
\label{subsec:systemInteraction}

A system contains a collection of process automata, one for each
location in $\Pi$.
We define the association with a mapping $Proc$, which maps each
location $i$ to a process automaton $Proc_i$.
Automaton $Proc_i$ has the following external signature. 
It has an input action $crash_i$, which is an output from the crash
automaton,  
a set of output actions $\set{send(m,j)_i| m \in \mathcal{M} \wedge j
  \in \Pi \setminus \set{i}}$, 
and a set of input actions $\set{receive(m,j)_i| m \in
  \mathcal{M} \wedge j \in \Pi \setminus \set{i}}$.  
A process automaton may also have other external actions with which it
interacts with the external environment or a failure detector; the set
of such actions may vary from one system to another. 

For every ordered pair $(i,j)$ of distinct locations, the system
contains a channel automaton $C_{i,j}$, which models the channel that
transports messages from process $Proc_i$ to process $Proc_j$. 
Channel $C_{i,j}$ has the following external actions. 
The set of input actions $input(C_{i,j})$ is $\set{send(m,j)_i| m \in
  \mathcal{M}}$, which is a subset of outputs of the process automaton
$Proc_i$.
The set of output actions $output(C_{i,j})$ is $\set{receive(m,i)_j| m
  \in \mathcal{M}}$, which is a subset of inputs to $Proc_j$.

The crash automaton $\mathcal{C}$ models the occurrence of crash
failures in the system.
Automaton $\mathcal{C}$ has $\hat{I} = \set{crash_i | i \in \Pi}$ as
its set of output actions, and no input actions.

The environment automaton $\mathcal{E}$ models the external world with
which the distributed system interacts. The automaton $\mathcal{E}$ is a composition of $n$ automata $\set{\mathcal{E}_i | i \in\Pi}$.
For each location $i$, the set of input actions to automaton $\mathcal{E}_i$ includes the action $crash_i$.
In addition, $\mathcal{E}_i$ may have input and output actions corresponding
(respectively) to any outputs and inputs of the process automaton $Proc_i$ that
do not match up with other automata in the system.

We assume that, for every location $i$, every external action of
$Proc_i$ and $\mathcal{E}_i$, respectively, occurs at $i$, that is, $loc(a) = i$ for every external
action $a$ of $Proc_i$ and $\mathcal{E}_i$.

We provide some constraints on the structure of the various automata
below.

\subsection{Process Automata}
\label{subsec:processAutomata}

The process automaton at location $i$, $Proc_i$, is an I/O automaton
whose external signature satisfies the constraints given above, and
that satisfies the following additional properties.
\begin{enumerate}
\item
Every internal action of $Proc_i$ occurs at $i$, that is,
$loc(a) = i$ for every internal action $a$ of $Proc_i$.  
We have already assumed that every external action of $Proc_i$ occurs
at $i$; now we are simply extending this requirement to the internal
actions.


\item
Automaton $Proc_i$ is deterministic, as defined in Section~\ref{sec:
  IOA:  automata-defs}.

\item
When $crash_i$ occurs, it permanently disables all locally controlled
actions of $Proc_i$.  
\end{enumerate}

We define a \emph{distributed algorithm} $A$ to be a collection of
process automata, one at each location; formally, it is simply a
particular $Proc$ mapping.  
For convenience, we will usually write $A_i$ for the process automaton
$Proc_i$.

\subsection{Channel Automata}
\label{subsec:channelAutomata}

The channel automaton for $i$ and $j$, $C_{i,j}$, is an I/O automaton
whose external signature is as described above.
That is, $C_{i,j}$'s input actions are $ \set{   send(m,j)_i | m \in
  \mathcal{M}}$ and its output actions are $\set{receive(m,i)_j | m
  \in \mathcal{M}}$.

Now we require $C_{i,j}$ to be a specific I/O automaton---a
\emph{reliable FIFO channel}, as defined in~\cite{lync:da}.
This automaton has no internal actions, and all its output actions are
grouped into a single task.
The state consists of a FIFO queue of messages, which is initially empty.
A $send$ input event can occur at any time. 
The effect of an event $send(m,j)_i$ is to add $m$ to the end of the
queue. 
When a message $m$ is at the head of the queue, the output action
$receive(m,i)_j$ is enabled, and the effect is to remove $m$ from the
head of the queue.
Note that this automaton $C_{i,j}$ is deterministic.

\subsection{Crash Automaton}
\label{subsec:crashAutomaton}

The crash automaton $\mathcal{C}$ is an I/O automaton with
$\hat{I} = \set{crash_i | i \in \Pi}$ as its set of output actions,
and no input actions.

Now we require the following constraint on the behavior of
$\mathcal{C}$:  Every sequence over $\hat{I}$ is a fair trace of the
crash automaton.
That is, any pattern of crashes is possible.
For some of our results, we will consider restrictions on the number of locations that crash.

\subsection{Environment Automaton}
\label{subsec:environmentAutomaton}
The environment automaton $\mathcal{E}$ is an I/O automaton whose
external signature satisfies the constraints described in Section \ref{subsec:systemInteraction}. Recall that $\mathcal{E}$ is a composition of $n$ automata $\set{\mathcal{E}_i | i \in\Pi}$. For each location $i$, the following is true.
\begin{enumerate}
\item $\mathcal{E}_i$ has a unique initial state.
\item $\mathcal{E}_i$ has tasks $Env_{i,x}$, where $x$ ranges
over some fixed task index set $X_i$.
\item $\mathcal{E}_i$ is task-deterministic.
\item When $crash_i$ occurs, it permanently disables all locally controlled
actions of $\mathcal{E}_i$.
\end{enumerate}
In addition, in some specific cases we will require the traces of
$\mathcal{E}$ to satisfy certain ``well-formedness'' restrictions,
which will vary from one system to another.
We will define these specifically when they are needed, later in the
paper.

\begin{figure}[htpb]
	\centering
\includegraphics[scale=0.35,page=1]{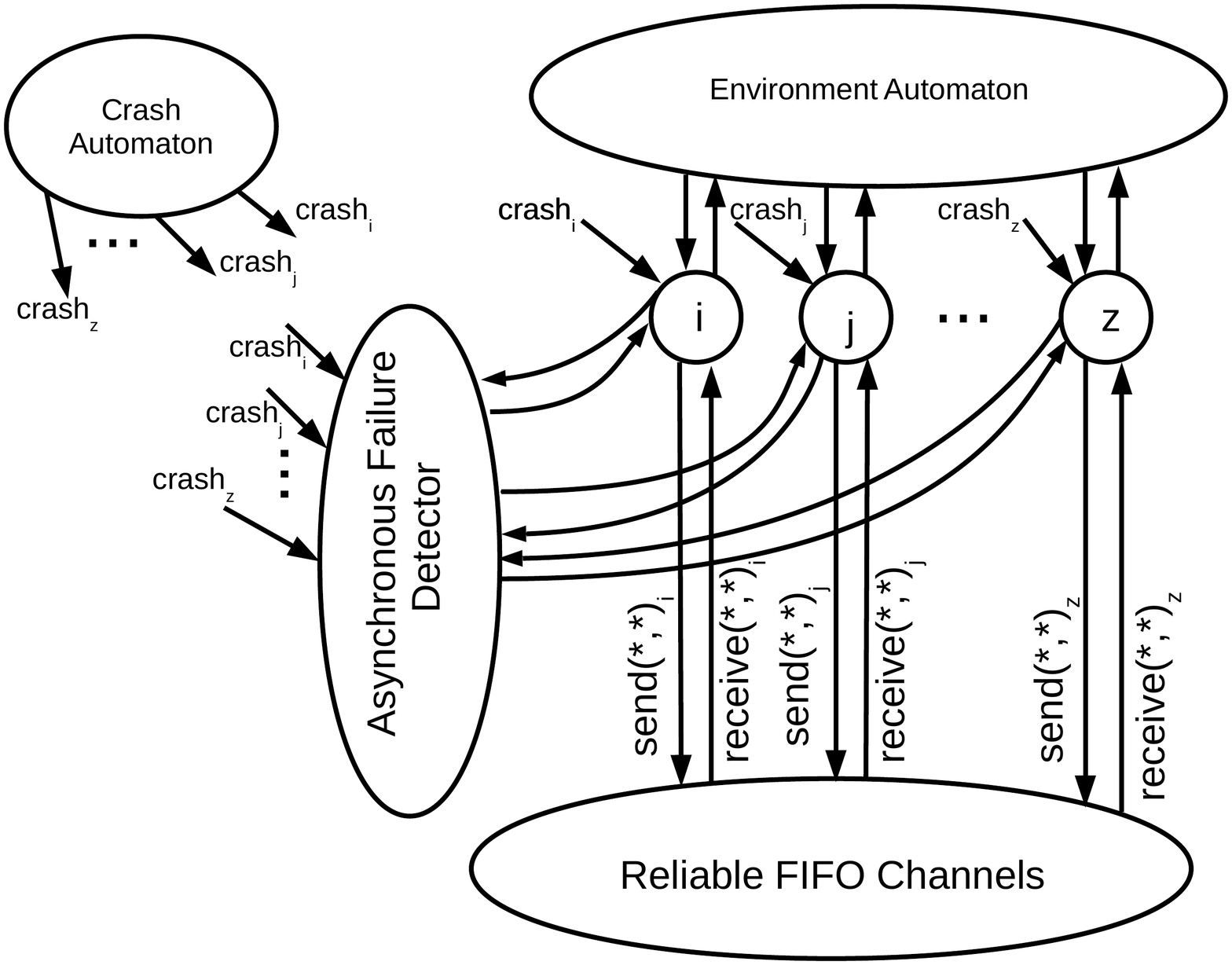}
	\caption{Interaction diagram for a message-passing asynchronous distributed system augmented with a failure detector automaton.}
	\label{fig:interaction_diagram}
\end{figure}


\section{Solving Problems}
\label{sec:solvingProblems}

In this section we define what it means for a distributed algorithm
to solve a crash problem in a particular environment.
We also define what it means for a distributed algorithm to solve one
problem $P$ using another problem $P'$. 
Based on these definitions, we define what it means for an AFD to
be sufficient to solve a problem.
%

\subsection{Solving a Crash Problem}
\label{subsec:solvingCrashProblems}

An automaton $\mathcal{E}$ is said to be an \emph{environment for} $P$
if the input actions of $\mathcal{E}$ are $O_P$, and the output
actions of $\mathcal{E}$ are $I_P \setminus \hat{I}$.
Thus, the environment's inputs and outputs ``match'' those of the
problem, except that the environment doesn't provide the problem's
$crash$ inputs.

If $\mathcal{E}$ is an environment for a crash problem $P =
(I_P,O_P,T_P)$, then an I/O automaton $U$ is said to \emph{solve} $P$
in environment $\mathcal{E}$ provided that the following conditions
hold:
\begin{enumerate}
\item
$input(U) = I_P$.
\item
$output(U) = O_P$.
\item
The set of fair traces of the composition of $U$, $\mathcal{E}$, and
the crash automaton is a subset of $T_P$.
\end{enumerate}

A distributed algorithm $A$ \emph{solves} a crash problem $P$ in an
environment $\mathcal{E}$ iff the automaton $\hat{A}$, which is
obtained by composing $A$ with the channel automata, solves $P$ in
$\mathcal{E}$. 
A crash problem $P$ is said to be \emph{solvable} in an environment
$\mathcal{E}$ iff there exists a distributed algorithm $A$ such
that $A$ solves $P$ in $\mathcal{E}$. 
If crash problem $P$ is not solvable in environment $\mathcal{E}$, then
it is said to be \emph{unsolvable} in $\mathcal{E}$.

\subsection{Solving One Crash Problem Using Another}

Often, an unsolvable problem $P$ may be solvable if the system
contains an automaton that solves some \emph{other} (unsolvable) crash
problem $P'$. We describe the relationship between $P$ and $P'$ as
follows.

Let $P = (I_P,O_P,T_P)$ and $P' = (I_{P'}, O_{P'}, T_{P'})$ be two
crash problems with disjoint sets of actions (except for $crash$
actions).
%
%
Let $\mathcal{E}$ be an environment for $P$.
Then a distributed algorithm $A$ \emph{solves} crash problem $P$ using
crash problem $P'$ in environment $\mathcal{E}$ iff the following are
true:
\begin{enumerate}
\item
For each location $i \in \Pi$, $input(A_i) = \cup_{j\in \Pi \setminus \{i\}}output(C_{j,i}) \cup I_{P,i} \cup O_{P',i}$. 
\item
For each location $i \in \Pi$, $output(A_i) = \cup_{j\in \Pi \setminus
  \{i\}}input(C_{i,j}) \cup O_{P,i} \cup 
I_{P',i} \setminus \set{crash_i}$.
\item
Let $\hat{A}$ be the composition of $A$ with the channel automata,
the crash automaton, and the environment automaton $\mathcal{E}$. 
Then for every fair trace $t$ of $\hat{A}$, 
if $t|_{I_{P'} \cup O_{P'}} \in T_{P'}$, then $t|_{I_P \cup O_P} \in
T_P$. 

In effect, in any fair execution of the system, if the sequence
of events associated with the problem $P'$ is consistent with the
specified behavior of $P'$, then the sequence of events associated
with problem $P$ is consistent with the specified behavior of $P$.
\end{enumerate}

Note that requirement 3 is vacuous if for every fair trace $t$ of
$\hat{A}$, $t|_{I_{P'} \cup O_{P'}} \notin T_{P'}$. 
However, in the definition of a problem $P'$, the requirement
that there exist some automaton whose set of fair traces is a subset
of $T_{P'}$ ensures that there are ``sufficiently many'' fair traces $t$
of $\hat{A}$, such that $t|_{I_{P'} \cup O_{P'}} \in T_{P'}$.
%

We say that a crash problem $P' = (I_{P'},O_{P'},T_{P'})$ 
\emph{is sufficient to solve} a crash problem $P = (I_P.O_P,T_P)$
in environment $\mathcal{E}$, denoted $P' \stronger_\mathcal{E} P$ iff
there exists a distributed algorithm $A$ that solves $P$ using $P'$ in
$\mathcal{E}$. If $P' \stronger_\mathcal{E} P$, then also we say that
$P$ \emph{is solvable using} $P'$ in $\mathcal{E}$. If no such
distributed algorithm exists, then we state that $P$ is
\emph{unsolvable} using $P'$ in $\mathcal{E}$, and we denote it as $P'
\not\stronger_{\mathcal{E}} P$.

\subsection{Using and Solving Failure-Detector Problems}
\label{subsec:solvingFDproblemsWithAnother}

Since an AFD is simply a kind of crash problem, the definitions above
automatically yield definitions for the following notions.
\begin{enumerate}
\item 
A distributed algorithm $A$ solves an AFD $\FD$ in environment $\mathcal{E}$.
\item 
A distributed algorithm $A$ solves a crash problem $P$ using an AFD
$\FD$ in environment $\mathcal{E}$.
\item 
An AFD $\FD$ is sufficient to solve a crash problem $P$ in environment
$\mathcal{E}$.
\item 
A distributed algorithm $A$ solves an AFD $\FD$ using a crash problem
$P$ in environment $\mathcal{E}$.
\item 
A crash problem $P$ is sufficient to solve an AFD $\FD$ in environment
$\mathcal{E}$.
\item 
A distributed algorithm $A$ solves an AFD $\FD'$ using another AFD
$\FD$.
\item 
An AFD $\FD$ is sufficient to solve an AFD $\FD'$.
\end{enumerate}

Note that, when we talk about solving an AFD, the environment
$\mathcal{E}$ has no output actions because the AFD has no input
actions except for $\hat{I}$, which are inputs from the crash
automaton. 
Therefore, we have the following lemma.

\begin{lemma}
Let $P$ be a crash problem and $\FD$ an AFD.
If $P \stronger_\mathcal{E} \FD$ in some environment $\mathcal{E}$
(for $\FD$), then for any other environment $\mathcal{E}'$ for $\FD$,
$P \stronger_{\mathcal{E}'} \FD$. 
\end{lemma}

Consequently, when we refer to an AFD $D$ being solvable using a crash
problem (or an AFD) $P$, we omit the reference to the environment
automaton and simply say that $P$ is sufficient to solve $\FD$; we
denote this relationship by $P \stronger \FD$. 
Similarly, when we say that an AFD $\FD$ is unsolvable using $P$, we
omit mention of the environment, and write simply $P \not \stronger \FD$.

Finally, if an AFD $\FD$ is sufficient to solve another AFD $\FD'$ (notion 7 in the list above), then
we say that $\FD$ \emph{is stronger than} $\FD'$, and we denote this
by $\FD \stronger \FD'$. 
If $\FD \stronger \FD'$, but $\FD' \not\stronger \FD$, then we
say that $\FD$ is \emph{strictly stronger} than $\FD'$, and we denote this
by $\FD \succ \FD'$. Also, if  $D \stronger D'$ and $D' \stronger D$, then we say that $D$ is \emph{equivalent} to $D'$.

We conjecture that for any AFD $D$, there exists a strong sampling AFD $D'$ such that $D$ is equivalent to $D'$; thus, if a non-strong-sampling AFD $D$ is a weakest to solve consensus, then there must exist an equivalent AFD $D'$ that is also a weakest to solve consensus. Therefore, it is sufficient to restrict our attention to strong-sampling AFDs.


\section{Observations}\label{sec:observations}


In this section, fix $\FD$ to be an AFD.
%
%
We define the notion of an observation $G$ of $D$ and present
properties of observations.
Observations are a key part of the emulation algorithm used to prove
the ``weakest failure detector'' result, in Section~\ref{sec: wfd}.


\subsection{Definitions and Basic Properties}
\label{subsec:observationDef}

An \emph{observation} is a DAG $G = (V,Z)$, where the set $V$ of
vertices consists of triples of the form $v = (i,k,e)$ where $i \in
\Pi$ is a location, $k$ is a positive integer, and $e$ is an action
from $O_{D,i}$; we refer to $i$, $k$, and $e$ as the location, index,
and action of $v$, respectively. 
Informally, a vertex $v = (i,k,e)$ denotes that $e$ is the $k$-th AFD
output at location $i$, and the observation represents a partial
ordering of AFD outputs at various locations. 
We say that an observation $G$ is \emph{finite} iff the set $V$ (and therefore the set $Z$) is finite; otherwise, $G$ is said to be \emph{infinite}.

We require the set $V$ to satisfy the following properties.
\begin{enumerate}
\item 
For each location $i$ and each positive integer $k$, $V$ contains at
most one vertex whose location is $i$ and index is $k$.
\item 
If $V$ contains a vertex of the form $(i,k,*)$ and $k'<k$, then $V$
also contains a vertex of the form $(i,k',*)$.
\end{enumerate}
Property 1 states that at each location $i$, for each positive integer $k$, there is at most
one $k$-th AFD output. 
Property 2 states that for any $i$ and $k$, if the $k$-th AFD output occurs at $i$,
then the first $(k-1)$ AFD outputs also occur at $i$.

The set $Z$ of edges imposes a partial ordering on the occurrence of
AFD outputs.  We assume that it satisfies the following properties.
\begin{enumerate}
\setcounter{enumi}{2}
\item 
For every location $i$ and natural number $k$, if $V$ contains
vertices of the form $v_1 = (i,k,*)$ and $v_2 = (i,k+1,*)$, then $Z$
contains an edge from $v_1$ to $v_2$. 
\item 
For every pair of distinct locations $i$ and $j$ such that $V$
contains an infinite number of vertices whose location is $j$, the
following is true. 
For each vertex $v_1$ in $V$ whose location is $i$, there is a vertex
$v_2$ in $V$ whose location is $j$ such that there is an edge from
$v_1$ to $v_2$ in $Z$.
\item 
For every triple $v_1$, $v_2$, $v_3$ of vertices such that $Z$
contains both an edge from $v_1$ to $v_2$ and an edge from $v_2$ to
$v_3$, $Z$ also contains an edge from $v_1$ to $v_3$. 
That is, the set of edges of $G$ is closed under transitivity.
\end{enumerate}
Property 3 states that at each location $i$, the $k$-th output at $i$ occurs before the
$(k+1)$-st output at $i$. 
Property 4 states that for every pair of locations $i$ and $j$ such that infinitely many AFD outputs occur at $j$, for every AFD output event $e$ at $i$ there exists some AFD output event $e'$ at $j$ such that $e$ occurs before
$e'$. 
Property 5 is a transitive closure property that simply captures the
notion that if event $e_1$ happens before event $e_2$ and $e_2$
happens before event $e_3$, then $e_1$ happens before $e_3$.

Given an observation $G = (V,Z)$, if $V$ contains an infinite number
of vertices of the form $(i,*,*)$ for some particular $i$, then $i$ is
said to be \emph{live} in $G$. 
We write $live(G)$ for the set of all the locations that are live in $G$.

\begin{lemma}
\label{prop:liveImpliesAllIndices}
Let $G = (V,Z)$ be an observation, $i$ a location in $live(G)$.
Then for every positive integer $k$, $V$ contains exactly one vertex
of the form $(i,k,*)$.
\end{lemma}
\begin{proof}
Follows from Properties 1 and 2 of observations.
\end{proof}

\begin{lemma}
\label{prop:outgoingEdgesToLiveVertices}
Let $i$ and $j$ be distinct locations with $j \in live(G)$.
Let $v$ be a vertex in $V$ whose location is $i$.
Then there exists a positive integer $k$ such that for every positive
integer $k' \geq k$, $Z$ contains an edge from $v$ to some vertex of
the form $(j,k',*)$.
\end{lemma}

\begin{proof}
Follows from Lemma \ref{prop:liveImpliesAllIndices}, and
Properties 3, 4, and 5 of observations.
\end{proof}
%


\begin{lemma}
\label{prop:outgoingEdgesToCrashedVertices}
Let $i$ and $j$ be distinct locations with $j \in live(G)$ and $i
\notin live(G)$; that is, $V$ contains infinitely many vertices whose
location is $j$ and only finitely many vertices whose location is $i$.
Then there exists a positive integer $k$ such that for every $k' \geq
k$, there is no edge from any vertex of the form $(j,k',*)$ to any
vertex whose location is $i$.
\end{lemma}

\begin{proof}
Fix $i$ and $j$ as in the hypotheses.
Let $v_1$ be the vertex in $V$ whose location is $i$ and whose index
is the highest among all the vertices whose location is $i$.
From Lemma \ref{prop:outgoingEdgesToLiveVertices} we know that
there exists a positive integer $k$ such that for every positive
integer $k' \geq k$, $Z$ contains an edge from $v_1$ to some vertex of the
form $(j,k',*)$. 
Since $G$ is a DAG, there is no edge from any vertex of the form
$(j,k',*)$, $k' \geq k$ to $v_1$. 
Applying Properties 3 and 5 of observations, we conclude that
there is no edge from any vertex of the form $(j,k',*)$ to any vertex
whose location is $i$.
\end{proof}

\begin{lemma}
\label{prop:finiteIncomingEdges}
Let $G = (V,Z)$ be an observation.
Every vertex $v$ in $V$ has only finitely many incoming edges in $Z$.
\end{lemma}

\begin{proof}
For contradiction, assume that there exists a vertex $v$ with
infinitely many incoming edges, and let $i$ be the location of $v$.
Then there must be a location $j \neq i$ such that there are
infinitely many vertices whose location is $j$ that have an outgoing
edge to $v$.   
Fix such a location $j$. 
Note that $j$ must be live in $G$.

Since there are infinitely many vertices whose location is $j$, by
Property 4 of observations, we know that $v$ has an outgoing edge to
some vertex $(j,k,*)$. 
Since infinitely many vertices of the form
$(j,k',*)$ have an outgoing edge to $v$, fix some such $k' > k$. 
By Properties 3 and 5 of observations, we know that there exists a
edge from $(j,k,*)$ to $(j,k',*)$. 
Thus, we see that there exist edges from $v$ to $(j,k,*)$, 
from $(j,k,*)$ to $(j,k',*)$, and from $(j,k',*)$ to $v$, which yield a
cycle. 
This contradicts the assumption that $G$ is a DAG. 
\end{proof}

\subsection{Viable Observations}

Now consider an observation $G = (V,Z)$. 
If $\mathcal{V}$ is any sequence of vertices in $V$, then we define
the \emph{event-sequence} of $\mathcal{V}$ to be the sequence obtained
by projecting $\mathcal{V}$ onto its second component.

We say that a trace $t \in T_D$ is \emph{compatible} with an
observation $G$ provided that $t|O_D$ is the event sequence of some
topological ordering of the vertices of $G$.  
$G$ is a \emph{viable} observation if there exists a trace $t \in T_D$
that is compatible with $G$.

\begin{lemma}
\label{prop:liveInObsAndTrace}
Let $G$ be a viable observation, and suppose that $t \in T_D$ is
compatible with $G$. 
For each location $i$, $i$ is live in $G$ iff $i \in live(t)$.
\end{lemma}

We now consider paths in an observation DAG, and their connection with
strong sampling, as defined in Section~\ref{subsec:FDproblems}. A \emph{path} in a observation is a sequence of vertices, where for each pair of consecutive vertices $u,v$ in a path, $(u,v)$ is an edge of the observation.

A \emph{branch} of an observation $G$ is a maximal path in $G$.
A \emph{fair branch} $b$ of $G$ is a branch of $G$ that satisfies the
additional property that, for every $i$ in $\Pi$, if $i$ is live in $G$, then
$b$ contains an infinite number of vertices whose location is $i$.
%

\begin{lemma}
\label{prop:sampledSubsequence}
\label{prop:pathIsTrace}
Let $G$ be a viable observation, and suppose that $t \in T_D$ is
compatible with $G$. 
Suppose $b$ is a fair branch of $G$, and let $\epsilon$ be the event
sequence of $b$. 
Then
\begin{enumerate}
\item
There exists a strong sampling $t'$ of $t$ such that
$t'|_{O_D} = \epsilon$. 
\item
If $\FD$ is a strong-sampling AFD, then there exists $t' \in T_D$ 
such that $t'$ is a strong sampling of $t$ and $t'|_{O_D} = \epsilon$.
\end{enumerate}
\end{lemma}
\begin{proof}
Fix $G$, $t$, $b$, $\epsilon$ and $D$ from the hypotheses of the Lemma statement.
 
\emph{Proof of Part 1.}
Since $b$ is a fair branch of $G$, for each location $i$ that is live in $t$, $\epsilon$ contains an infinite number of outputs at $i$. Furthermore, for each location $i$, the projection of $\epsilon$ on the events at $i$ is a subsequence of the projection of $t$ on the AFD outputs at $i$. Therefore, by deleting all the AFD output events from $t$ that do not appear in $\epsilon$, we obtain a strong-sampling $t'$ of $t$ such that $t'|_{O_D} = \epsilon$.

\emph{Proof of Part 2.}
In Part 2, assume $\FD$ is a strong-sampling AFD. From Part 1, we have already established that there exists a strong-sampling $t'$ of $t$ such that $t'|_{O_D} = \epsilon$. Fix such a $t'$. By closure under strong-sampling, since $t \in T_D$, we conclude that $t' \in T_D$ as well.
\end{proof}

Lemma \ref{prop:sampledSubsequence} is crucial to our results. In Section \ref{sec: wfd}, we describe an emulation algorithm that uses outputs from an AFD to produce viable observations, and the emulations consider paths of the observation and simulate executions of a consensus algorithm with AFD outputs from each path in the observation. Lemma \ref{prop:sampledSubsequence} guarantees that each fair path in the observation corresponds to an actual sequence of AFD outputs from some trace of the AFD. In fact, the motivation for closure-under-strong-sampling property is to establish Lemma \ref{prop:sampledSubsequence}. 

\subsection{Relations and Operations on Observations}\label{subsec:OperationsOnObservations}

The emulation construction in Section~\ref{sec: wfd} will require
processes to manipulate observations.  To help with this, we define
some relations and operations on DAGs and observations. 

\paragraph{Prefix.}
Given two DAGs $G$ and $G'$, $G'$ is said to be a
\emph{prefix} of $G$ iff $G'$ is a subgraph of $G$ and for every
vertex $v$ of $G'$, the set of incoming edges of $v$ in $G'$ is equal
to the set of incoming edges of $v$ in $G$.

\paragraph{Union.}
Let $G = (V,Z)$ and $G' = (V',Z')$ be two observations. 
Then the \emph{union} $G''$ of $G$ and $G'$, denoted $G \cup G'$, is
the graph $(V \cup V', Z \cup Z')$.
Note that, in general, this union need not be another observation.
%
%
%
However, under certain conditions, wherein the observations are finite and ``consistent'' in terms of the vertices and incoming edges at each vertex, the union of two observations is also an observation. We state this formally in the following Lemma.

\begin{lemma}
\label{prop:obs:union}
Let $G = (V,Z)$ and $G' = (V',Z')$ be two finite observations. 
Suppose that the following hold:
\begin{enumerate}
\item
There do not exist $(i,k,e) \in V$ and $(i,k,e') \in V'$
with $e \neq e'$.  
\item
If $v \in V \cap V'$ then $v$ has the same set of incoming edges (from
the same set of other vertices) in $G$ and $G'$.
\end{enumerate}
Then $G \cup G'$ is also an observation. 
\end{lemma}

\begin{proof}
Straightforward.
\end{proof}



\paragraph{Insertion.}
Let $G = (V,Z)$ be a finite observation, $i$ a location, and $k$ the
largest integer such that $V$ contains a vertex of the form $(i,k,*)$.
Let $v$ be a triple $(i,k+1,e)$.
Then $insert(G,v)$, the result of inserting $v$ into $G$, is a new graph 
$G' = (V',Z')$, where $V' = V \cup \set{v}$ and 
$Z' = Z \cup \set{(v',v) | v' \in V}$. 
That is, $G'$ is obtained from $G$ by adding vertex $v$ and adding
edges from every vertex in $V$ to $v$.


\begin{lemma}
\label{prop:insertVertexYeildsObservation}
Let $G = (V,Z)$ be a finite observation, $i$ a location.
Let $k$ be the largest integer such that $V$ contains a vertex of the
form $(i,k,*)$.
Let $v$ be a triple $(i,k+1,e)$.
Then $insert(G,v)$ is a finite observation.
\end{lemma}

\subsection{Limits of Sequences of Observations}

Consider an infinite sequence $G_1 = (V_1,Z_1), G_2 = (V_2,Z_2),
\ldots$ of finite observations,
where each is a prefix of the next.
Then the \emph{limit} of this sequence is the graph $G^\infty
= (V,Z)$ defined as follows:
\begin{itemize}
\item
$V = \bigcup_y V_y$.
\item
$Z = \bigcup_y Z_y$.
\end{itemize}

\begin{lemma}\label{lem:GxIsPrefixOfGInfty}
For each positive integer $y$, 
$G_y$ is a prefix of $G^\infty$.
\end{lemma}

%
Under certain conditions, the limit of the infinite sequence of observations $G_1,G_2,\ldots$ is also an observation; we note this in Lemma  \ref{lem:obs:limitObs}.
\begin{lemma}
\label{lem:obs:limitObs}
Let $G^\infty = (V,Z)$ be the limit of the infinite sequence $G_1 =
(V_1,Z_1), G_2 = (V_2,Z_2), \ldots$ of finite observations,
where each is a prefix of the next.
Suppose that the sequence satisfies the following property:
\begin{enumerate}
\item
For every vertex $v \in V$ and any location $j \in live(G^\infty)$,
there exists a vertex $v' \in V$ with location
$j$ such that  $Z$ contains the edge $(v,v')$.
\end{enumerate}
Then $G^\infty$ is an observation.
\end{lemma}

\begin{proof}
All properties are straightforward from the definitions, except for
Property 4 of observations, which follows from the assumption of the
lemma.
\end{proof}

We define an infinite sequence $G_1 = (V_1,Z_1), G_2 = (V_2,Z_2),
\ldots$ of finite observations, where each is a prefix of the next, to be to be \emph{convergent} if the limit $G^\infty$ of this sequence is an observation.

\section{Execution Trees}
\label{subsec:treeOfExec}

In this section, we define a tree representing executions of a system
$\mathcal{S}$ that are consistent with a particular observation $G$ of a
particular failure detector $\FD$.
Specifically, we define a tree that describes executions of
$\mathcal{S}$ in which the sequence of AFD outputs is exactly the event-sequence
of some path in observation $G$.

Section \ref{subsec:TreeExecutionSystem} defines the system $\mathcal{S}$ for which the tree is defined.
The tree is constructed in two parts: Section \ref{subsec:taskTree} defines a ``task tree'', and Section \ref{subsec:TheAugmentedTree} adds tags to the nodes and edges of the task tree to yield the final execution tree. Additionally, Sections \ref{subsec:taskTree}  and \ref{subsec:TheAugmentedTree} prove certain basic properties of execution trees, and they establish a correspondence between the nodes in the tree and finite executions of $\mathcal{S}$.
 Section \ref{subsec:taggedTreeProps} defines that two nodes in the execution tree are ``similar'' to each other if they have the same tags, and therefore correspond to the same execution of $\mathcal{S}$; the section goes on to prove certain useful properties of nodes in the subtrees rooted at any two similar nodes. 
 Section \ref{subsec:PropertiesOfSimilarModuloNodes} defines that two nodes in the execution tree are ``similar-modulo-$i$'' to each other if the executions corresponding to the two nodes are indistinguishable for process automata at any location except possibly the the process automaton at $i$; the section goes on to prove certain useful properties of nodes in the subtrees rooted at any two similar-modulo-$i$ nodes. 
 Section \ref{subset:PropertiesOfTreesFromPrefixObservations} establishes useful properties of nodes that are in different execution trees that are constructed using two observations, one of which is a prefix of another. Finally, Section \ref{subset:FairBranchesOfExecutionTrees} proves that a ``fair branch" of infinite execution trees corresponds to a fair execution of system $\mathcal{S}$. The major results in this section are used in Sections \ref{sec:consensusAndAFD} and \ref{sec: wfd}, which show that $\Omega_f$ is a weakest strong-sampling AFD to solve consensus if at most $f$ locations crash.

\subsection{The System}\label{subsec:TreeExecutionSystem}

Fix $\mathcal{S}$ to be a system consisting of a distributed algorithm $A$, channel
automata, and an environment automaton $\mathcal{E}$ such that $A$ solves
a crash problem $P$ using $D$ in $\mathcal{E}$. 
%

The system $\mathcal{S}$ contains the following tasks.
The process automaton at $i$ contains a single task $Proc_i$. 
Each channel automaton $Chan_{i,j}$, where $j \in \Pi \setminus \{i\}$
contains a single task, which we also denote as $Chan_{i,j}$;  the actions in task $Chan_{i,j}$ are of the form $receive(*,i)_j$, which results in a message received at location $j$. 
Each automaton $\mathcal{E}_i$ has tasks $Env_{i,x}$, where $x$ ranges
over some fixed task index set $X_i$.
Let $T$ denote the set of all the tasks of $\mathcal{S}$.

Each task has an associated location, which is the location of all the
actions in the task.
The tasks at location $i$ are $Proc_i$, $Chan_{j,i}| j \in \Pi \setminus \{i\}$, 
and $Env_{i,x} |x \in X_i$.

Recall from Section \ref{sec:systemModel} that each process automaton, each channel automaton, and the environment automaton have unique initial states. Therefore, the system $\mathcal{S}$ has a unique initial state. From the definitions of the constituent automata of $\mathcal{S}$, we obtain the following lemma.

\begin{lemma}\label{lem:addCrashAfterLastEvent}
 Let $\alpha$ be an execution of system $\mathcal{S}$, and let $t = t_{pre} \cdot t_{suff}$ be the trace of $\alpha$ such that for some location $i$, $t_{suff}$ does not contain any locally-controlled actions at $Proc_i$ and $\mathcal{E}_i$. Then, there exists an execution $\alpha'$ of system $\mathcal{S}$ such that $t' = t_{pre} \cdot crash_i \cdot t_{suff}$ is the trace of $\alpha'$. 
\end{lemma}
\begin{proof}
Fix $\alpha$, $t= t_{pre} \cdot t_{suff}$ and $i$ as in the hypothesis of the claim.
Let $\alpha_{pre}$ be the prefix of $\alpha$ whose trace is $t_{pre}$. Let $s$ be the final state of $\alpha_{pre}$.
Let $\alpha'_{pre}$ be the execution $\alpha_{pre} \cdot crash_i \cdot s'$, where $s'$ is the state of $\mathcal{S}$ when $crash_i$ is applied to state $s$.

Note that $crash_i$ disables all locally-controlled actions at $Proc_i$ and $\mathcal{E}_i$, and it does not change the state of any other automaton in $\mathcal{S}$. Therefore, the state of all automata in $\mathcal{S}$ except for $Proc_i$ and $Env_i$ are the same in state $s$ and $s'$. Also, note that $t_{suff}$ does not contain any locally-controlled action at $Proc_i$ or $Env_i$, and $t_{suff}$ can be applied to state $s$. Therefore, $t_{suff}$ can also be applied to $s'$, thus extending $\alpha'_{pre}$ to an execution $\alpha'$ of $\mathcal{S}$. By construction, the trace $t'$ of $\alpha'$ is $t_{pre} \cdot crash_i \cdot t_{suff}$.
\end{proof}


\subsection{The Task Tree}\label{subsec:taskTree}
For any observation $G = (V,Z)$, we define a tree $\mathcal{R}^G$ that describes all executions of $\mathcal{S}$
in which the sequence of AFD output events is the event-sequence of
some path in $G$.


We describe our construction in two stages.
The first stage, in this subsection, defines the basic structure of
the tree, with annotations indicating where particular system tasks and
observation vertices occur.  
The second stage, described in the next subsection, adds information
about particular actions and system states.

The task tree is rooted at a special node called ``$\top$'' which
corresponds to the initial state of the system $\mathcal{S}$.
The tree is of height $|V|$; if $|V|$ is infinite, the tree has infinite height.\footnote{The intuitive reason for limiting the depth of the tree to $|V|$ is the following.
If $G$ is a finite observation, then none of the locations in $\Pi$ are live in $G$. In this case, we want all the branches in the task tree to be finite.
	On the other hand, if $G$ is an infinite observation, then some location in $\Pi$ is live in $G$, and in this case we want all the branches in the task tree to be infinite.
	On way to ensure these properties is to restrict the depth of the tree to $|V|$.} 
Every node $N$ in the tree that is at a depth $|V|$ is a leaf
node. All other nodes are internal nodes. 
Each edge in the tree is labeled by an element from $T \cup \{FD_{i} |
i \in \Pi\}$. 
Intuitively, the label of an edge corresponds to a task being given a
``turn'' or an AFD event occurring. 
An edge with label $l$ is said to be an $l$-edge, for short. 
The child of a node $N$ that is connected to $N$ by an edge labeled
$l$ is said to be an $l$-child of $N$.

In addition to labels at each edge, the tree is also augmented with
a \emph{vertex tag}, which is a vertex in $G$, at each node and edge. 
We write $v_N$ for the vertex tag at node $N$ and $v_E$ for the vertex
tag at edge $E$. Intuitively, each vertex tag denotes the latest
AFD output that occurs in the execution of $\mathcal{S}$ corresponding to 
the path in the tree from the root to node $N$ or the head node of edge $E$ (as appropriate).
The set of outgoing edges from each node $N$ in the tree is determined by the vertex tag $v_N$. 

We describe the labels and vertex tags in the task tree recursively,
starting with the $\top$ node. 
We define the vertex tag of $\top$ to be a special placeholder element
$(\bot,0,\bot)$, representing a ``null vertex'' of $G$.
For each internal node $N$ with vertex tag $v_N$, the outgoing edges
from $N$ and their vertex tags are as follows.
\begin{itemize}
\item
\emph{Outgoing $Proc$, $Chan$, and $Env$ edges.}
For every task $l$ in $T$, the task tree contains exactly one outgoing
edge $E$ from $N$ with label $l$ from $N$, \textit{i.e.}, an $l$-edge. 
The vertex tag $v_E$ of $E$ is $v_N$. 
\item
\emph{Outgoing $FD$-edges.} 
If $v_N = (\bot,0,\bot)$, then for every vertex $(i,k,e)$ of $G$, the task tree includes an edge $E$ from $N$ with label $FD_i$ and vertex tag $v_E = (i,k,e)$. 
For every location $i$ such that $G$ contains no vertices with location $i$, the task tree includes a single outgoing edge $E$ from $N$ with label $FD_i$ and vertex tag $(\bot,0,\bot)$.

Otherwise, ($v_N$ is a vertex of $G$) for every vertex (say) $(i,k,e)$ of $G$ that has an edge in $G$ from
vertex $v_N$, the task tree includes an outgoing edge $E$ from $N$ with
label $FD_i$ and vertex tag $v_E = (i,k,e)$. 
For every location $i$ such that there is no edge in $G$ from $v_N$ to
any vertex whose location is $i$, the task tree includes an outgoing edge
$E$ from $N$ with label $FD_i$ and vertex tag $v_E = v_N$.
\end{itemize}
For each node $\hat{N}$ that is a child of $N$ and whose incoming edge is
$E$, $v_{\hat{N}} = v_E$.

A \emph{path} in a rooted tree is an alternating sequence of nodes and edges,
  beginning and ending with a node, where (1) each node is incident to
  both the edge that precedes it and the edge that follows it in the
  sequence, and (2) the nodes that precede and follow an edge are the
  end nodes of that edge.
  
A \emph{branch} in a rooted tree is a maximal path in the tree that starts at the root.

The following two Lemmas follow from the construction of the task tree.
\begin{lemma}\label{prop:allLabelsExist}
For each label $l$, each internal node $N$ in $\mathcal{R}^{G}$ has at
least one outgoing $l$-edge.
\end{lemma}

\begin{lemma}
Let $q$ be a path in the tree that begins at the root node.
Let $\mathcal{V}$ be the sequence of distinct non-$(\bot,0,\bot)$ vertex tags of edges in path $q$.
Then there exists some path $p$ in $G$ such that $\mathcal{V}$ is the sequence of vertices along $p$.
\end{lemma}

\subsection{The Augmented Tree}\label{subsec:TheAugmentedTree}


Now we augment the task tree produced in the previous section to
include additional tags --- configuration tags $c_N$ at the nodes, which are states of
the system $\mathcal{S}$, and action tags $a_E$ at the edges, which are actions of $\mathcal{S}$
or $\bot$. However, the action tags cannot be $crash$ actions.
The resulting tagged tree is our execution tree $\mathcal{R}^G$. Intuitively, the configuration tag $c_N$ of a node $N$ denotes a state of system $\mathcal{S}$, and the action tag $a_E$ for an edge $E$ with label $l$ from node $N$ denotes an action $a_E$ from task $l$ that occurs when system $\mathcal{S}$ is in state $c_N$. It is easy to see that for any path in the execution tree, the sequence of alternating configuration tags and action tags along the path represents an execution fragment of $\mathcal{S}$.

We define the tags recursively, this time starting from the
already-defined task tree.
For the $\top$ node, the configuration tag is the initial state of
$\mathcal{S}$.
For each internal node $N$ with configuration tag $c_N$ and vertex tag $v_N$,
the new tags are defined as follows:

\begin{itemize}
\item
\emph{Outgoing $FD$-edges.} 
For every edge $E$ from node $N$ with label $FD_i$, the action tag
$a_E$ is determined as follows.
If the vertex tag $v_E = (i,k,e) \neq v_N$, then $a_E = e$.
If $v_E = v_N$, then $a_E = \bot$.

Essentially, if $v_E = (i,k,e) \neq v_N$, then 
this corresponds to the action $e$ of $v_E$ occurring when $\mathcal{S}$ is in state $c_N$; we model this by setting $a_E$ to $e$. Otherwise, $v_E = v_N$ and no event from $FD_i$ occurs  when $\mathcal{S}$ is in state $c_N$; we mode this by setting $a_E$ to $\bot$. 

\item
\emph{Outgoing $Proc$ and $Env$ edges.} 
For every edge $E$ from node $N$ with label $l \in \set{Proc_i} \cup \set{Env_{i,x} |x \in X_i}$ for some
location $i$, the action tag $a_E$ is determined as follows.
If (1) some action $a$ in task $l$ is enabled in state $c_N$,  and (2) either 
(a) $v_N$ is a vertex of $G$ and $G$ contains an edge from $v_N$ to a vertex with location $i$, or (b) $v_N = (\bot, 0, \bot)$ and $G$ has a vertex with location $i$, then
$a_E$ is $a$; otherwise $a_E$ is $\bot$. 
%
%
Note that since each process automaton and each constituent automaton of the environment automaton in $\mathcal{S}$ is task-deterministic, for each location $i$ at most one
action in the $Proc_i$ task is enabled in $c_N$ and, for each location $i$ and each $x \in X_i$, at most one action in the $Env_{i,x}$ task is enabled in $c_N$. Therefore, at most one action $a$ in task $l$ is enabled in state $c_N$, and thus $a_E$ is well-defined.

Fix node $N$ in $\mathcal{R}^G$ and a location $i$. Observe that if the action tag of an $FD_i$ edge from $N$ is $\bot$, then for all $FD_i$ edges that are descendants of $N$, their action tag is $\bot$. The condition (2) above for determining $a_E$ for a $Proc_i$ or $Env_i$ edge $E$ from $N$ implies that, if no AFD output events at $i$ follow $N$ in the maximal subtree of $\mathcal{R}^G$ rooted at $N$, then no $Proc_i$ event of $Env_i$ event follows $N$ in that subtree either; we formalize this claim is Lemma \ref{lem:crashedLocationNoActionsInSubtree}. 

\item
\emph{Outgoing $Chan$ edges.}
For every edge $E$ from node $N$ with label $l \in \set{Chan_{i,j} | i \in \Pi \wedge j \in \Pi \setminus \set{i}}$, the action tag $a_E$ is determined as
follows:
If some action $a$ in task $l$ is enabled in state $c_N$, then $a_E =
a$; otherwise $a_E = \bot$. 
Note that since all automata in $\mathcal{S}$ are task deterministic, at most
one action in task $l$ is enabled in $c_N$. Informally, we state that if some action in task $l$ is enabled in state $c_N$, then that event occurs along the edge $E$; otherwise, no event occurs along the edge $E$.
\end{itemize}

Each node $\hat{N}$ that is a child of $N$ and whose incoming edge is
$E$ is tagged as follows. 
If the action tag $a_E = \bot$ then $c_{\hat{N}} = c_N$. 
Otherwise, $c_{\hat{N}}$ is the state of $\mathcal{S}$ resulting from applying the
action $a_E$ to state $c_N$. 




The following Lemmas establish various relationships between nodes,
paths, and branches in
$\mathcal{R}^G$. Note that these Lemmas following immediately
from the construction.


 For each node $N$, let $path(N)$ be the path from the root node $\top$
 to $N$ in the tree $\mathcal{R}^{G}$. 
Let $exe(N)$ be the sequence of alternating config tags and action
tags along $path(N)$ such that $exe(N)$ contains exactly the
non-$\bot$ action tags and their preceding config tags in $path(N)$ and
ends with the config tag $c_N$. 

\begin{lemma}\label{prop:nodeFiniteExe}
 For each node $N$ in $\mathcal{R}^{G}$, the sequence $exe(N)$ is a finite execution of the system $\mathcal{S}$ that ends in state $c_N$ and if $exe(N)|_{O_D}$ is non-empty (and therefore, $v_N$ is a vertex of $G$), then $exe(N)|_{O_D}$ is the event-sequence of the vertices in $G$ for some path to $v_N$.
\end{lemma}

\begin{lemma}\label{prop:botChild}\label{prop:nonBotChild}
Let $N$ be a node, let $\hat{N}$ be a child of $N$, and let $E$ be the edge from $N$ to $\hat{N}$ in $\mathcal{R}^G$. Then the following are true.
 \begin{enumerate}
 \item If $a_E = \bot$, then $c_N = c_{\hat{N}}$, $exe(N) = exe(\hat{N})$ and $v_N = v_{\hat{N}}$.
 \item If $a_E \neq \bot$, then $exe(\hat{N}) = exe(N) \cdot a_E \cdot c_{\hat{N}}$.
 \end{enumerate}
\end{lemma}

\begin{lemma}\label{prop:ancesterPrefix}\label{prop:ancesterPrefixAFDEvents}
 For each node $N$ in $\mathcal{R}^{G}$ and any descendant $\hat{N}$ of $N$, $exe(N)$ is a prefix of $exe(\hat{N})$ and $exe(N)|_{O_D}$ is a prefix of $exe(\hat{N})|_{O_D}$.
\end{lemma}
\begin{proof}
 Follows from repeated application of Lemmas \ref{prop:botChild} along the path from $N$ to $\hat{N}$.
\end{proof}

\begin{lemma}\label{lem:childNodeUniqueByLabelAndVertexTag}
For each node $N$ in $\mathcal{R}^G$, each child node $\hat{N}$ of $N$ is uniquely determined by the label $l$ of the edge from $N$ to $\hat{N}$ and the vertex tag $v_{\hat{N}}$.
\end{lemma}
\begin{proof}
The proof follows from the construction of $\mathcal{R}^G$. Fix $N$. If two outgoing edges $E_1$ and $E_2$ from $N$ have the same label, then that label must be from $\set{FD_i | i\in \Pi}$. However, for each location $i$, each of the outgoing $FD_i$-edges from $N$ have a different vertex tag, and the vertex tag of an  $FD_i$-child $\hat{N}$ of $N$ is the same as the vertex tag of the edge from $N$ to $\hat{N}$. Hence, for any no two child nodes of $N$, either the label of the edge from $N$ to each of the child nodes is distinct, or the vertex tag of each of the child nodes is distinct. 
\end{proof} 

\begin{lemma}\label{prop:FDchildVertexEdgeExists}
For each node $N$ in $\mathcal{R}^G$ and any child $\hat{N}$ of $N$ such that the edge $E$ from $N$ to
$\hat{N}$ has the label $FD_i$ (for some location $i$) and the action tag $a_E$ of the edge is non-$\bot$, the following
is true. (1) $v_{\hat{N}} \neq v_N$, (2) $a_E$ is the action of $v_{\hat{N}}$, and (3) if $v_N \neq (\bot, 0, \bot)$, then there is an edge
from $v_N$ to $v_{\hat{N}}$ in $G$.
\end{lemma}

\begin{lemma}\label{prop:sameVertexTag}
For each node $N$ in $\mathcal{R}^{G}$ and any descendant $\hat{N}$ of $N$ such that there is no $FD$-edge in the path from $N$ to $\hat{N}$, $v_N = v_{\hat{N}}$.
\end{lemma}
\begin{proof}
The proof is by induction on the length of the path from $N$ to $\hat{N}$.
\end{proof}

\begin{lemma}\label{lem:descendantVertexTagEdge}
For each node $N$ in $\mathcal{R}^{G}$ and for any descendant $\hat{N}$ of $N$, either $v_{N} = v_{\hat{N}}$, or if $v_N \neq (\bot, 0, \bot)$, then there is an edge from $v_N$ to $v_{\hat{N}}$ in $G$.
\end{lemma}
\begin{proof}
Fix $N$ and $\hat{N}$ as in the hypothesis of the lemma. Let the path from $N$ to $\hat{N}$ contain $d$ edges. We prove the lemma by strong induction on $d$.

\emph{Inductive hypothesis.} For any pair of nodes $N_1$ and $N_2$ such that $N_2$ is a descendant of $N_1$ in $\mathcal{R}^G$, and the path from $N_1$ to $N_2$ contains $d$ edges,  either $v_{N_1} = v_{N_2}$, or if $v_{N_1} \neq (\bot, 0, \bot)$, then  there is an edge from $v_{N_1}$ to $v_{N_2}$ in $G$.

\emph{Inductive step.} Fix $N_1$ and $N_2$. If $d=0$, note that $N_1 = N_2$, and therefore $v_{N_1} = v_{N_2}$; therefore the lemma is satisfied. For $d=1$, $N_2$ is a child of $N_1$, and let $E_{1,2}$ be the edge from $N_1$ to $N_2$. If $v_{N_1} = v_{N_2}$, the the lemmas is satisfied. Assume $v_{N_1} \neq v_{N_2}$ and $v_{N_1} \neq (\bot, 0, \bot)$; note that if $v_{N_1} \neq v_{N_2}$, then by construction $E_{1,2}$ is an FD-edge and $a_{E_{1,2}} \neq \bot$. Invoking Lemma \ref{prop:FDchildVertexEdgeExists}, we know that there is an edge from $v_{N_1}$ to $v_{N_2}$ in $G$.

For any $d > 1$, there exists at least one node $N_{1.5}$ in the path from $N_1$ to $N_2$. Fix $N_{1.5}$. By construction, the path from $N_1$ to $N_{1.5}$ contains fewer than $d$ edges, and the path from $N_{1.5}$ to $N_2$ contains fewer than $d$ edges. Invoking the inductive hypothesis for nodes $N_1$ and $N_{1.5}$, we know that either $v_{N_1} = v_{N_{1.5}}$ or, if $v_{N_1}  \neq (\bot, 0, \bot)$, then there is an edge from $v_{N_1}$ to $v_{N_{1.5}}$ in $G$. Similarly, invoking the inductive hypothesis for nodes $N_{1.5}$ and $N_2$, we know that either $v_{N_{1.5}} = v_{N_2}$ or, if $v_{N_{1.5}}  \neq (\bot, 0, \bot)$,  there is an edge from $v_{N_{1.5}}$ to $v_{N_2}$ in $G$. Therefore, either (1) $v_{N_1} = v_{N_2}$, or (2) if $v_{N_1} \neq (\bot, 0, \bot)$, then $v_{N_{1.5}} \neq (\bot, 0, \bot)$, and there is a path from  $v_{N_1}$ to $v_{N_2}$ in $G$. In case (1) the induction is complete. In case (2), invoking the transitive closure property of $G$, we know that there is an edge from $v_{N_1}$ to $v_{N_2}$ in $G$, and the induction is complete.

\end{proof}

\begin{lemma}\label{prop:liveLocationNonBotAction}
For each label $FD_i$ where $i$ is live in $G$, every $FD_i$-edge in $\mathcal{R}^{G}$ has a non-$\bot$ action tag.
\end{lemma}


\begin{lemma}\label{lem:fairBranchUnfairExe}
For every branch $b$ of $\mathcal{R}^{G}$, $exe(b)$ is an execution of system $\mathcal{S}$.
\end{lemma}
\begin{proof}
Fix a branch $b$ of $\mathcal{R}^{G}$. Let $\top,E_1,N_1,E_2,N_2,\ldots$, where each $E_x$ is an edge in $\mathcal{R}^G$ and each $N_x$ is a node in $\mathcal{R}^G$,
denote the sequence of nodes that constitute $b$. By definition,
$exe(b)$ is the limit of the prefix-ordered sequence
$exe(\top)$, $exe(N_1)$, $exe(N_2)$, $\ldots$; note that this sequence might be infinite. Note that
$exe(\top)|_{O_D}$ is a prefix of
$exe(N_1)|_{O_D}$, and from Lemma
\ref{prop:ancesterPrefixAFDEvents}, we know that  $exe(N_x)|_{O_D}$ is
a prefix of $exe(N_{x+1})|_{O_D}$ for any positive integer
$x$. Therefore, the limit of the prefix-ordered sequence
$exe(\top)|_{O_D}$, $exe(N_1)|{O_D}$, $exe(N_2)|_{O_D}$, $\ldots$
exists, and this limit is $exe(b)|_{O_D}$. By Lemma
\ref{prop:nodeFiniteExe}, we know that $exe(\top)$ and each
$exe(N_x)$, where $x$ is a positive integer, is a finite execution of
$\mathcal{S}$, and therefore, $exe(b)$ is an execution of $\mathcal{S}$. 
\end{proof}
\begin{lemma}\label{lem:crashedLocationNoOutgoingActions}
For any node $N$ in $\mathcal{R}^G$, any location $i$ and any $FD_i$-edge $E$ outgoing from $N$, if $a_E = \bot$, then for each outgoing $Proc_i$-edge or $Env_i$-edge $E'$ from $N$, $a_{E'} = \bot$.
\end{lemma}
\begin{proof}
Fix $N$, $i$, and $E$ as in the hypothesis of the Lemma; thus, $a_E = \bot$. From the construction of $\mathcal{R}^G$, we know that $a_{E} = \bot$ iff either $v_N$ is not a vertex in $G$ and there  is no vertex in $G$ whose location is $i$, or there is no edge in $G$ from $v_{N}$ to any vertex whose location is $i$.

Fix $E'$ to be either a $Proc_i$-edge or $Env_i$-edge outgoing from $N$. From the construction of  $\mathcal{R}^G$, we know that if  either 
(a) $v_N$ is a vertex of $G$ and $G$ contains no edges from $v_N$ to a vertex with location $i$, or (b) $v_N = (\bot, 0, \bot)$ and $G$ has a no vertex with location $i$, then $a_{E'}$ is $\bot$. 
\end{proof}

For any node $N$ in $\mathcal{R}^G$,  let $\mathcal{R}^G|_N$ denote the maximal subtree of $\mathcal{R}^G$ rooted at $N$.

\begin{lemma}\label{lem:crashedLocationNoActionsInSubtree}
For any node $N$ in $\mathcal{R}^G$, any location $i$ and any $FD_i$-edge $E$ outgoing from $N$, if $a_E = \bot$, then for each $Proc_i$-edge or $Env_i$-edge $E'$ in  $\mathcal{R}^G|_N$,  $a_{E'} = \bot$.
\end{lemma}
\begin{proof}
Fix $N$ and $i$ as in the hypothesis of the Lemma; thus, an outgoing $FD_i$-edge $E$ from $N$, $a_E = \bot$. From the construction of $\mathcal{R}^G$, we know that $a_{E} = \bot$ iff either $v_N$ is not a vertex in $G$ and there is no vertex in $G$ whose location is $i$, or there is no edge in $G$ from $v_{N}$ to any vertex whose location is $i$.

Fix $N'$ to be any node in $\mathcal{R}^G|_N$. By construction, $N'$ is a descendant of $N$. From the construction of $\mathcal{R}^G$, note that for every descendant $N'$ of $N$ in $\mathcal{R}^G$, if $v_{N}$ is a vertex in $G$, then $v_{N'}$ is a descendant of $v_N$. Thus, either there is no vertex in $G$ whose location is $i$, or $v_{N'}$ does not have any outgoing edges to a vertex in $G$ whose location is $i$. From the construction of $\mathcal{R}^G$, we see that  $a_{E''} = \bot$ for an outgoing $FD_i$-edge $E''$ from $N'$. From Lemma \ref{lem:crashedLocationNoOutgoingActions}, we know that  for each outgoing $Proc_i$-edge or $Env_i$-edge $E'$ from $N$, $a_{E'} = \bot$.
\end{proof}

Next, we establish the relationship between traces compatible with $G$ and the action tags of $FD$-edges in $\mathcal{R}^G$. Specifically, we show that
the following is true. 
For any node $N$ in $\mathcal{R}^G$ such that the vertex tag $v_N$ is a vertex in $G$, let $a$ be the event of $v_N$, and assume that some $FD_i$-edge of $N$ has a non-$\bot$ action tag. Then  in any trace $t$ compatible with $G$, and for any location $i$, no $crash_i$ event precedes $a$ in $t$.

\begin{lemma}\label{lem:noCrashIUntilBot}
 Let $N$ be any node in $\mathcal{R}^G$ such that $G$ contains $v_N$.
Let there exist an $FD_i$-edge $E$ in $\mathcal{R}^G|_N$ such that $a_E \neq \bot$. Then for any arbitrary trace in $t \in T_D$ that is compatible with $G$, no $crash_i$ event precedes the event of $v_N$ in $t$.
\end{lemma}
\begin{proof}
Fix $N$, $i$, and $E$ as in the hypotheses of the lemma. Let $N'$ denote the upper endpoint of $E$.
Since $N'$ is in $\mathcal{R}^G|_N$, $N'$ is a descendant of $N$, and consequently, there exists a path from $v_N$ to $v_{N'}$ in $G$. Since $a_E \neq \bot$, we know that $v_{N'}$ has an outgoing edge to some vertex $v$ in $G$, fix $v$; note that $a_E$ is the event of $v$. Since we have a path from $v_N$ to $v_N'$ in $G$ and an edge from $v_N'$ to $v$ in $G$, we have a path from $v_N$ to $v$ in $G$.
Therefore, in every topological sort of $G$, $v$ follows $v_N$.

Now consider $t$, and assume for contradiction that $crash_i$ precedes the event of $v_N$ in $t$.
Since $t|_{O_D}$ is a topological sort of $G$, the event $a_E$ of $v$ follows the event of $v_N$ in $t$.
Then $crash_i$ precedes the event $a_E$ of $v$ in $t$. Recall that $E$ is an $FD_i$ edge and therefore $a_E \in O_{D,i}$. In other words, $crash_i$ precedes an $O_{D,i}$ event in $t$; thus, $t$ is not a valid sequence. This contradicts our assumption that $t$ is a trace in $T_D$, because all traces in $T_D$ are valid.
\end{proof}

We define a \emph{non-$\bot$} node. A node $N$ in $\mathcal{R}^G$ is said to be a \emph{non-$\bot$} node iff the path from the root to $N$ does not contain any edges whose action tag is $\bot$. In the subsequent sections, non-$\bot$ nodes play a significant role, and so we prove some useful properties about non-$\bot$ nodes next.

\begin{lemma}\label{lem:nonBotNodeUniqueExe}
Suppose $N$ and $N'$ are a non-$\bot$ nodes in $\mathcal{R}^G$ such that (1) $N$ and $N'$ are at the same depth $d$,  (2) the projection of the paths from $\top$ to $N$ and $\top$ to $N'$ on the set of labels are equal, (3) the projection of  the paths from $\top$ to $N$ and $\top$ to $N'$ on the vertex tags are also equal. Then $N=N'$.
\end{lemma}
\begin{proof}
The proof is a straightforward induction on $d$.
\end{proof}

The inductive extension of Lemma \ref{lem:childNodeUniqueByLabelAndVertexTag} is that each non-$\bot$ node $N$ in $\mathcal{R}^G$ is uniquely determined by the sequence of labels and vertex tags of the edges from $\top$ to $N$. We prove this next.

\begin{lemma}\label{lem:LabelsAndVertexAgsDenoteUniqueNode}
Each non-$\bot$ node $N$ in $\mathcal{R}^G$ is uniquely determined by the sequence of labels and vertex tags of the edges from $\top$ to $N$.
\end{lemma}
\begin{proof}
The proof is by induction the depth $d_N$ of $N$.

\emph{Base case.} $d_\top = 0$, and there is unique $\top$ node in $\mathcal{R}^G$.

\emph{Inductive Hypothesis.} For some positive integer $d$, each non-$\bot$ node $N$ in $\mathcal{R}^G$ at depth $d$ is uniquely determined by the sequence of labels and vertex tags of the edges from $\top$ to $N$.

\emph{Inductive step.} Fix $N'$ to be any non-$\bot$ node in $\mathcal{R}^G$ at depth $d+1$. By construction, there is an edge whose lower endpoint is $N'$ and whose upper end point is a node $N''$ at depth $d$. By Lemma \ref{lem:childNodeUniqueByLabelAndVertexTag}, we know that given $N''$, $N'$ is uniquely determined by the label $l$ of the edge from $N'$ to $N''$ and the vertex tag $v_{N''}$. However, by the inductive hypothesis, $N''$ is uniquely determined by the sequence of labels and vertex tags of the edges from $\top$ to $N''$. Therefore, $N'$ is uniquely determined in $\mathcal{R}^G$ by the sequence of labels and vertex tags of the edges from $\top$ to $N'$. This completes the induction.
\end{proof}



\subsection{Properties of ``Similar'' Nodes in Execution Trees}
\label{subsec:taggedTreeProps}

For any two nodes $N$ and $N'$ in $\mathcal{R}^{G}$ such that $c_N =
c_{N'}$ and $v_N = v_{N'}$, the following lemmas establish a
relationship between the descendants of $N$ and $N'$. Informally,
these lemmas establish that the maximal subtrees of $\mathcal{R}^G$
rooted at $N$ and $N'$ are in some sense similar to each other.
Lemma \ref{lem:sameConfigSameChild} establishes that for every child
$\hat{N}$ of $N$ there exists a child $\hat{N'}$ of $N'$ that is
``similar'' to $\hat{N}$. Lemma \ref{lem:sameConfigSameExtensionLength}
extends such similarity to arbitrary descendants of $N$; that is,
for any descendant $\hat{N}$ of $N$, there exist ``similar'' descendants
of $N'$. Lemma \ref{lem:everyDescendantNoBot} states that for any descendant $\hat{N}$ of $N$, there exists a descendant $\hat{N}_{\not\bot}$ of $N$ that is ``similar'' to $\hat{N}$, but the path from $N$ to $\hat{N}_{\not\bot}$ does not contain any edges with a $\bot$ action tag.

The proofs use
the notion of ``distance'' between a node and its descendant as
defined next. The \emph{distance} from a node $N$ to its descendant
$\hat{N}$ is the number of edges in the path from $N$ to
$\hat{N}$. Note that if the distance from $N$ to $\hat{N}$ is $1$,
then $\hat{N}$ is a child of $N$.

%


\begin{lemma}\label{lem:sameConfigSameChild}
 Let $N$ and $N'$ be two nodes in $\mathcal{R}^{G}$ such that $c_N =
 c_{N'}$ and $v_N = v_{N'}$. Let $l$ be an arbitrary label in
 $T \cup \set{FD_i | i \in \Pi}$. Let  $\hat{E}$ and $\hat{N}$ be an $l$-edge and the
 corresponding $l$-child of $N$, respectively. There exists an
 $l$-edge $E^{\prime}$ of $N'$ and the corresponding $l$-child
 $\hat{N^{\prime}}$ of $N'$ such that $a_{\hat{E}} = a_{\hat{E^{\prime}}}$, $v_{\hat{E}} = v_{\hat{E'}}$,
 $c_{\hat{N}} = c_{\hat{N^{\prime}}}$, and $v_{\hat{N}} = v_{\hat{N^{\prime}}}$. 
\end{lemma}

\begin{proof}
 Fix $N$, $N'$, $l$, $\hat{E}$, and $\hat{N}$ as in the hypotheses of the lemma. We consider two cases: $l$ is in $T$, and $l$ is in $\set{FD_i | i \in \Pi}$. 

 \emph{Case 1.} $l \in T$. Since $c_N = c_{N'}$, $v_N = v_{N'}$, and the system is task deterministic, we know that there exists an outgoing $l$ edge $\hat{E^\prime}$ from $N'$ such that $a_{\hat{E}} = a_{\hat{E^\prime}}$. Let $\hat{N^\prime}$ be the $l$-child of $N'$ connected by edge $\hat{E^\prime}$. Since $c_{\hat{N}}$ is obtained by applying $a_{\hat{E}}$ to $c_N$, and $c_{\hat{N^\prime}}$ is obtained by applying $a_{\hat{E^\prime}}$ to $c_{N'}$, we see that $c_{\hat{N}} = c_{\hat{N^\prime}}$. Also, by construction, $v_{\hat{N}} = v_N = v_{\hat{E}}$ and $v_{\hat{N^\prime}} = v_{N'} = v_{\hat{E'}}$; therefore, $v_{\hat{E}} = v_{\hat{E'}}$ and $v_{\hat{N}} = v_{\hat{N^\prime}}$.

 \emph{Case 2.} $l$ is of the form $FD_i$, for some particular $i$. Then we consider two subcases: (a) $a_{\hat{E}} = \bot$ and (b) $a_{\hat{E}} \neq \bot$.
 
 \emph{Subcase 2(a).} $a_{\hat{E}} = \bot$. Then either (i) $v_N = (\bot, 0, \bot)$ and $G$ has no vertices with location $i$, or (ii) $v_N$ is a vertex of $G$ and $G$ has no vertices with location $i$ to which $v_N$ has an outgoing edge. In both cases (i) and (ii), by construction, $v_{\hat{N}} = v_N = v_{\hat{E}}$. Since $v_N = v_{N'}$, from the construction of $\mathcal{R}^G$, we know that there is an $l$-edge $\hat{E^\prime}$ of $N'$ such that $a_{\hat{E^\prime}} = \bot$, and we also know that for the $l$-child $\hat{N^\prime}$ of $N'$ that is connected to $N'$ by $\hat{E^\prime}$, $v_{\hat{N^\prime}} = v_{N'} = v_{\hat{E'}}$. Therefore, $v_{\hat{E}} = v_{\hat{E'}}$ and $v_{\hat{N}} = v_{\hat{N^\prime}}$.

\emph{Subcase 2(b).} $a_{\hat{E}} \neq \bot$. Then either (i) $v_N = (\bot, 0, \bot)$ and $G$ has a vertex $v'$ of the form $(i,*,a_{\hat{E}})$, or (ii) $v_N$ is a vertex of $G$ and $G$ has a vertex $v'$ of the form $(i,*,a_{\hat{E}})$ to which $v_N$ has an outgoing edge such that $v_{\hat{N}} = v_{\hat{E}} = v'$. Since $v_N = v_{N'}$, in both cases (i) and (ii), from the construction of $\mathcal{R}^G$, we know that there is an $l$-edge $\hat{E^\prime}$ of $N'$ such that $a_{\hat{E^\prime}} = a_{\hat{E}}$ and $v_{\hat{E'}} = v_{\hat{E}}$, and we also know that for the $l$-child $\hat{N^\prime}$ of $N'$ that is connected to $N'$ by $\hat{E^\prime}$, $v_{\hat{N^\prime}} = v'$. Therefore, $v_{\hat{E}} = v_{\hat{E'}}$  and $v_{\hat{N}} = v_{\hat{N^\prime}}$.

In both subcases, since $c_{\hat{N}}$ is obtained by applying $a_{\hat{E}}$ to $c_N$, and $c_{\hat{N^\prime}}$ is obtained by applying $a_{\hat{E^\prime}}$ to $c_{N'}$, we see that $c_{\hat{N}} = c_{\hat{N^\prime}}$.  
 \end{proof}

 \begin{lemma}\label{lem:sameConfigSameExtensionLength}
 Let $N$ and $N'$ be two nodes in $\mathcal{R}^{G}$ such that $c_N = c_{N'}$ and $v_N = v_{N'}$, and let $\hat{N}$ be a descendant of $N$. There exists a descendant $\widehat{N'}$ of $N'$ such that the following is true.
 \begin{enumerate}
 \item $v_{\hat{N}} = v_{\widehat{N'}}$.
 \item $c_{\hat{N}} = c_{\widehat{N'}}$. 
 \item Let the path from $N$ to $\hat{N}$ be $p$ and the path from $N'$ to $\widehat{N'}$ be $p'$. Then, $p$ and $p'$ are of the same length.
  \item The suffix of $exe(\hat{N})$ following $exe(N)$ is identical to the suffix of $exe(\widehat{N'})$ following $exe(N')$.
  \end{enumerate}
\end{lemma}
\begin{proof}
The lemma is a simple inductive extension of Lemma \ref{lem:sameConfigSameChild}. The proof follows from a straightforward induction on the length of the path from $N$ to $\hat{N}$.
\end{proof}

Next, we show that for any node $N$ and any descendant $\hat{N}$ of $N$, there exists a node $\hat{N}_{\not\bot}$ of $N$ that is ``similar'' to $\hat{N}$, and the path from $N$ to $\hat{N}_{\not\bot}$ does not contain any edges with a $\bot$ action tag.

\begin{lemma}\label{lem:everyDescendantNoBot}
 Let $N$ be an arbitrary node in $\mathcal{R}^{G}$. For every descendant $\hat{N}$ of $N$, there exists a descendant $\hat{N}_{\not\bot}$ of $N$ such that  $v_{\hat{N}} = v_{\hat{N}_{\not\bot}}$, the suffix of $exe(\hat{N})$ following $exe(N)$ is identical to the suffix of $exe(\hat{N}_{\not\bot})$ following $exe(N)$, and the path  from $N$ to $\hat{N}_{\not\bot}$ does not contain any edges whose action tag is $\bot$.
\end{lemma}
\begin{proof}
Fix $N$ and $\hat{N}$ as in the hypothesis of the lemma. Let $p$ be the path from $N$ to $\hat{N}$. If $p$ does not contain any edges whose action tag is $\bot$, then the lemma is satisfied when $\hat{N}_{\not\bot} = \hat{N}$. Otherwise, the following arguments hold.

Let $\alpha$ be the suffix of $exe(\hat{N})$ following $exe(N)$, starting with the state $c_N$. Let $\alpha_t$ denote the trace of $\alpha$, and let $\alpha_l$ denote the sequence of tasks in $\mathcal{S}$ such that for each $x$, $\alpha_t[x]$ is an action in task $\alpha_l[x]$. By construction, there exists a path from $N$ whose projection on the labels is $\alpha_l$, and furthermore, since $\alpha_t$ is the trace of $\alpha$, and the starting state of $\alpha$ is $c_N$, there exists path $\hat{p}$ from $N$ whose projection on action tags is $\alpha_t$; fix such a path $\hat{p}$. Note that, by construction, $\alpha_t$ does not contain any $\bot$ elements. Thus, path $\hat{p}$ has no edges with $\bot$ action tag, and the suffix of $exe(\hat{N})$ following $exe(N)$ is identical to the suffix of $exe(\hat{N}_{\not\bot})$ following $exe(N)$.
\end{proof}

\begin{corollary}\label{cor:NonBotNodesExist}
For each node $N$ in $\mathcal{R}^{G}$, there exists a non-$\bot$ node $N'$ in $\mathcal{R}^{G}$ such that $exe(N) = exe(N')$, $v_N = v_{N'}$.
\end{corollary}
\begin{proof}
 Follows by applying Lemma \ref{lem:everyDescendantNoBot} to the root node and noting that $N$ is a descendant of the root node.
\end{proof}


\subsection{Properties of Similar-Modulo-$i$ Nodes in Execution Trees}\label{subsec:PropertiesOfSimilarModuloNodes}
Next, we establish properties of $\mathcal{R}^G$ with respect to nodes
whose configuration tags and vertex tags are indistinguishable at all process automata
except one. The aforementioned relation between nodes is formalized as
the \emph{similar-modulo-$i$} relation (where $i$ is a location).
Intuitively, we say that node $N$ is \emph{similar-modulo-$i$ to} $N'$ if
the only process automaton that can distinguish state $c_N$ from state
$c_{N'}$ is the process automaton at $i$. 
The formal definition follows.

Given two nodes $N$ and $N'$ in $\mathcal{R}^{G}$ and a location $i$,
$N$ is said to be \emph{similar-modulo-$i$ to} $N'$ (denoted $N \sim_i N'$)
if the following are true. 
\begin{enumerate}
\item $v_N = v_{N'}$.
\item For every location $j \in \Pi \setminus \set{i}$, the state of
       $Proc_j$ is the same in $c_N$ and $c_{N'}$.
\item For every location $j \in \Pi \setminus \set{i}$, the state of $\mathcal{E}_j$ is the same in $c_N$ and $c_{N'}$. 
\item For every pair of distinct locations $j,k \in \Pi \setminus \set{i}$,
      the state of $Chan_{j,k}$ is the same in $c_N$ and $c_{N'}$. 
\item For every location $j \in \Pi \setminus \set{i}$, the contents of
      the queue in $Chan_{i,j}$ in state $c_N$ is a prefix of the contents
      of the queue in $Chan_{i,j}$ in state $c_{N'}$.
\end{enumerate}

 Note that due to property 5, the $\sim_i$ relation is not symmetric; that is, $N \sim_i N'$ does not imply $N' \sim_i N$. However, the relation is reflexive; that is, $N \sim_i N$ for any node $N$.

Also note that if $N' \sim_i N$, then the states of $Proc_i$, $\mathcal{E}_i$, and $Chan_{j,i}$ for all $j \neq i$ may be different in $c_{N'}$ and $c_{N}$. Furthermore, the states of $Chan_{i,j}$ for all $j \neq i$ may also be different in $c_{N'}$ and $c_{N}$, but it is required that the messages in transit from $i$ to $j$ in state $c_N$ form a prefix of the messages in transit from $i$ to $j$ in state $c_{N'}$.

We define a node $N$ to be a \emph{post-$crash_i$} node, where $i$ is a location,  if the following property is satisfied. If $v_N = (\bot, 0 ,\bot)$, then there are no vertices in $G$ whose location is $i$. Otherwise, there are no outgoing edges in $G$ from $v_N$ to any vertex whose location is $i$. Note that if $\mathcal{R}^G$ contains any post-$crash_i$ node, then $i$ is not live in $G$. Furthermore, if a node $N$ in $\mathcal{R}^G$ is a post-$crash_i$ node, and there exists a node $N'$ such that $N \sim_i N'$, then $N'$ is also a post-$crash_i$ node.
 
 \begin{lemma}\label{lem:similarModuloChild}
 Let $N$ and $N'$ be two post-$crash_i$ nodes in $\mathcal{R}^{G}$ for some location $i$ in $\Pi$, such that $N \sim_i N'$. Let $l$ be any label, and let
 $N^l$ be an $l$-child of $N$. Then, one of the following is true: (1) $N^l \sim_i N'$, or (2) there exists an $l$-child $N^{\prime l}$ of $N'$ such that $N^l$ and $N^{\prime l}$ are post-$crash_i$ nodes and $N^l \sim_i N^{\prime l}$.
\end{lemma}

\begin{proof}
 Fix $N$, $N'$, $i$, $l$, and $N^l$ as in the hypotheses of the lemma. Let $E$ be the $l$-edge from $N$ to $N^l$, and let $a_E$ be the action tag of $E$. 
 
 If $a_E = \bot$, then by Lemma \ref{prop:botChild}, we know that $c_N = c_{N^l}$ and $v_N = v_{N^l}$. Therefore, $N^l \sim_i N'$, and the lemma is satisfied. For the remainder of this proof, we assume $a_E \neq \bot$.
 
 Note that label $l$ is an element of $\set{Proc_i} \cup \set{Env_{i,x} | x \in X_i}  \cup \set{FD_i} \cup  \set{Proc_j | j \in \Pi \setminus \set{i}} \cup \set{Env_{j,x} | j \in \Pi \setminus \set{i} \wedge x \in X_j} \cup \set{FD_j | j \in \Pi \setminus \set{i}} \cup \set{Chan_{j,k} | j \in \Pi \wedge k \in \Pi \setminus \set{j}}$.
  
  \emph{Case 1.} $l \in \set{Proc_i} \cup \set{Env_{i,x} | x \in X_i}$. From the definition of a post-$crash_i$ node, we know that there are no vertices with location $i$ that have an incoming edge from $v_N$ ($=v_{N'}$). Therefore, from the construction of $\mathcal{R}^G$, we see that $a_E = a_{E'} = \bot$. In this case, we have already established that $N^l \sim_i N^\prime$.
  
  \emph{Case 2.} $l = FD_i$. We know that there are no vertices with location $i$ that have an incoming edge from $v_N$, and therefore, $a_E = a_{E'} = \bot$. In this case, we have already established that $N^l \sim_i N^\prime$.
  
For the remainder of the cases, let $N^{\prime l}$ be the $l$-child of $N'$ connected to $N'$ by edge $E'$. Note that since $v_N = v_{N^\prime}$, we know that $v_{N^l} = v_{N^{\prime l}}$.

 \emph{Case 3.} $l  \in \set{Proc_j | j \in \Pi \setminus \set{i}}  \cup \set{Env_{j,x} | j \in \Pi \setminus \set{i} \wedge x \in X_j}$. From the definition of the $\sim_i$ relation, we know that the state of $Proc_j$  is the same in states $c_N$ and $c_{N'}$, and similarly, the state of $\mathcal{E}_j$ is the same in states $c_N$ and $c_{N'}$. Therefore, $a_E = a_{E'}$. Consequently, the state of $Proc_j$  is the same in $c_{N^l}$ and $c_{N^{\prime l}}$, and the state of $\mathcal{E}_j$ is the same in $c_{N^l}$ and $c_{N^{\prime l}}$. 
 
 Also, from the definition of the $\sim_i$ relation, we know that for every location $k \in \Pi \setminus \set{i,j}$, the state of $Chan_{j,k}$ is the same in $c_N$ and $c_{N'}$. Therefore, from state $c_N$, if $a_E$ changes the state of $Chan_{j,k}$ for some $k \neq i$, then we know that the state of $Chan_{j,k}$ is the same in $c_{N^l}$ and $c_{N^{\prime l}}$. 
 
 Thus, the states of all other automata in $\mathcal{S}$ are unchanged. We have already established that $v_{N^l} = v_{N^{\prime l}}$, and we can verify that $N^l \sim_i N^{\prime l}$.
 

\emph{Case 4.} $l \in \set{FD_j | j \in \Pi \setminus \set{i}}$. Since $v_N = v_{N'}$, we see that $a_E = a_{E'}$. Applying $a_E$ to $c_N$ and applying $a_{E'}$ to $c_{N'}$, and recalling that we have already established $v_{N^l} = v_{N^{\prime l}}$, we can verify that $N^l \sim_i N^{\prime l}$.

 \emph{Case 5.} Let $l$ be $Chan_{j,k}$ where $j\in \Pi$ and $k \in \Pi \setminus \set{j}$. Recall that we have already established $v_{N^l} = v_{N^{\prime l}}$. 
 We consider three subcases: (a) $k=i$, (b) $j\neq i$ and $k \neq i$, (c) $j = i$.
 
 \emph{Case 5(a).} Let $l$ be $Chan_{j,i}$ where $j \in \Pi \setminus \set{i}$. Since the definition of $\sim_i$ does not restrict the state of $Chan_{j,i}$ or the state of the process automaton at $i$, we see that $N^l \sim_i N^{\prime l}$.
 
 \emph{Case 5(b).} Let $l$ be $Chan_{j,k}$ where $j\in \Pi \setminus \set{i}$ and $k \in \Pi \setminus \set{i,j}$. From the definition of the $\sim_i$ relation, we know that the state of $Chan_{j,k}$ is the same in $c_N$ and $c_{N'}$. Therefore, $a_E = a_{E'}$.
 
 Thus, we see that the state of $Chan_{j,k}$ is the same in $c_{N^l}$ and $c_{N^{\prime l}}$. Similarly, since $N \sim_i N'$ and $a_E = a_{E'}$, we see that the state of the process automaton at $k$ is also the same in $c_{N^l}$ and $c_{N^{\prime l}}$. The states of all other automata in $\mathcal{S}$ are unchanged.
 Thus, we can verify that $N^l \sim_i N^{\prime l}$.
 
 \emph{Case 5(c).} Let $l$ be $Chan_{i,k}$ where $k \in \Pi \setminus \set{i}$. 
 Since we have assumed $a_E \neq \bot$, $a_E$ must be the action $receive(m,i)_k$ for some message $m \in \mathcal{M}$. From the definition of the $\sim_i$ relation, we know that the queue of messages in $Chan_{i,k}$ in state $c_N$ is a prefix of the queue of messages in $Chan_{i,k}$ in state $c_{N'}$, and the state of the process automaton at $k$ is also the same in $c_N$ and $c_{N'}$. Therefore, action $a_{E}$ is enabled in state $c_{N'}$ , and $a_E$ is in task $l$; therefore $a_E = a_{E'}$.
 
 Consequently, the queue of messages in $Chan_{i,k}$ in state $c_{N^l}$ is a prefix of the queue of messages in $Chan_{i,k}$ in state $c_{N^{\prime l}}$. 
 Recall that the state of the process automaton at $k$ is the same in $c_{N}$ and $c_{N'}$. 
 Therefore, the state of the process automaton at $k$ is the same in states $c_{N^l}$ and $c_{N^{\prime l}}$. The states of all other automata in $\mathcal{S}$ are unchanged.
 Thus, we can verify that $N^l \sim_i N^{\prime l}$. Furthermore, note that by construction, if a node $N_0$ is a post-$crash_i$ node, then all its descendants are post-$crash_i$ nodes. Therefore, $N^l$ and $N^{\prime l}$ are post-$crash_i$ nodes.
 \end{proof}

\begin{theorem}\label{thm:similarModuloDescendant}
 Let $N$ and $N'$ be two post-$crash_i$ nodes in $\mathcal{R}^{G}$ for some location $i$ in $\Pi$ such that $N \sim_i N'$. For every descendant $\hat{N}$ of $N$, there exists a descendant $\widehat{N'}$ of $N'$ such that $\hat{N}$ and $\widehat{N'}$ are post-$crash_i$ nodes and $\hat{N} \sim_i \widehat{N'}$.
\end{theorem}
\begin{proof}
 Fix $N$, $N'$, and $i$ as in the hypothesis of the lemma; thus, $N$ and $N'$ are post-$crash_i$ nodes and $N \sim_i N'$. The proof is by induction on the distance from $N$ to $\hat{N}$.
 
 \emph{Base Case.} 
 Let the distance from $N$ to $\hat{N}$ be $0$. That is, $N=\hat{N}$. Trivially, we see that $\widehat{N'} = N'$ satisfies the lemma.
 
 \emph{Inductive Hypothesis.} 
 For every descendant $\hat{N}$ of $N$ at a distance $k$ from $N$, there exists a descendant $\widehat{N'}$ of $N'$ such that $\hat{N}$ and $\widehat{N'}$ are post-$crash_i$ nodes and $\hat{N} \sim_i \widehat{N'}$.
 
 \emph{Inductive Step.}
 Fix $\hat{N}$ to be a descendant of $N$ at a distance $k+1$ from $N$. Let $\hat{N}_k$ be the parent of $\hat{N}$. Note that, by construction, $\hat{N}_k$ is a descendant of $N$ at a distance $k$ from $N$. Let $l$ be the label of edge $E$ that connects $\hat{N}_k$ and $\hat{N}$.
 By the inductive hypothesis, there exists a descendant $\widehat{N'}_k$ of $N'$ such that $\hat{N}_k$ and $\widehat{N'}_k$ are post-$crash_i$ nodes and $\hat{N}_k \sim_i \widehat{N'}_k$.
  Invoking Lemma \ref{lem:similarModuloChild}, we know that at least one of the following is true. (1) $\hat{N} \sim_i \widehat{N'}_k$. (2) there exists an $l$-child $\widehat{N'}$ of $\widehat{N'}_k$ such that $\hat{N}$ and $\widehat{N'}$ are post-$crash_i$ nodes and $\hat{N} \sim_i \widehat{N'}$. In other words, there exists a descendant $\widehat{N'}$ of $N'$ such that  $\hat{N}$ and $\widehat{N'}$ are post-$crash_i$ nodes and $\hat{N} \sim_i \widehat{N'}$.
 
 This completes the induction and the proof.
\end{proof}

\subsection{Properties of Task Trees from Different Observations}\label{subset:PropertiesOfTreesFromPrefixObservations}
Next, we present the properties of task trees from two observations $G$ and $G'$, where $G'$ is a prefix of $G$. Lemma \ref{lem:prefixObsYieldsSubtree} states that for every path in $\mathcal{R}^{G'}$ that does not contain any edges with $\bot$ action tags, a corresponding path of the same length with the same tags and labels on the corresponding nodes and edges exists in $\mathcal{R}^{G}$. Corollaries \ref{cor:prefixObsPrefixExe} and \ref{cor:NonBotExistsFromPrefix} state that for every node in $\mathcal{R}^{G'}$, there exist nodes in $\mathcal{R}^{G}$ such that both nodes represent the same execution of the system $\mathcal{S}$. Lemma \ref{lem:UniqueNodeMappingFromPrefixObsToSuperObs} proves a stronger property about non-$\bot$ nodes; specifically, it shows that for every non-$\bot$ node in $G'$, there is a corresponding node, called a ``replica'', in $G'$ such that both the nodes have the identical paths from the $\top$ node in their respective execution trees.

Lemma \ref{lem:superObsYieldsSuperTree} states that for every path $p$ in $\mathcal{R}^G$ such that the sequence of distinct non-$(\bot,0,\bot)$ vertex labels in $p$ is a path in $G'$, there exists a corresponding path in $\mathcal{R}^{G'}$ of the same length with the same tags and labels on the corresponding nodes and edges. 

We extend the result from Lemma \ref{lem:UniqueNodeMappingFromPrefixObsToSuperObs} to execution trees constructed from a sequence of observations, where each is a prefix of the next observation in the sequence; in Lemma \ref{lem:nodePersistsOneHop}, we show that non-$\bot$ nodes persist from one execution tree to the next, and in Lemma \ref{lem:nodePersistsForEver}, we show that they persist in an infinite suffix of the execution trees. 

\begin{lemma}\label{lem:prefixObsYieldsSubtree}
Let an observation $G'$ be a prefix of an observation $G$. Fix any path $p'$ in $\mathcal{R}^{G'}$ that starts at the root node and does not contain edges with $\bot$ action tags. Let the length of $p'$ be $k$ edges. Then there exists a ``corresponding'' path $p$ in $\mathcal{R}^G$ of length $k$ such that the following is true. 
(1) For every positive integer $x \leq k+1$, let $N'_x$ be the $x$-th node in $p'$ and let $N_x$ be the $x$-th node in $p$. Then the tags of $N'_x$ are identical to the tags of $N_x$. (2) For any positive integer $x \leq k$, let $E'_x$ be the $x$-th edge in $p'$, and let $E_x$ be the $x$-th edge in $p$. Then the tags and labels of $E'_x$ are identical to the tags and labels of $E_x$.
\end{lemma}
\begin{proof}
Fix $G'$ and $G$ as in the hypothesis of the lemma. The proof follows from a simple induction on the length $k$ of path $p'$.

\emph{Base case.} $k=0$. There exists a single path $p'$ consisting of $k$ edges that starts at the root node of $\mathcal{R}^{G'}$. Let $N'_0$ be the root node of $\mathcal{R}^{G'}$ with vertex tag $v_{N'_0} = (\bot, 0, \bot)$ and config tag $c_{N'_0}$ is the start state of  system $\mathcal{S}$. Similarly, there exists a single path $p$ consisting of $k$ edges that starts at the root node of $\mathcal{R}^G$ and contains no edges. $N_0$ is the root node of $\mathcal{R}^G$ with vertex tag $v_{N_0}= (\bot, 0, \bot)$ and config tag $c_{N_0}$ is the start state of  system $\mathcal{S}$.

\emph{Inductive hypothesis.} For some non-negative integer $k$, for every path $p'$ consisting of $k$ edges in $\mathcal{R}^{G'}$ that starts at the root node and does not contain edges with $\bot$ action tags. Then there exists a ``corresponding'' path $p$ in $\mathcal{R}^G$ consisting of $k$ edges such that the following is true. 
(1) For every positive integer $x \leq k+1$, let $N'_x$ be the $x$-th node in $p'$ and let $N_x$ be the $x$-th node in $p$. Then the tags of $N'_x$ are identical to the tags of $N_x$. (2) For any positive integer $x \leq k$, let $E'_x$ be the $x$-th edge in $p'$, and let $E_x$ be the $x$-th edge in $p$. Then the tags and labels of $E'_x$ are identical to the tags and labels of $E_x$.

\emph{Inductive step.} Fix any path $p'$ consisting of $k+1$ edges that starts at the root node of $\mathcal{R}^{G'}$ and does not contain edges with $\bot$ action tags. Let $p'_{pre}$ be the prefix of $p'$ that consists of $k$ edges. By the inductive hypotheses, there exists a ``corresponding'' path $p_{pre}$ in $\mathcal{R}^G$ consisting of $k$ edges such that the following is true. 
(1) For every positive integer $x \leq k+1$, let $N'_x$ be the $x$-th node in $p'_{pre}$ and let $N_x$ be the $x$-th node in $p_{pre}$. Then the tags of $N'_x$ are identical to the tags of $N_x$. (2) For any positive integer $x \leq k$, let $E'_x$ be the $x$-th edge in $p'_{pre}$, and let $E_x$ be the $x$-th edge in $p_{pre}$. Then the tags and labels of $E'_x$ are identical to the tags and labels of $E_x$.

The last node of $p'_{pre}$ and $p_{pre}$ are $N'_{k+1}$ and $N_{k+1}$, respectively. By the inductive hypotheses, $c_{N'_{k+1}} = c_{N_{k+1}}$ and $v_{N'_{k+1}} = v_{N_{k+1}}$. Consider the node $N''$ that is the last node of path $p'$. By construction, there is an edge $E''$ from $N'_{k+1}$ to $N''$, and furthermore, $a_{E''} \neq \bot$. Let the label of $E''$ be $l''$. Note that either (1) $l''$ is of the form $FD_*$, or (2) $l'' \in T$ is a task in system $\mathcal{S}$. We consider each case separately.

\emph{Case 1.} $l''$ is of the form $FD_*$. Since $a_{E''} \neq \bot$, we know from the construction of the task tree that the vertex tags $v_{E''} = v_{N''}$, $v_{E''}$ is of the form $(i,k,a_{E''})$, where $i$ is a location and $k$ is a positive integer. Furthermore, we know that $v_{E''}$ is a vertex in $G$', and either (a) $v_{N'_{k+1}} = (\bot, 0, \bot)$ or (b) $G'$ contains an edge from $v_{N'_{k+1}}$ to $v_{E''}$. 
From the inductive hypothesis we know that $v_{N'_{k+1}} = v_{N_{k+1}}$. Since $G'$ is a prefix of $G$, we know that $G$ contains the vertex $v_{E''}$. 

If $v_{N'_{k+1}} = (\bot, 0, \bot)$, then $v_{N_{k+1}} = (\bot, 0, \bot)$. Otherwise, $G'$ contains an edge from $v_{N'_{k+1}}$ to $v_{E''}$, and since $G'$ is a prefix of $G$, $G$ contains an edge from $v_{N'_{k+1}}$ to $v_{E''}$.
In both cases, we see that, from the construction of the task tree, $\mathcal{R}^G$ contains an $l''$-edge $\widehat{E''}$ from $N_{k+1}$ to a node $\widehat{N''}$ such that $a_{\widehat{E''}} = a_{E''}$ and $v_{\widehat{E''}} = v_{\widehat{N''}} = v_{E''}$. From the inductive hypothesis, we know that $c_{N'_{k+1}} = c_{N_{k+1}}$.  Since $c_{N''}$ is obtained by applying $a_{E''}$ to $c_{N'_{k+1}}$, and $c_{\widehat{N''}}$ is obtained by applying $a_{\widehat{E''}}$ to $c_{N_{k+1}}$, we see that $c_{N''} = c_{\widehat{N''}}$.

\emph{Case 2.} $l'' \in T$. Since $c_{N'_{k+1}} = c_{N_{k+1}}$, $v_{N'_{k+1}} = v_{N_{k+1}}$, and the system is task deterministic, we know that there exists an outgoing $l''$ edge $\widehat{E''}$ from $N_{k+1}$ to a node $\widehat{N''}$ such that $a_{\widehat{E''}} = a_{E''}$. Since $c_{N''}$ is obtained by applying $a_{E''}$ to $c_{N'_{k+1}}$, and $c_{\widehat{N''}}$ is obtained by applying $a_{\widehat{E''}}$ to $c_{N_{k+1}}$, we see that $c_{N''} = c_{\widehat{N''}}$. Also, by construction, $v_{N''} = v_{N_{k+1}} = v_{E''}$ and $v_{\widehat{N''}} = v_{N'_{k+1}} = v_{\widehat{E''}}$; therefore, $v_{E''} = v_{\widehat{E''}}$ and $v_{N''} = v_{\widehat{N''}}$.

Therefore, in all cases there exists an $l''$-edge $\widehat{E''}$ of $N_{k+1}$ to a node $\widehat{N''}$ in $\mathcal{R}^G$ such that  the tags  of $N''$ and $\widehat{N''}$ are identical, and the tags and labels of $E''$ and $\widehat{E''}$ are identical. Recall that $E''$ is an $l''$-edge from $N'_{k+1}$ to $N''$.

Recall that $p'$ is a path consisting of $k+1$ edges whose prefix is path $p'_{pre}$ consisting of $k$ edges starting from the root node in $\mathcal{R}^{G'}$  and does not contain edges with $\bot$ action tags, and $p_{pre}$ is a path consisting of $k$ edges starting from the root node in $\mathcal{R}^{G}$  and does not contain edges with $\bot$ action tags. Furthermore, the last node of $p'_{pre}$ is $N'_{k+1}$ and the last node of $p_{pre}$ is $N_{k+1}$. Also recall that, (1) for every positive integer $x \leq k+1$, the tags of $N'_x$ are identical to the tags of $N_x$, and (2) for every positive integer $x \leq k$, the tags and labels of $E'_x$ are identical to the tags and labels of $E_x$. Therefore, we extend $p_{pre}$ by edge $\widehat{E''}$ to obtain a path $p$ such that the following is true.

(1) For every positive integer $x \leq k+2$, let $N'_x$ be the $x$-th node in $p'$ and let $N_x$ be the $x$-th node in $p$. Then the tags of $N'_x$ are identical to the tags of $N_x$. (2) For any positive integer $x \leq k+1$, let $E'_x$ be the $x$-th edge in $p'$, and let $E_x$ be the $x$-th edge in $p$. Then the tags and labels of $E'_x$ are identical to the tags and labels of $E_x$.

This completes the induction.
\end{proof}

\begin{corollary}\label{cor:prefixObsPrefixExe}
If an observation $G'$ is a prefix of an observation $G$, then for every node $N'$ in $\mathcal{R}^{G'}$, there exists a node $N$ in $\mathcal{R}^G$ such that $exe(N') = exe(N)$ and $v_{N'} = v_N$.
\end{corollary}
\begin{proof}
Fix a node $N'$ in $\mathcal{R}^{G'}$. By Lemma \ref{lem:everyDescendantNoBot}, we know that there exists a node $N'_{\not\bot}$ in $\mathcal{R}^{G'}$ such that $v_{N'} = v_{N'_{\not\bot}}$ and $exe(N') = exe(N'_{\not\bot})$, and the path from the root to $N'_{\not\bot}$ does not contain any edges with $\bot$ action tag. Invoking Lemma \ref{lem:prefixObsYieldsSubtree}, we know that there exists a node $N$ in $\mathcal{R}^G$ such that the path from the root to $N$ in $\mathcal{R}^G$ and the path from root to $N'_{\not\bot}$ in $\mathcal{R}^{G'}$ contain the same sequence of action tags and vertex tags. Therefore, $exe(N) = exe(N'_{\not\bot})$ and $v_N = v_{N'_{\not\bot}}$. Therefore, $exe(N') = exe(N)$ and $v_{N'} = v_N$.
\end{proof}

\begin{corollary}\label{cor:NonBotExistsFromPrefix}
If an observation $G'$ is a prefix of an observation $G$, then for every node $N'$ in $\mathcal{R}^{G'}$, there exists a non-$\bot$ node $N$ in $\mathcal{R}^G$ such that $exe(N') = exe(N)$ and $v_{N'} = v_{N}$.
\end{corollary}
\begin{proof}
Fix $G'$, $G$ and $N'$ as in the hypotheses of the corollary. Invoking Lemma \ref{cor:prefixObsPrefixExe}, we know there exists a node $N_0$ in $\mathcal{R}^G$ such that $exe(N') = exe(N_0)$ and $v_{N'} = v_{N_0}$. Invoking Corollary \ref{cor:NonBotNodesExist} on node $N_0$, we know that there exists a non-$\bot$ node $N$ in $\mathcal{R}^G$ such that $exe(N_0) = exe(N)$ and $v_{N_0} = v_{N}$. In other words, there exists a non-$\bot$ node $N$ in $\mathcal{R}^G$ such that $exe(N') = exe(N)$ and $v_{N'} = v_{N}$.
\end{proof}

\begin{lemma}\label{lem:UniqueNodeMappingFromPrefixObsToSuperObs}
If an observation $G'$ is a prefix of an observation $G$, then for every node non-$\bot$ node $N'$ in $\mathcal{R}^{G'}$, there exists a unique non-$\bot$ node $N$ in $\mathcal{R}^G$ such that the sequence of labels and vertex tags of the edges from $\top$ to $N'$ in $\mathcal{R}^{G'}$ is identical to the sequence of labels and vertex tags of the edges from $\top$ to $N$ in $\mathcal{R}^G$.
\end{lemma}
\begin{proof}
Fix $G$, $G'$, and $N$ as in the hypothesis of the lemma. Applying Lemma \ref{lem:prefixObsYieldsSubtree} to the path in $\mathcal{R}^{G'}$ from $\top$ to $N$ we conclude at least one non-$\bot$ node $N'$ in $\mathcal{R}^G$ such that the sequence of labels and vertex tags of the edges from $\top$ to $N'$ in $\mathcal{R}^{G'}$ is identical to the sequence of labels and vertex tags of the edges from $\top$ to $N$ in $\mathcal{R}^G$. Fix any such node $N$. Applying Lemma \ref{lem:LabelsAndVertexAgsDenoteUniqueNode} to $N$, we conclude that $N$ is unique.
\end{proof}

\begin{lemma}\label{lem:superObsYieldsSuperTree}
Let an observation $G'$ be a prefix of an observation $G$. Fix any path $p$ in $\mathcal{R}^G$ such that (1) $p$ starts at the root node and (2) the sequence of distinct non-$(\bot,0,\bot)$ vertex tags in $p$ is the sequence of vertices in some path in $G'$. Let the length of $p$ be $k$ edges. 
Then there exists a ``corresponding'' path $p'$ in $\mathcal{R}^{G'}$ of length $k$ such that the following is true. 
(1) For every positive integer $x \leq k+1$, let $N'_x$ be the $x$-th node in $p'$ and let $N_x$ be the $x$-th node in $p$. Then the tags of $N'_x$ are identical to the tags of $N_x$.
(2) For any positive integer $x \leq k$, let $E'_x$ be the $x$-th edge in $p'$, and let $E_x$ be the $x$-th edge in $p$. Then the tags and labels of $E'_x$ are identical to the tags and labels of $E_x$.
\end{lemma}
\begin{proof}
The proof follows from a simple induction on the length $k$ of path $p$. 
\end{proof}

Given any pair of observations $G$ and $G'$ such that $G'$ is a prefix of $G$, and given a non-$\bot$ node $N'$ in $G'$, we define the \emph{replica} of $N'$ from $G'$ in $G$ to be the unique node $N$ in $G$ that satisfies Lemma  \ref{lem:UniqueNodeMappingFromPrefixObsToSuperObs}. We use this notion of a replica node to talk about a non-$\bot$ node ``persisting'' over task trees constructed from a sequence observations such that each observation in the sequence is a prefix of each succeeding observation.

Given a non-$\bot$ node $N$ in a tree $\mathcal{R}^G$ and its replica $N'$ in a tree $\mathcal{R}^{G'}$, since the sequence of labels and vertex tags of the edges from $\top$ to $N$ in $\mathcal{R}^G$ is identical to the sequence of labels and vertex tags of the edges from $\top$ to $N'$ in $\mathcal{R}^{G'}$, we refer to any non-$\bot$ node $N$ and its replicas as $N$.

Let $\mathcal{G} = G_1,G_2,\ldots$ be an infinite sequence of finite observations such that (1) for any positive integer $x$, $G_x$ is a prefix of $G_{x+1}$, and (2) the sequence of observations converge to some observation $G^\infty$.
\begin{lemma}\label{lem:nodePersistsOneHop}
Fix a positive integer $x$ and suppose $N$ is a non-$\bot$ node  $\mathcal{R}^{G_x}$. Then $\mathcal{R}^{G_{x+1}}$ contains $N$.
\end{lemma}
\begin{proof}
Follows from Lemma \ref{lem:UniqueNodeMappingFromPrefixObsToSuperObs}.
\end{proof}

\begin{lemma}\label{lem:nodePersistsForEver}
Fix a positive integer $x$ and suppose $N$ is a non-$\bot$ node  $\mathcal{R}^{G_x}$. Then for any $x' > x$, $\mathcal{R}^{G_{x'}}$ contains $N$.
\end{lemma}
\begin{proof}
The proof follows from a simple induction on $x-x'$.
\end{proof}
 
\subsection{Fair Branches of Execution Trees}\label{subset:FairBranchesOfExecutionTrees}


In this subsection, we define fair branches of execution trees, and we establish the correspondence between fair branches
in the execution trees and fair traces of system $\mathcal{S}$.\footnote{Recall that $\mathcal{S}$ consists of the process automata, the environment automaton, and the the channel automata.} 

We define a branch of an infinite task tree $\mathcal{R}^G$ of the observation $G$ to be a \emph{fair branch} if, for each label $l$, the branch contains an infinite number of edges labeled $l$. Therefore, a fair branch satisfies the following properties. 
\begin{lemma}\label{lem:fairTreeBranchProperties}
For each location $i$, and each fair branch $b$ of $\mathcal{R}^G$, the following are true.
\begin{enumerate}
\item Branch $b$ contains infinitely many $FD_i$, $Proc_i$ and $Env_{i,x}$ edges (for all $x \in X_i$) (regardless of whether $i$ is live or not live in $G$).
\item If $i$ is live in $G$, then (a) every $FD_i$ edge in $b$ has a non-$\bot$ action tag and (b) some infinite subset of the $O_{D,i}$ events contained in $G$ occur in $b$.\footnote{Note that $b$ is not guaranteed to contain all the $O_{D,i}$ events contained in $G$.}
\item If $i$ is not live in $G$, then there exists a suffix of $b$ such that the action tag of each $FD_i$, $Proc_i$, and $Env_{i,x}$ edge (for all $x \in X_i$) is $\bot$.
\end{enumerate}
\end{lemma}
For any location $i$ and fair branch $b$ of $\mathcal{R}^G$, $b$ may contain a $Proc_i$ or an $Env_{i,x}$ edge $E$ such that $a_E = \bot$ for either of two reasons. (1) If $i$ is not live in $G$, then it may be the case that there is no outgoing edge from $v_E$ to any vertex whose location is $i$. (2) There is no enabled action from the corresponding task in $c_N$, where $N$ is the node immediately preceding $E$ in $b$; this is regardless of whether $i$ is live in $G$ or otherwise.



The main result of this subsection is Theorem \ref{thm:fairBranchFairExec}, which says that, if $D$ is a strong-sampling AFD, then for any viable observation $G$ of $D$ and for every fair branch $b$ in $\mathcal{R}^G$, (1) the projection of $b$ on the actions of the system $\mathcal{S}$ corresponds to a fair trace of system $\mathcal{S}$, and (2) the projection of $b$ on the AFD actions corresponds to a trace in $T_D$.
We use  multiple helper lemmas to prove the main result, which we summarize after the following definitions.

 For the remainder of this section, fix $D$ to be a strong-sampling AFD and fix $G$ to be an infinite observation of $D$. 
 Consider a branch $b$ in $\mathcal{R}^{G}$; since $G$ is an infinite observation, $b$ must also be of infinite length. Let the sequence of nodes in
$b$ be $\top,N_1,N_2,\ldots$ in that order. The sequence $exe(b)$ is
the limit of the prefix-ordered infinite sequence $exe(\top)$,
$exe(N_1)$, $exe(N_2)$, $\ldots$.\footnote{Note that we have 
  overloaded the function $exe$ to map both nodes and branches to sequences 
  of alternating states and actions. Since the domains of all the instances
  of the $exe()$ function are disjoint, we can refer to $exe(N)$ or $exe(b)$   
  without any ambiguity.} Note that $exe(b)$ may be a finite
or an infinite sequence.
Let $trace(b)$ denote the trace of the execution $exe(b)$.
Recall that for any node $N$ in $\mathcal{R}^G$,  $\mathcal{R}^G|_N$ denotes the maximal subtree of $\mathcal{R}^G$ rooted at $N$.
  
 In Lemma \ref{lem:fairBranchEventSequence}, we show that for any fair branch $b$ in $\mathcal{R}^G$, $exe(b)|_{O_D}$ is the event-sequence of some fair branch in $G$. 
However, note that even if $b$ is a fair branch of $\mathcal{R}^G$, $exe(b)$ need not be a fair execution of $\mathcal{S}$; also, even if $G$ is viable for $D$, the projection of $exe(b)$ on $O_D \cup \hat{I}$ need not be in $T_D$. The primary reason for these limitations is that the tree $\mathcal{R}^G$ does not contain any $crash$ events. We rectify this omission  in Lemma \ref{lem:fairExe}; we insert $crash$ events in $trace(b)$ to obtain a trace $t_{\mathcal{S}}$ of $\mathcal{S}$ such that $t_{\mathcal{S}}$ is a fair trace of $\mathcal{S}$, and if $G$ is viable for $D$, then $t_{\mathcal{S}}|_{O_D \cup \hat{I}}$ is compatible with $G$.
%
%
Lemma \ref{lem:fairExe} implies Theorem \ref{thm:fairBranchFairExec}.

\begin{lemma}
\label{lem:fairBranchEventSequence}
For every fair branch $b$ of $\mathcal{R}^{G}$, $exe(b)|_{O_D}$ is the
event-sequence of some fair branch in $G$.
\end{lemma}
\begin{proof}
Fix $b$ to  be a fair branch of $\mathcal{R}^{G}$.
Let $b = \top,E_1,N_1,E_2,N_2,\ldots$, where for each natural number $x$, $N_x$ is a node in $\mathcal{R}^G$ and $E_x$ is an edge with lower endpoint $N_x$ in $\mathcal{R}^G$.
Applying Lemma \ref{prop:nodeFiniteExe}, we know that
for any positive integer $x$, if $v_{N_x} = (\bot, 0, \bot)$, then $exe(N_x)|_{O_D}$ is the empty sequence, and otherwise, $exe(N_x)|_{O_D}$ ends with the event of 
$v_{N_x}$. Note that since $G$ is an infinite observation and $b$ is a fair branch of $\mathcal{R}^{G}$, there exists a positive integer $x$ such that for all $x' \geq x$, $v_{N_{x'}} \neq (\bot, 0, \bot)$\footnote{
We know such a positive integer $x$ exists for the following reason. Since $G$ is an infinite observation, $G$ has some live location $i$. By Lemma \ref{lem:fairTreeBranchProperties}, we know that every $FD_i$ edge in $\mathcal{R}^G$ has a non-$\bot$ action tag, and this can happen only if for each $FD_i$ edge, the vertex tag of the node preceding that edge is not $(\bot, 0, \bot)$. Since $b$ is a fair branch, $b$ contains infinitely many such nodes; fix any such a node $N$. By Lemma \ref{prop:FDchildVertexEdgeExists}, we know that for each descendant $\hat{N}$ of $N$ in $b$, $v_{\hat{N}}$ is a vertex in $G$ and therefore, $v_{\hat{N}} \neq (\bot, 0, \bot)$.}.
Applying Lemma \ref{prop:ancesterPrefixAFDEvents}, we
know that for any positive integer $x$, $exe(N_x)|_{O_D}$ is a prefix of
$exe(N_{x+1})|_{O_D}$. Therefore, $exe(b)|_{O_D}$ is the limit of the event-sequence of 
$v_{N_1},v_{N_2},\ldots$. By the construction of $\mathcal{R}^G$, this means that, $exe(b)|_{O_D}$ is the
event-sequence of some branch $b' = v_{N_1},v_{N_2},\ldots$ in $G$. It
remains to show that $b'$ is a fair branch in $G$. 
Recall that $b'$ is a fair branch if for every location $i$ that is live in $G$, $b$ contains an infinite number of vertices whose location is $i$.

Fix a location $i \in live(G)$. Since $b$ is a fair branch of $\mathcal{R}^G$, there are infinitely many edges in $b$ whose
label is $FD_i$; for each such $FD_i$-edge, applying Lemma
\ref{prop:liveLocationNonBotAction}, we know that the action tag of
the $FD_i$-edge is non-$\bot$. Therefore, the sequence $v_{N_1},v_{N_2},\ldots$ contains infinitely many
vertices whose location is $i$. Thus, by definition, $b'$ is a fair branch
in $G$. Therefore, $exe(b)|_{O_D}$ is the event-sequence of $b'$, which is a fair
branch in $G$.
\end{proof}

Next, we assume that $G$ is a viable observation for $D$. 
In Lemma \ref{lem:fairExe}, for each fair branch $b$ of $\mathcal{R}^G$, we insert $crash$ events in $trace(b)$ to get a trace $t_{\mathcal{S}}$ of the system $\mathcal{S}$ such that $trace(b) =
t_\mathcal{S}|_{act(\mathcal{S}) \setminus \hat{I}}$ and $t_{\mathcal{S}}|_{O_D \cup \hat{I}} \in T_D$.

\begin{lemma}
\label{lem:fairExe}
For every fair branch $b$ of $\mathcal{R}^{G}$, there exists a fair
execution  $\alpha_z$ of the system $\mathcal{S}$ such that $trace(b) = \alpha_z|_{act(\mathcal{S}) \setminus \hat{I}}$ and $\alpha_z|_{O_D \cup \hat{I}} \in T_D$.
\end{lemma}

\begin{proof}
Fix a fair branch $b$ of $\mathcal{R}^{G}$. Let $b = \top,E_1,N_1,E_2,N_2,\ldots$, where for each natural number $x$, $N_x$ is a node in $\mathcal{R}^G$ and $E_x$ is an edge with lower endpoint $N_x$ in $\mathcal{R}^G$.
By Lemma \ref{lem:fairBranchUnfairExe}, we know that $exe(b)$ is an execution of system $\mathcal{S}$. 
We construct a new execution $\alpha_z$ of system $\mathcal{S}$ by starting with $exe(b)$, and inserting $crash$ events as permitted by Lemma \ref{lem:addCrashAfterLastEvent}; we then define $\alpha_z$ to be an execution whose trace $t_z$. In order to invoke Lemma \ref{lem:addCrashAfterLastEvent} we must ascertain the specific positions within $exe(b)$ where we may insert $crash$ events. We determine these positions, by deriving a trace $t'_G \in T_D$ such that the sequence of AFD output events in $t'_G$ is the projection of $exe(b)$ on AFD output events. We then use the positions of $crash$ events in $t'_G$ to determine the positions in $exe(b)$ where $crash$ events are inserted.

Recall that $G$ is a viable observation for $D$. By Lemma \ref{lem:fairBranchEventSequence}, we know that $exe(b)|_{O_D}$ is the event-sequence of some fair branch $b'_G$ in $G$. Let $t_G \in T_D$ be compatible with $G$, and we assume that $t_G$ has no extra crashes.\footnote{Note that for any trace $t_G$ that is compatible with $G$, the trace $mincrash(t_G)$ is also compatible with $G$ and does not contain any extra crashes. So, it is reasonable to assume that $t_G$ does not contain any extra crashes.}
By Lemma \ref{prop:sampledSubsequence}, we know that there exists $t'_G \in T_D$ such that $t'_G|_{O_D}$ is a strong sampling of $t_G$ and $t'_G$ is the event sequence of $b'_G$. Fix such a trace $t'_G$. By construction, $t'_G|_{O_D} = exe(b)|_{O_D} = trace(b)|_{O_D}$ and $live(t'_G) = live(G)$. Note that $t'_G$ does not contain any extra crashes.

For each location $i$ that is not live in $G$, let $e_{i^+}$ be the earliest event from $O_D$ that follows the $crash_i$ event in $t'_G$.

We construct $\alpha_z$ by iteratively applying Lemma \ref{lem:addCrashAfterLastEvent} to $exe(b)$, once for each location $i$ that is not live in $G$, as follows. Starting with  $trace(b)$, for each location $i$ that is not live in $G$, insert $crash_i$ immediately before event $e_{i^+}$. If more than one crash event is inserted in the same position in $trace(b)$, order these crash events in the order in which they appear in $t'_G$.
Let the trace, thus obtained, be $t_z$. 

Note that by construction $t_z|_{\hat{I} \cup O_D} = t'_G$. Therefore, $t_z|_{\hat{I} \cup O_D}$ is a strong sampling of $t_G$. For each location $i$ that is not live in $G$, let vertex $v_{i^+}$ be the vertex corresponding to event $e_{i^+}$; since $t_z|_{\hat{I} \cup O_D}$ is a strong sampling of $t_G$, $crash_i$ precedes $e_{i^+}$ in $t_G$, and therefore, there are no edges from $v_{i^+}$ to any vertex whose location is $i$. Therefore, by construction of $\mathcal{R}^G$, for any node $N$ whose vertex tag is $v_{i^+}$, and for any outgoing $FD_i$-edge $E$ from $N$, $a_E = \bot$. Therefore, by Lemma \ref{lem:crashedLocationNoOutgoingActions}, we know that for outgoing $Proc_i$, $Env_i$ and $FD_i$ edges from the descendants  of $N$, their action tags are also $\bot$. Therefore, in $trace(b)$, for each location $i$ that is not live in $G$, there are no $Proc_i$, $O_{D,i}$, or $\mathcal{E}_i$ events following $e_{i^+}$.
Therefore, starting with $trace(b)$ and iteratively applying Lemma \ref{lem:addCrashAfterLastEvent} for each $crash$ event inserted, we conclude that there exists an execution $\alpha_z$ of $\mathcal{S}$ whose trace is $t_z$.

It remains to show that (1) $\alpha_z|_{O_D \cup \hat{I}} \in T_D$ and (2) $\alpha_z$ is a fair execution of $\mathcal{S}$. We prove each part separately.

\emph{Claim 1.} $\alpha_z|_{\hat{I} \cup O_D} \in T_D$.
\begin{proof}
Note that by construction $\alpha_z|_{\hat{I} \cup O_D} = t'_G$ and therefore, $\alpha_z|_{\hat{I} \cup O_D} \in T_D$.
\end{proof}



\emph{Claim 2.}
$\alpha_z$ is a fair execution of $\mathcal{S}$.

\begin{proof}
By construction, $\alpha_z$ is an execution of $\mathcal{S}$. In order to show that $\alpha_z$ is a fair execution of $\mathcal{S}$, we have to show the following. 
(a) If $\alpha_z$ is finite, then for each task $l \in T$, $l$ is not enabled in the final state of $\alpha_z$; and (b) if $\alpha_z$ is infinite, then for each task $l \in T$, $\alpha_z$ contains either infinitely many events from $l$ or infinitely many occurrences of states in which $l$ is not enabled. (Recall that $T$ is the set of tasks in $\mathcal{S}$.)

\emph{Case (a)} $\alpha_z$ is finite. We show that this is impossible as follows. Assume for contradiction that $\alpha_z$ is finite. Since $G$ is an infinite observation, there exists a location $j$ such that there are infinitely many vertices in $G$ whose location is $j$. Since $b$ is a fair branch of $\mathcal{R}^G$,  we know that $b$ contains infinitely many $FD_j$ edges. Applying Lemma \ref{prop:liveLocationNonBotAction}, we conclude that the action tag of each $FD_i$-edge in $b$ is non-$\bot$, and therefore, $exe(b)$ is infinite. Therefore, $trace(b)$ is infinite. Since $t_z$ is obtained by inserting events into $trace(b)$, $t_z$ is infinite, and consequently $\alpha_z$ is infinite. Thus, we have a contradiction.

\emph{Case (b)} $\alpha_z$ is infinite. For contradiction, assume that $\alpha_z$ is not a fair execution. Therefore, there must exist a task $l$ such that  $\alpha_z$ contains only finitely many events from $l$ and only finitely many occurrences of states in which $l$ is not enabled. Fix such an $l$. We consider each possible value of $l$.

\begin{itemize}
\item $l \in \set{Chan_{j,k} | j \in \Pi, k \in \Pi \setminus \set{j}}$. From the construction of $\mathcal{R}^G$, we know that for each $l$-edge $E$ from a node $N'$ in $b$, if some action $a$ in $l$ is enabled in $c_{N'}$, then $a_E = a$. Furthermore, note that in any execution $\alpha$ of $\mathcal{S}$, if some action $a$ in $l$ is enabled in a state $s$ of $\alpha$, $a$ remains enabled in the suffix of $\alpha$ following $s$ until $a$ occurs. By assumption, since there are only finitely many events from $l$ in $\alpha_z$, and $t_z$ (the trace of $\alpha_z$) is constructed by inserting events into $trace(b)$, it follows that in some infinite suffix of $b$, for each node $N'$, no action from $l$ is enabled in $c_{N'}$. Since inserting $crash$ events does not change the state of the channel automata, it follows that no action from $l$ is enabled in some infinite suffix of $\alpha_z$. This contradicts our assumption that $\alpha_z$ contains only finitely many occurrences of states in which $l$ is not enabled.

\item  $l \in \set{Proc_j, Env_{j,x} | j \in \Pi, x \in X_j}$. Fix the location of $l$ to be $k$. We consider two subcases: (i) $k$ is not live in $G$, and (ii) $k$ is live in $G$.
  \begin{itemize}
  \item $k$ is not live in $G$. By construction, $b$ contains only finitely many $l$-edges whose action tags are non-$\bot$, and by construction of $\alpha_z$, we know that $\alpha_z$ contains a $crash_k$ event, following which there are no events from task $l$. In other words, $\alpha_z$ contains only finitely many events from $l$. However, recall that a $crash_k$ event disables all the actions from $l$ forever thereafter. Therefore, in the suffix of $\alpha_z$ following a $crash_i$ event, no action from $l$ is enabled. This contradicts our assumption that    $\alpha_z$ contains only finitely many events from $l$ and only finitely many occurrences of states in which $l$ is not enabled.
  
  \item $k$ is live in $G$. Therefore $G$ contains infinitely many vertices whose location is $k$. Note that in $b$, $l$-edges occur infinitely often. By construction of the tree $\mathcal{R}^G$, we know that for each node $N'$ in $b$ that immediately precedes an $l$-edge $E'$, either $v_{N'}$ is not a vertex in $G$ and $G$ contains infinitely many vertices whose location is $k$, or $v_{N'}$ has an outgoing edge to some vertex in $G$ whose location is $k$; consequently, if some action in $l$ is enabled in $N'$, then $a_{E'} \neq \bot$. Therefore, if $exe(b)$ contains only finitely many events from $l$, then it must have only finitely many occurrences of states in which $l$ is enabled; in other words, $exe(b)$ contains infinitely many occurrences of states in which $l$ is not enabled. 
  
  By construction of $\alpha_z$, we know that $\alpha_z$ does not contain a $crash_k$ event. Since $t_z$ (the trace of $\alpha_z$) is obtained from $trace(b)$ by inserting only $crash$ events and $trace(b)$ does not contain any $crash_k$ events, we know that the projection of $exe(b)$ on the states of $Proc_k$ and $\mathcal{E}_k$ is equal to the projection of $\alpha_z$ on the states of $Proc_k$ and $\mathcal{E}_k$. Therefore, if $exe(b)$ contains infinitely many occurrences of states in which $l$ is not enabled, then $\alpha_z$ contains infinitely many occurrences of states in which $l$ is not enabled. Thus we have a contradiction.
  \end{itemize}
\end{itemize} 
Thus, we have proved that $\alpha_z$ is a fair execution of $\mathcal{S}$.
\end{proof}

The proof follows from Claims 1 and 2.
\end{proof}

\begin{theorem}\label{thm:fairBranchFairExec}
Let $D$ be a strong-sampling AFD. Let $G$ be a viable observation for $D$.
For every fair branch $b$ of $\mathcal{R}^{G}$, there exists a fair
trace $t_{\mathcal{S}}$ of $\mathcal{S}$ such that $trace(b) = t_{\mathcal{S}}|_{act(\mathcal{S}) \setminus \hat{I}}$ and $t_{\mathcal{S}}|_{O_D \cup \hat{I}} \in T_D$.
\end{theorem}


\begin{proof}
 Fix $D$ and $G$ as in the hypotheses of the theorem statement. The proof follows directly from Lemma \ref{lem:fairExe}.
\end{proof}

\section{Consensus Using Strong-Sampling AFDs}
\label{sec:consensusAndAFD}


In this section, we show how a strong-sampling AFD sufficient to solve 
crash-tolerant  consensus circumvents the impossibility of consensus
in asynchronous systems. We use this result in the next section to
demonstrate that $\Omega_f$ is a weakest strong-sampling AFD to
solve $f$-crash-tolerant binary consensus, which is defined next.

 
\subsection{Crash-Tolerant Binary Consensus}
\label{subsec:consensusDef} 

For any $f$  in $[0,\ldots,n-1]$, the \emph{$f$-crash-tolerant binary consensus} problem $P =(I_P,O_P,T_{P,f})$ is specified as follows. The set $I_P$ is $\{ propose(v)_i| v \in \{0,1\} \wedge i \in\Pi\} \cup \set{crash_i | i \in\Pi}$, and the set $O_P$ is $\{ decide(v)_i | v \in \{0,1\} \wedge i \in\Pi\}$. Before defining the set of sequences $T_{P,f}$, we provide the following auxiliary definitions. 

Let $t$ be an arbitrary (finite or infinite) sequence over $I_P \cup O_P$. The following definitions apply to the sequence $t$. 

\paragraph{Decision value.} If an event $decide(v)_i$ occurs for some $i\in\Pi$ in sequence $t$, then $v$ is said to be a \emph{decision value} of $t$.

\paragraph{Environment well-formedness:} 
The \emph{environment well-formedness} property states that (1) the environment provides each location with 
at most one input value, (2) the environment does not provide any input values at a location after a crash event at that location, and (3) the environment provides each live location with exactly one input value. 
Precisely, (1) for each location $i \in \Pi$
at most one event from the set $\{propose(v)_i|v \in \{0,1\}\}$ occurs in $t$, (2) for each location $i \in faulty(t)$ no event from the set $\{propose(v)_i|v \in \{0,1\}\}$ follows a $crash_i$ event in $t$, and (3) for each location $i \in live(t)$ exactly one event from the set $\{propose(v)_i|v \in \{0,1\}\}$ occurs in $t$.

\paragraph{$f$-crash limitation:} The \emph{$f$-crash limitation} property states that at most $f$ locations crash. Precisely, there exist at most $f$ locations $i$ such that $crash_i$ occurs in $t$. 

\paragraph{Crash validity:} The \emph{crash validity} property states that no location decides after crashing.  That is, for every location $i \in crash(t)$, no event from the set $\set{decide(v)_i| v \in \set{0,1}}$ follows a $crash_i$ event in $t$.

\paragraph{Agreement:} The \emph{agreement} property states that no two locations decide differently. That is, if two events $decide(v)_i$ and $decide(v')_j$ occur in $t$, then $v=v'$. 

\paragraph{Validity:} The \emph{validity} property states that any decision value at any location must be an input value at some location. That is, for each location $i \in \Pi$, if an event $decide(v)_i$ occurs in $t$, then there exists a location $j \in \Pi$ such that the event $propose(v)_j$ occurs in $t$.

\paragraph{Termination:} The \emph{termination} property states that each location decides at most once, and each live location decides exactly once. That is, for each location $i \in \Pi$, at most one event from the set $\set{decide(v)_i| v \in \set{0,1}}$ occurs in $t$, and for each location $i \in live(t)$, exactly one event from the set $\set{decide(v)_i| v \in \set{0,1}}$ occurs in $t$.

Using the above definitions, we define the set $T_{P,f}$ for $f$-crash-tolerant binary consensus as follows.

\paragraph{The set $T_{P,f}$.} $T_{P,f}$ is the set of all sequences $t$ over $I_P \cup O_P$ such that, if $t$ satisfies  environment well-formedness and $f$-crash limitation,
then $t$ satisfies crash validity, agreement, validity, and termination. Note that $T_{P,f}$ contains all the sequences over $I_P \cup O_P$ in which more than $f$ locations crash; informally, $f$-crash-tolerant consensus provides no guarantees if more than $f$ locations crash.

\subsection{A Well-formed Environment Automaton for Consensus}\label{subsec:conEnvDef}
Given an environment automaton $\mathcal{E}$ whose set of input actions is $O_P \cup \hat{I}$ and set of output actions is $I_P\setminus \hat{I}$, $\mathcal{E}$ is said to be a \emph{well-formed environment} iff every fair trace $t$ of $\mathcal{E}$ satisfies environment well-formedness. For our purpose, we assume a specific well-formed environment $\mathcal{E}_C$ defined next.

The automaton $\mathcal{E}_C$ is a composition of $n$ automata $\set{\mathcal{E}_{C,i} | i \in\Pi}$. Each automaton $\mathcal{E}_{C,i}$ has two output actions $propose(0)_i$ and $propose(1)_i$, three input actions $decide(0)_i$, $decide(1)_i$, and $crash_i$, and no internal actions. Each output action constitutes a separate task. Action $propose(v)_i$, where $v \in \set{0,1}$, permanently disables actions $propose(v)_i$ and $propose(1-v)_i$. The $crash_i$ input action disables  actions $propose(v)_i$ and $propose(1-v)_i$. The automaton $\mathcal{E}_{C,i}$ is shown in Algorithm \ref{alg:ConsensusEnv}.

Next, we show that $\mathcal{E}_C$ is a well-formed environment automaton. Observe that the automaton $\mathcal{E}_{C}$ satisfies the following Lemma.

\begin{algorithm}\footnotesize
\caption{Automaton $\mathcal{E}_{C,i}$, where $i \in \Pi$. The composition of $\set{\mathcal{E}_{C,i} | i \in\Pi}$ constitutes the environment automaton $\mathcal{E}_C$ for consensus.}
\label{alg:ConsensusEnv}

\textbf{Signature:}

\tab input $crash_i$, $decide(0)_i$, $decide(1)_i$

\tab output $propose(0)_i$, $propose(1)_i$

\textbf{Variables:}

\tab $stop$: Boolean, initially $false$

\textbf{Actions:}

\tab input $crash_i$

\tab effect

\tab \tab $stop$ := $true$

\tab

\tab input $decide(b)_i$, $b \in \set{0,1}$

\tab effect

\tab \tab *none*

\tab

\tab output $propose(b)_i$, $b \in \set{0,1}$

\tab precondition

\tab \tab $stop = false$

\tab effect

\tab \tab $stop$ := $true$

\tab

\textbf{Tasks:}

\tab $Env_{i,0} = \set{propose(0)_i}$, $Env_{i,1} = \set{propose(1)_i}$

\end{algorithm}

Note that for each location $i$, each action $propose(v)_i$ (where $v \in \set{0,1}$ and $i \in \Pi$) in $\mathcal{E}_C$ constitutes a separate task $Env_{i,v}$ in $\mathcal{E}_{C,i}$.

\begin{lemma}\label{prop:proposeDisables}
 In $\mathcal{E}_C$, action $propose(v)_i$ (where $v \in \set{0,1}$ and $i \in \Pi$) permanently disables the actions $propose(v)_i$ and $propose(1-v)_i$.
\end{lemma}
\begin{proof}
 Fix $v \in \{0,1\}$ and $i \in \Pi$. From the pseudocode in Algorithm \ref{alg:ConsensusEnv}, we know that the precondition for actions $propose(v)_i$ and $propose(1-v)_i$ is $(stop = false)$. We also see that the effect of action $propose(v)_i$ is to set $stop$ to $false$. Thus, the Lemma follows.
\end{proof}

\begin{theorem}\label{thm:environmentWellFormed}
 Automaton $\mathcal{E}_C$ is a well-formed environment.
\end{theorem}
\begin{proof}
 To establish the theorem, we have to prove the following three claims for every fair trace $t$ of $\mathcal{E}_C$. (1) For each location $i \in \Pi$, at most one event from the set $\set{propose(v)_i| v \in \set{0,1}}$ occurs in $t$. (2) For each location $i \in faulty(t)$, no event from the set $\set{propose(v)_i| v \in \set{0,1}}$ follows a $crash_i$ event in $t$. (3) For each location $i \in live(t)$, exactly one event from the set $\set{propose(v)_i| v \in \set{0,1}}$ occurs in $t$.
 
 \textit{Claim 1.} For each location $i \in \Pi$, at most one event from the set $\set{propose(v)_i| v \in \set{0,1}}$ occurs in $t$.
 \begin{proof}
 Fix $i$.
 If no event from $\set{propose(v)_i| v \in \set{0,1}}$ occurs in $t$, then the claim is satisfied. For the remainder of the proof of this claim, assume some event from $\set{propose(v)_i| v \in \set{0,1}}$ occurs in $t$; let $e$ be the earliest such event. Let $t_{pre}$ be the prefix of $t$ that ends with $e$. After event $e$ occurs, we know from Lemma \ref{prop:proposeDisables} that $e$ disables all actions in  $\set{propose(v)_i| v \in \set{0,1}}$. Therefore, the suffix of $t$ following $t_{pre}$, no event from $\set{propose(v)_i| v \in \set{0,1}}$ occurs.
 \end{proof}
 
 \textit{Claim 2.} For each location $i \in faulty(t)$, no event from the set $\set{propose(v)_i| v \in \set{0,1}}$ follows a $crash_i$ event in $t$.
 \begin{proof}
 Fix $i$ to be a location in $faulty(t)$. From the pseudocode in Algorithm \ref{alg:ConsensusEnv}, we know that action $crash_i$ sets $stop$ to $true$. Furthermore, no action sets $stop$ to $false$. Also, observe that the precondition for actions in  $\set{propose(v)_i| v \in \set{0,1}}$ is $stop=false$. Therefore, actions in $\set{propose(v)_i| v \in \set{0,1}}$ do not follow a $crash_i$ event in $t$.
 \end{proof}
 
 \textit{Claim 3.}
 For each location $i \in live(t)$, exactly one event from the set $\set{propose(v)_i| v \in \set{0,1}}$ occurs in $t$.
 
 \begin{proof}
 Fix $i$ to be a location in $live(t)$.  In Algorithm \ref{alg:ConsensusEnv}, we see that $stop$ is initially $false$, and is not set to true until either $crash_i$ occurs or an event from $\set{propose(v)_i| v \in \set{0,1}}$ occurs. Since $i \in live(t)$, we know that $crash_i$ does not occur in $t$.
 Since $t$ is a fair trace, actions in $\set{propose(v)_i| v \in \set{0,1}}$ remain enabled until one of the actions occur. After one event from $\set{propose(v)_i| v \in \set{0,1}}$ occurs, from Claim 1, we know that no more events from $\set{propose(v)_i| v \in \set{0,1}}$ occur.
 \end{proof}
 
 The theorem follows from Claims 1, 2, and 3.
\end{proof}

\subsection{System Definition}
\label{subsec:consensusSystem}

For the remainder of this section, fix a strong-sampling AFD $D$, a distributed algorithm $A$, and a natural number $f$ ($f< n$) such that $A$ solves $f$-crash-tolerant
binary consensus using AFD $D$ in environment $\mathcal{E}_C$.
Let $\mathcal{S}$ be a system that is composed of distributed
algorithm $A$, channel automata, and the well-formed environment
automaton $\mathcal{E}_C$.

Based on the properties of $f$-crash-tolerant binary consensus and
system $\mathcal{S}$, we have the following Lemma which
restricts the number of decision values in an execution of
$\mathcal{S}$.


\begin{lemma}
\label{lem:fairExecExactlyOneDecision}
For every
 fair execution $\alpha$ of $\mathcal{S}$, where
$\alpha|_{\hat{I} \cup O_D} \in T_D$ and $\alpha|_{I_P \cup O_P}$
satisfies $f$-crash-limitation, $\alpha|_{I_P \cup O_P}$ has exactly
one decision value.
\end{lemma}

%
%


\begin{proof}
Fix $\alpha$ to be a fair execution of $\mathcal{S}$ such that
$\alpha|_{\hat{I} \cup O_D} \in T_D$ and $\alpha|_{I_P \cup O_P}$
satisfies $f$-crash-limitation. Recall that $\mathcal{S}$ consists of
a distributed algorithm $A$ that solves $f$-crash-tolerant binary
consensus using AFD $D$, the channel automata, and
$\mathcal{E}$. Since $\alpha|_{\hat{I} \cup O_D} \in T_D$, we know
from the definition of ``solving a problem using an AFD'' that
$\alpha|_{I_P \cup O_P} \in T_{P,f}$.
%

Recall that $T_{P,f}$ is the set of all sequences $t$ over $I_P \cup O_P$
such that if $t$ satisfies environment well-formedness and $f$-crash
limitation, then $t$ satisfies crash validity, agreement, validity,
and termination. We assumed that $\alpha|_{I_P \cup O_P}$ satisfies
$f$-crash limitation.

From Theorem \ref{thm:environmentWellFormed}, we know that
$\mathcal{E}_C$ is a well-formed environment. Therefore, $\alpha|_{I_P
  \cup O_P}$ satisfies environment well-formedness. Consequently,
$\alpha|_{I_P \cup O_P}$ satisfies agreement and termination. By the
agreement property we know that $\alpha|_{I_P \cup O_P}$ contains at
most one decision value. Since $f < n$, we know that there is at least one location for which no crash event occurs, and therefore, by the
termination property, we know that at least one location decides. In
other words,  $\alpha|_{I_P \cup O_P}$ has exactly one decision value.
\end{proof}

\subsection{Trees of Executions}\label{subsec:consensusTree}
For the remainder of this section, fix $G$ to be an arbitrary
viable observation of $D$ such that at most $f$ locations are not
live in $G$. Recall the construction of the execution trees from Section \ref{subsec:treeOfExec}; construct
the tree $\mathcal{R}^{G}$ for system $\mathcal{S}$. 

The primary reasons for fixing $G$ to be a viable observation are the following. Lemmas \ref{prop:fairBranchOneDecision}, \ref{prop:eachNodeAtMostOneDecision}, and \ref{prop:eachNodeAtMostOneDecisionInAllGs}, which talk about possible decision values in branches of $\mathcal{R}^{G}$, are true only for viable observations. Furthermore, the notion of ``valence'' defined in Section \ref{subsec:valence} is applicable only when Lemmas \ref{prop:fairBranchOneDecision}, \ref{prop:eachNodeAtMostOneDecision}, and \ref{prop:eachNodeAtMostOneDecisionInAllGs} hold, and consequently, ``valence'' makes sense only for viable observations. Since the rest of Section \ref{sec:consensusAndAFD} discusses the properties of branches of the execution trees and their valences, we must fix $G$ to be a viable observation for the remainder of the section.

Since $G$ is a viable observation of $D$, by definition, there exists a trace $t_D \in T_D$ such that
$t_D|_{O_D}$ is the event sequence of some topological ordering of the
vertices in $G$. Fix such a trace $t_D$ for the remainder of this section.

The set $L$ of labels in  $\mathcal{R}^{G}$ is $\set{FD_i|i\in\Pi}
\cup \set{Proc_i| i \in \Pi} \cup  \set{Env_{i,v}|i\in\Pi \wedge v \in
  \{0,1\}} \cup \set{Chan_{i,j}| i\in\Pi \wedge j
  \in \Pi \setminus \{i\}}$. 

Recall from Section \ref{subsec:consensusDef} that in any sequence $t$
over $I_P \cup O_P$, if an event $decide(v)_i$ occurs, then $v$ is
said to be a decision value of $t$. We extend this definition to
arbitrary sequences; for any sequence $t$, if $t$ contains an element
$decide(v)_i$ (where $v \in \set{0,1}$ and $i \in \Pi$), then $v$ is
said to be a \emph{decision value} of $t$.

The next Lemma follows immediately from Theorem \ref{thm:fairBranchFairExec} and Lemma \ref{lem:fairExecExactlyOneDecision}.
\begin{lemma}\label{prop:fairBranchOneDecision}
 For each fair branch $b$ in $\mathcal{R}^{G}$, $exe(b)$ has exactly one decision value.
\end{lemma}
\begin{proof}
Fix a fair branch $b$ in $\mathcal{R}^{G}$. Invoking Theorem \ref{thm:fairBranchFairExec}, we know that there exists a fair trace $t_\mathcal{S}$ of $\mathcal{S}$ such that $trace(b) = t_\mathcal{S}|_{act(\mathcal{S}) \setminus \hat{I}}$ and $t_\mathcal{S}|_{\hat{I} \cup O_D} \in T_D$. Let $\alpha_\mathcal{S}$ be an execution of $\mathcal{S}$ whose trace is $t_\mathcal{S}$, and let $t'_D = t_\mathcal{S}|_{\hat{I} \cup O_D}$. Since $trace(b) = t_\mathcal{S}|_{act(\mathcal{S}) \setminus \hat{I}}$, we know that $t'_D|_{O_D} = trace(b)|_{O_D} = exe(b)|_{O_D}$. Invoking Lemma \ref{lem:fairBranchUnfairExe}, we know that $exe(b)|_{O_D}$ is the event-sequence of some fair branch in $G$. Therefore, $t'_D|_{O_D}$ is the event-sequence of some fair branch in $G$.
 
Since at most $f$ locations are not live in $G$, there are at most $f$ locations $i$ such that $t'_D$ has only finitely many events from $O_{D,i}$. Since $t'_D \in T_D$, we know that $t'_D$ has at most $f$ locations that are not live in $t'_D$. Recall that $t'_D = t_\mathcal{S}|_{\hat{I} \cup O_D}$, and therefore, there are at most $f$ locations that are not live in $\alpha_\mathcal{S}$. In other words,  $\alpha_\mathcal{S}|_{I_P \cup O_P}$
satisfies $f$-crash-limitation. Thus, invoking Lemma \ref{lem:fairExecExactlyOneDecision}, we know that $\alpha_\mathcal{S}|_{I_P \cup O_P}$ has exactly one decision value. Since $trace(b) = \alpha_\mathcal{S}|_{act(\mathcal{S}) \setminus \hat{I}}$, we know that $trace(b)$, and therefore $exe(b)$, has exactly one decision value.
\end{proof}
%
\begin{lemma}\label{prop:eachNodeAtMostOneDecision}
 For each node $N$ in $\mathcal{R}^{G}$, $exe(N)$ has at most one decision value.
\end{lemma}
\begin{proof}
Fix node $N$ of $\mathcal{R}^G$. Fix $b$ to be a fair branch that contains node $N$. By construction $exe(N)$ is a prefix of $exe(b)$. Invoking Lemma \ref{prop:fairBranchOneDecision} yields that $exe(b)$ has exactly one decision value. Therefore, $exe(N)$ must have at most one decision value.
\end{proof}

Fix a convergent sequence of finite observations $G_1, G_2, G_3, \ldots$ that converge to $G$; that is, for each positive integer $x$, $G_x$ is a prefix of $G_{x+1}$, and $\lim_{x \rightarrow \infty} G_x = G$. Construct the sequence of trees $\mathcal{R}^{G_1}, \mathcal{R}^{G_2}, \ldots$ for system $\mathcal{S}$

\begin{lemma}\label{prop:eachNodeAtMostOneDecisionInAllGs}
 For each observation $G' \in \set{G, G_1, G_2, \ldots }$, for each node $N$ in $\mathcal{R}^{G'}$, $exe(N)$ has at most one decision value.
\end{lemma}
\begin{proof}
Fix an observation $G'$ and a node $N$ as in the hypothesis of the lemma. For contraction, assume that $exe(N)$ has more than one decision value.

Recall that every observation in $\set{G, G_1, G_2, \ldots }$ is a prefix of $G$, and therefore, $G'$ is a prefix of $G$. Therefore, by Lemma \ref{cor:NonBotExistsFromPrefix}, we know that $\mathcal{R}^G$ has a node $N_G$ such that $exe(N) = exe(N_G)$. Since $exe(N)$ has more than one decision value, $exe(N_G)$ must also have more than one decision value. However, this contradicts Lemma \ref{prop:eachNodeAtMostOneDecision}.
\end{proof}

\subsection{Valence}\label{subsec:valence}

For any arbitrary observation $G' \in \set{G, G_1, G_2, \ldots }$ and any arbitrary node $N$ in $\mathcal{R}^{G'}$, we define the notion of ``valence" as follows. From Lemma
\ref{prop:nodeFiniteExe}, we know that $exe(N)$ is a finite execution
of system $\mathcal{S}$. Node $N$ is said to be \emph{bivalent} in $\mathcal{R}^{G'}$ if
there exist two descendants $N_0$ and $N_1$ of $N$ such that
$exe(N_0)$ has a decision value $0$ and $exe(N_1)$ has a decision
value $1$; recall from Lemma
\ref{prop:eachNodeAtMostOneDecisionInAllGs} that every node has at most one
decision value.  Similarly, $N$ is said to be $v$-valent in $\mathcal{R}^{G'}$ if there exists a descendant
$N_v$ of $N$ such that $v$ is a decision value of $exe(N_v)$, and for
every descendant $N_{v'}$ of $N$, it is not the case that $1-v$ is a
decision value of $exe(N_{v'})$. If $N$ is either $0$-valent or
$1$-valent, then it is said to be \emph{univalent}.

\subsubsection{Valence of nodes in execution trees of $\set{G,G_1,G_2,\ldots}$}
Here we show the following properties related to valence for any arbitrary observation $G' \in \set{G, G_1, G_2, \ldots }$ and any arbitrary node $N$ in $\mathcal{R}^{G'}$. If $N$ is bivalent in $\mathcal{R}^{G'}$, then it does not have a decision value (Lemma \ref{prop:bivalentNoDecisionValue}). If a non-$\bot$ node $N$ is bivalent in $\mathcal{R}^{G_x}$ for some $x$, then for all $x' \geq x$, $N$ remains a non-$\bot$ bivalent node in $\mathcal{R}^{G_{x'}}$ and in $\mathcal{R}^{G}$ (Lemma \ref{lem:bivalentNodeStableInNextG} and Corollaries  \ref{cor:bivalentNodeStableInNextG} and \ref{cor:bivalentNodeStableInAllGs}). If a non-$\bot$ node $N$ is bivalent in $\mathcal{R}^G$, then for some positive integer $x$ and all $x' \geq x$, $N$ remains a non-$\bot$ bivalent node in $\mathcal{R}^{G_{x'}}$ (Lemma \ref{lem:bivalentNodeFiniteTimeInAllGs} and Corollary \ref{cor:bivalentNodeFiniteTimeInAllGs}). Finally, if a non-$\bot$ node $N$ is univalent in $\mathcal{R}^G$, then for some positive integer $x$ and all $x' \geq x$, $N$ remains a non-$\bot$ univalent node in $\mathcal{R}^{G_{x'}}$ (Lemma \ref{lem:univalentNodeFiniteTimeInAllGs}).

\begin{lemma}\label{prop:bivalentNoDecisionValue}
 Fix $G'$ to be an arbitrary observation in $\set{G, G_1, G_2, \ldots }$. Then, for every bivalent node $N$ in $\mathcal{R}^{G'}$, $exe(N)$ does not have a decision value in $\mathcal{R}^{G'}$.
\end{lemma}
\begin{proof}
Fix $N$ be a bivalent node in $\mathcal{R}^{G'}$. By Lemma
\ref{prop:eachNodeAtMostOneDecisionInAllGs}, $exe(N)$ has at most one
decision value. For contradiction, let $exe(N)$ have a decision value
(say) $v$. Then, every descendant $\hat{N}$ of $N$ also has exactly one decision
value $v$. However, since $N$ is bivalent, some descendant $\hat{N}$ of $N$ must
have a decision value $1-v$. Thus, we have a contradiction.
\end{proof}

Applying Lemma \ref{lem:nodePersistsForEver} to the sequence $G_1,G_2,\ldots$, we conclude the following. For each positive integer $x$, for each non-$\bot$ node $N$ in $\mathcal{R}^{G_x}$, for each positive integer $x'\geq x$, $\mathcal{R}^{G_{x'}}$ contains node $N$.

\begin{lemma}\label{lem:bivalentNodeStableInNextG}
For each positive integer $x$, if a non-$\bot$ node $N$ is bivalent in $\mathcal{R}^{G_x}$, then  node $N$ in $\mathcal{R}^{G_{x+1}}$ is a non-$\bot$ node and is bivalent.
\end{lemma}
\begin{proof}
Fix $x$ and $N$ as in the hypotheses of the lemma.  
Since $N$ is bivalent, there exists some descendant $N_1$ of $N$ in $\mathcal{R}^{G_x}$ such that the decision value of $exe(N_1)$ is $1$, and there exists some descendant $N_0$ of $N$ in $\mathcal{R}^{G_x}$ such that the decision value of $exe(N_0)$ is $0$. Applying Lemma \ref{cor:NonBotNodesExist}, we know that there exist descendants $N_{\not\bot.1}$ and $N_{\not\bot.0}$ such that decision value of $exe(N_{\not\bot.1})$ is $1$, and the decision value of $exe(N_{\not\bot.0})$ is $0$, and in the paths from $N$ to $N_{\not\bot.1}$ and from $N$ to $N_{\not\bot.0}$, there is no edge whose action tag is $\bot$. In other words, $N_{\not\bot.1}$ and $N_{\not\bot.0}$ are non-$\bot$ nodes. 

Thus, in $\mathcal{R}^{G_x}$, the path from root to $N$, from root to $N_{\not\bot.1}$, and from root to $N_{\not\bot.0}$ does not contain any edge whose action tag is $\bot$. Recall that $G_x$ is a prefix of $G_{x+1}$. 
Applying Lemma \ref{lem:prefixObsYieldsSubtree}, we know that $\mathcal{R}^{G_{x+1}}$ contains the non-$\bot$ nodes $N$, $N_{\not\bot.1}$, and $N_{\not\bot.0}$. That is, node $N$ in $\mathcal{R}^{G_{x+1}}$ is bivalent.
\end{proof}

\begin{corollary}\label{cor:bivalentNodeStableInNextG}
For each positive integer $x$, if a non-$\bot$ node $N$ is bivalent in $\mathcal{R}^{G_x}$, then for all positive integers $x'\geq x$, node $N$ in $\mathcal{R}^{G_{x'}}$ is a non-$\bot$ node and is bivalent.
\end{corollary}
\begin{proof}
The Corollary is an inductive extension of Lemma \ref{lem:bivalentNodeStableInNextG}, where the induction is on $x'-x$.
\end{proof}

\begin{corollary}\label{cor:bivalentNodeStableInAllGs}
For each positive integer $x$, if a non-$\bot$ node $N$ is bivalent in $\mathcal{R}^{G_x}$, then $N$ is a non-$\bot$ node and is bivalent in $\mathcal{R}^{G}$.
\end{corollary}

\begin{lemma}\label{lem:bivalentNodeFiniteTimeInAllGs}
If a non-$\bot$ node $N$ is bivalent in $\mathcal{R}^{G}$, then there exists a positive integer $x$ such that $N$ is a non-$\bot$ node and is bivalent in $\mathcal{R}^{G_x}$.
\end{lemma}
\begin{proof}
Fix $N$ as in the hypotheses of the lemma. Since $N$ is bivalent in $\mathcal{R}^{G}$, there exist descendants $N_0$ and $N_1$ of $N$ such that $exe(N_0)$ has a decision value $0$ and $exe(N_1)$ has a decision value $1$ in $\mathcal{R}^{G}$. 

Let $d_0$ be the depth of node $N_{0}$ in $\mathcal{R}^{G}$, and let $d_1$ be the depth of node $N_{1}$ in $\mathcal{R}^{G}$. Let $d$ denote $\max(d_0,d_1)$. Since $G$ is the limit of $G_x$ as $x$ tends to $\infty$, we know that there exists a positive integer $x_1$ such that $G_{x_1}$ contains vertices $v_{N_{0}}$ and $v_{N_{1}}$. Since $i$ is live in $t_D$, there exists a positive integer $x_2$ such that $G_{x_2}$ contains at least $d$ vertices. Let $x$ be $\max(x_1,x_2)$, and therefore,  both $G_{x_1}$ and $G_{x_2}$ are prefixes of $G_x$. Therefore, $G_x$ contains vertices $v_{N_{0}}$ and $v_{N_{1}}$; therefore, $G_x$ contains at least $d$ vertices, and hence, the sequence of distinct non-$(\bot, 0, \bot)$ vertex tags in the paths from the root to $N_0$ and from the root to $N_1$ in $\mathcal{R}^{G}$ is also a path in $G_x$. By Lemma \ref{lem:superObsYieldsSuperTree}, we know that $\mathcal{R}^{G_x}$ contains nodes $N$, $N_{1}$, and $N_{0}$. Furthermore, we conclude that $N$ is bivalent in $\mathcal{R}^{G_x}$.
\end{proof}

\begin{corollary}\label{cor:bivalentNodeFiniteTimeInAllGs}
For each non-$\bot$ bivalent node $N$ in $\mathcal{R}^{G}$, there exists a positive integer $x$ such that for all positive integers $x' \geq x$, node $N$ is non-$\bot$ bivalent in $\mathcal{R}^G_{x'}$.
\end{corollary}
\begin{proof}
Fix $N$ as in the hypothesis of the corollary. From Lemma \ref{lem:bivalentNodeFiniteTimeInAllGs} , we know that there exists a positive integer $x$ such that $N$ is a non-$\bot$ bivalent node in $\mathcal{R}^G_x$. For any $x' \geq x$, we know that 
$G_x$ is a prefix of $G_{x'}$.
Applying Lemma \ref{lem:prefixObsYieldsSubtree}, we conclude that $N$ is a non-$\bot$ bivalent node in $\mathcal{R}^G_{x'}$.
\end{proof}

\begin{lemma}\label{lem:univalentNodeFiniteTimeInAllGs}
If a node $N$ is univalent in $\mathcal{R}^{G}$, then there exists a positive integer $x$ such that for all positive integers $x' \geq x$, node $N$ is univalent in $\mathcal{R}^{G_{x'}}$.
\end{lemma}
\begin{proof}
Fix $N$ as in the hypotheses of the lemma. Let $N$ be $c$-valent for some $c \in \set{0,1}$.
Let $d$ be the smallest positive integer such that there exists some some descendant $N_c$ of $N$ in $\mathcal{R}^{G}$ such that $N_c$ is at depth $d$ and $exe(N_c)$ has a decision value $c$. Since $N$ is $c$-valent, we know that $d$ exists. 

Let $x$ be the smallest positive integer such that the following is true. (1) $G_x$ contains the vertices $v_N$ and $v_{N_{c}}$. (2) For each location $j$ that is live in $t_D$, $G_x$ contains at least $d$ vertices whose location is $j$. (3) For each location $j$ that is not live in $t_D$, the set of vertices of $G_x$ whose location is $j$ is identical to the set of vertices of $G$ whose location is $j$. Therefore, the sequence of distinct non-$(\bot, 0, \bot)$ vertex tags in the paths from the root to $N_c$ is also a path in $G_x$.

Fix a positive integer $x'\geq x$. Recall that $G_x$ is a prefix of $G_{x'}$, and invoking Lemma \ref{lem:superObsYieldsSuperTree}, we know that $\mathcal{R}^{G_{x'}}$ contains nodes $N$, and $N_{c}$.

Note that since $N$ is $c$-valent in $\mathcal{R}^{G}$, there exists no descendant $N'$ of $N$ such that $exe(N')$ has a decision value $(1-c)$. By the contrapositive of  Corollary \ref{cor:prefixObsPrefixExe}, we know that $\mathcal{R}^{G_{x'}}$ does not contain any descendant $N'$ of $N$ such that $exe(N')$ has a decision value $(1-c)$. By definition, $N$ is $c$-valent in $\mathcal{R}^{G_{x'}}$.
\end{proof}

\subsubsection{Valence of nodes in $\mathcal{R}^{G}$}

Now consider only the viable observation $G$. For every fair branch $b$ in $\mathcal{R}^G$, we know from Lemma
\ref{prop:fairBranchOneDecision} that $exe(b)$ has exactly one
decision value. Since every node $N$ is a node in some fair branch
$b$, we conclude the following.

\begin{lemma}\label{prop:everyNodeIsBiOrUnivalent}
 Every node $N$ in $\mathcal{R}^G$ is either bivalent or univalent.
\end{lemma}
%


\begin{lemma}\label{prop:initialBivalent}
 The root node $\top$, of $\mathcal{R}^G$, is bivalent.
\end{lemma}
\begin{proof}
Let $\Pi = {i_1,i_2,\ldots,i_n}$.
Note that by construction there exists a path $p_0 = Env_{i_1,0},Env_{i_2,0},\ldots,Env_{i_n,0}$ of edges from $\top$. Let $b_0$ be a fair path that contains $p_0$ as its prefix. By Lemma \ref{prop:fairBranchOneDecision}, we know $b_0$ contains a single decision value. By Theorem \ref{thm:fairBranchFairExec}, we know that there exists a fair trace $t_{0,\mathcal{S}}$ of $\mathcal{S}$ such that  $trace(b) = t_{0,\mathcal{S}}|_{act(\mathcal{S}) \setminus \hat{I}}$. By the validity property we know that the decision value of $trace(b_0)$ must be $0$.

Similar to the above construction, there exists a path $p_1 = Env_{i_1,1},Env_{i_2,1},\ldots,Env_{i_n,1}$ of edges from $\top$. Let $b_1$ be a fair path whose prefix is $p_1$. By Lemma \ref{prop:fairBranchOneDecision}, we know $b_1$ contains a single decision value. By Theorem \ref{thm:fairBranchFairExec}, we know that there exists a fair trace $t_{1,\mathcal{S}}$ of $\mathcal{S}$ such that  $trace(b_1) = t_{1,\mathcal{S}}|_{act(\mathcal{S}) \setminus \hat{I}}$. By the validity property we know that the decision value of $trace(b_1)$ must be $1$.

In other words, $\top$ contains two nodes $N_0$ (in $b_0$) and $N_1$ ($b_1$) such that $exe(N_0)$ has a decision value $0$ and $exe(N_1)$ has a decision value $1$. By definition, $\top$ is bivalent.
\end{proof}
Based on the properties of the $f$-crash-tolerant binary consensus
problem, we have the following lemma.

\begin{lemma}\label{lem:univalentDescendentIsUnivalent}
For each node $N$ in $\mathcal{R}^{G}$, if $N$ is $v$-valent, then for
every descendant $\hat{N}$ of $N$, $\hat{N}$ is also $v$-valent.
\end{lemma}
\begin{proof}
Fix $N$ and $v$ as in the hypothesis of the lemma. Let $\hat{N}$ be an arbitrary descendant of $N$. By construction, every descendant of $\hat{N}$ is also a descendant of $N$. Since $N$ is $v$-valent, for every descendant $N'$ of $N$, it not the case that $1-v$ is the decision value of $v_{N'}$; therefore, for every descendant $N'$ of $\hat{N}$, it not the case that $1-v$ is the decision value of $v_{N'}$. Fix some fair branch $b$ in $\mathcal{R}^G$ that contains the node $\hat{N}$. By Lemma \ref{prop:fairBranchOneDecision}, we know that $exe(b)$ has exactly one decision value. Let $N''$ be a node in $b$ that occurs after $\hat{N}$ such that $exe(N'')$ has a decision value. We have already established that this decision value cannot be $1-v$; therefore the decision value must be $v$. In other words, $\hat{N}$ is $v$-valent.
\end{proof}

\subsection{Gadgets}
Consider the system $\mathcal{S}$, which
consists of a distributed algorithm $A$, the channel automata, and the environment automaton $\mathcal{E}_C$ such that solves $f$-crash-tolerant
consensus using $D$ in $\mathcal{E}_C$. In this section, we define ``gadgets" and ``decision gadgets", which are structures within $\mathcal{R}^G$ that demonstrate how executions
of a system $\mathcal{S}$ evolve from being bivalent to becoming
univalent.

 
 A \emph{gadget} is a tuple of the form $(N,l,E^l,E^{\prime \ell})$ or $(N,l,r,E^l,E^r,E^{rl})$, where $N$ is a node, $l$ and $r$ are distinct labels, $E^l$, $E^{\prime \ell}$, $E^r$, and $E^{rl}$ are edges, such that the following properties are satisfied.
\begin{enumerate}
\item $E^l$ and $E^{\prime \ell}$ are $l$-edges of $N$.
 \item $E^r$ is an $r$-edge of $N$.
 \item $E^{rl}$ is an $l$-edge of $N^r$, where $N^r$ is the node to which $E^r$ is the incoming edge.
\end{enumerate}

Let $Y$ be a decision gadget; $Y$, which is either of the form $(N,l,E^l,E^{\prime \ell})$ or of the form $(N,l,r,E^l,E^r,E^{rl})$, said to be a \emph{non-$\bot$ gadget} if $N$ is a non-$\bot$ node.\footnote{Recall that a node $N$ is a non-$\bot$ node iff the path from $\top$ to node $N$ in $\mathcal{R}^G$ does not contain an edge whose action tag is $\bot$.}

 A gadget is said to be a \emph{decision gadget} iff the gadget is either a ``fork'' or a ``hook'': Section \ref{subsubsec:ForkProperties} defines a ``fork'' and establishes properties of a fork, Section \ref{subsubsec:HookProperties} defined a ``hook'' and the establishes properties of a hook.   In both cases, we show that a decision gadget must have what we call a ``critical location", which is guaranteed to be live in G.





\subsubsection{Forks}
\label{subsubsec:ForkProperties}

In the tree $\mathcal{R}^{G}$, a \emph{fork} is a gadget $(N,l,E^l,E^{\prime \ell})$ such that the following are true.
\begin{enumerate}
 \item $N$ is bivalent.
 \item For some $v \in \set{0,1}$, the lower endpoint $N^l$ of $E^l$ is $v$-valent and the lower endpoint $N^{\prime \ell}$ of $E^{\prime \ell}$ is $(1-v)$-valent.
 \end{enumerate}
 
 \begin{lemma}\label{lem:forkTaskFD}
For every fork $(N,l,E^l,E^{\prime \ell})$ in $\mathcal{R}^G$, $l \in \set{FD_j|j\in \Pi}$.
\end{lemma}
\begin{proof}
Fix a fork $(N,l,E^l,E^{\prime \ell})$ in $\mathcal{R}^G$. 
From the construction of $\mathcal{R}^G$, we know that for each label $l'$ in $T$, node $N$ has exactly one $l'$-edge. For each label $l'$ in $\set{FD_j|j\in \Pi}$, node $N$ has at least one $l'$-edge. Therefore, $l \in \set{FD_j|j\in \Pi}$.
\end{proof}

Any fork $(N,l,E^l,E^{\prime \ell})$ in $\mathcal{R}^G$ satisfies three properties: (1) the action tags $a_{E^l}$ and $a_{E^{\prime \ell}}$ are not $\bot$, (2) the locations of the action tags $a_{E^l}$ and $a_{E^{\prime \ell}}$ are the same location (say) $i$, and (3) location $i$, called the \emph{critical location} of the hook, must be live in $G$. We prove each property separately.

For the remainder of this subsection fix a fork $(N,l,E^l,E^{\prime \ell})$ from $\mathcal{R}^G$,;we use the following convention from the definition of a fork: $N^l$ denotes the $l$-child of $N$ connected by the edge $E^l$, and $N^{\prime \ell}$ denotes the $l$-child of $N$ connected by the edge $E^{\prime \ell}$.

\begin{lemma}\label{lem:forkEventTagsNotBot}
The action tags $a_{E^l}$ and $a_{E^{\prime \ell}}$ are not $\bot$.
%
\end{lemma}
\begin{proof}
Without loss of generality, assume, for contradiction, that the action tag $a_{E^l}$ is $\bot$. 
From Lemma \ref{lem:forkTaskFD}, we know $l \in \set{FD_j|j\in \Pi}$; fix a location $i$ such that $l = FD_i$. 
From the definition of a fork we know that $N$ has at least two $FD_i$ edges. 
From the construction of $\mathcal{R}^G$, we know that an $FD_i$-edge of $N$ has an action tag $\bot$ iff either $G$ has no vertices whose location is $i$ or $v_N$ has no outgoing edge in $G$ to a vertex whose location is $i$. In both cases, $N$ has exactly one $FD_i$ edge. However, this contradicts our earlier conclusion that $N$ has at least two $FD_i$ edges.
\end{proof}

\begin{lemma}\label{thm:forkSameCriticalProcess}
The locations of the action tags $a_{E^l}$ and $a_{E^{\prime \ell}}$ are the same.
\end{lemma}
\begin{proof}
Note that for any label $l'$, the actions associated with $l'$ occur in a single location. Since $E^l$ and $E^{\prime \ell}$ have the same label $l$, and from Lemma \ref{lem:forkEventTagsNotBot} we know that the action tags $a_{E^l}$ and $a_{E^{\prime \ell}}$ are not $\bot$, we conclude that the location of $a_{E^l}$ and $a_{E^{\prime \ell}}$ must be the same location.
\end{proof}

Next, we present the third property of a fork. Before stating this property, we have to define a \emph{critical location} of a fork. The \emph{critical location} of the fork $(N,l,E^l,E^{\prime \ell})$ is the location of $a_{E^l}$ and $a_{E^{\prime \ell}}$; from Lemma \ref{thm:forkSameCriticalProcess}, we know that this is well-defined.

Next, we show that the critical location of the fork $(N,l,E^l,E^{\prime \ell})$ must be live. We use the following helper lemma.

\begin{lemma}\label{lem:forkSimilarModuloCriticalLocation}
$N^{l} \sim_i N^{\prime \ell}$, where $i$ is the critical location of $(N,l,E^l,E^{\prime \ell})$.
\end{lemma}
\begin{proof}
By construction, the following is true of states of automata in system $\mathcal{S}$. For each location $x \in \Pi \setminus \set{i}$, the state of the process automaton $A_x$ is the same in states $c_{N^l}$ and $c_{N^{\prime \ell}}$; similarly, the state of the environment automaton $\mathcal{E}_{C,x}$ is the same in states $c_{N^l}$ and $c_{N^{\prime \ell}}$. For every pair of distinct locations $x,y \in \Pi$, the state of the channel automaton $Chan_{x,y}$ is the same in states $c_{N^l}$ and $c_{N^{\prime \ell}}$. Therefore, we conclude that $N^{l} \sim_i N^{\prime \ell}$.
\end{proof}

\begin{lemma}\label{thm:forkCriticalProcessCorrect}
The critical location of $(N,l,E^l,E^{\prime \ell})$ is in $live(G)$.
\end{lemma}
\begin{proof}
 Let $i$ be the critical location of $(N,l,E^l,E^{\prime \ell})$. Applying Lemma \ref{lem:forkTaskFD} we conclude that $l$ is $FD_i$. Since $N^l$ and $N^{\prime \ell}$ are $l$-children of $N$, we note that the states of all automata in system $\mathcal{S}$ in states $c_{N^l}$ and $c_{N^{\prime \ell}}$ are the same, except for the state of the process automaton at $i$. Recall that $v_{N^l}$ and $v_{N^{\prime \ell}}$ are the vertex tags of $N^l$ and $N^{\prime \ell}$, respectively. From Lemma \ref{lem:forkEventTagsNotBot} we know that the action tags $a_{E^l}$ and $a_{E^{\prime \ell}}$ are not $\bot$. Therefore, $v_{N^l}$ and $v_{N^{\prime \ell}}$ are vertices in $G$.
Note that $N^l$ is $v$-valent for some $v\in\{0,1\}$ and $N^{\prime \ell}$ is $(1-v)$-valent. In order to show that $i$ is in $live(G)$, we have to show that $G$ contains infinitely many vertices whose location is $i$.

For contradiction assume that the critical location $i$ of $(N,l,E^l,E^{\prime \ell})$ is not in $live(G)$. Then by definition, $G$ contains only finitely many vertices whose location is $i$. 
Recall that $G$ is a viable observation of $D$ such that at most $f$ locations are not live in $G$. Since $f < n$,
we conclude that at least one location is live in $G$. Fix such a location $j$.

From Lemma \ref{prop:outgoingEdgesToCrashedVertices} we know that there exists a positive integer $k$ such that for every positive integer $k' \geq k$, there is no edge from any vertex of the form $(j,k',*)$ to any vertex whose location is $i$. Fix such a positive integer $k$, and fix the corresponding vertex $(j,k,*)$.

From Lemma \ref{prop:outgoingEdgesToLiveVertices}, we know that there exists a positive integer $k' \geq k$ such that there are outgoing edges from $v_{N^l}$ and $v_{N^{\prime \ell}}$ to a vertex $(j,k',*)$; fix such a vertex $v_1 = (j,k',e')$.

From the construction of $\mathcal{R}^G$, we know that there exist $FD_j$-edges $E^{FD_j}$ and $E^{\prime FD_j}$ from $N^l$ and $N^{\prime \ell}$, respectively, whose action tag is $e'$ and vertex tag is $v_1$. Let $N^{FD_j}$ and $N^{\prime FD_j}$ be the $FD_j$-children of $N^l$ and $N^{\prime \ell}$, respectively, connected to their parent by edges $E^{FD_j}$ and $E^{\prime FD_j}$, respectively. By construction, $v_{N^{FD_j}} = v_{N^{\prime FD_j}} = v_1$. 


By Lemma \ref{lem:forkSimilarModuloCriticalLocation} we know that $N^{l} \sim_i N^{\prime \ell}$. Since the action tags of $E^{FD_j}$ and $E^{\prime FD_j}$ are the same, we conclude that the states of all automata in system $\mathcal{S}$ in states $c_{N^{FD_j}}$ and $c_{N^{\prime FD_j}}$ are the same, except for the state of the process automaton $A_i$. Therefore, $N^{FD_j} \sim_i N^{\prime FD_j}$. We have already established that $v_{N^{FD_j}} = v_{N^{\prime FD_j}} = v_1$, and there are no outgoing edges from $v_1$ to vertices whose location is $i$. Thus, by definition, $N^{FD_j}$ and $N^{\prime FD_j}$ are post-$crash_i$ nodes\footnote{Recall from Section \ref{subsec:PropertiesOfSimilarModuloNodes} that a node $N$ is a \emph{post-$crash_i$} node if the following property is satisfied. If $v_N = (\bot, 0 ,\bot)$, then there are no vertices in $G$ whose location is $i$. Otherwise, there are no outgoing edges in $G$ from $v_N$ to any vertex whose location is $i$.}


Recall that $N^l$ is $v$-valent and $N^{\prime \ell}$ is $(1-v)$-valent. Therefore, applying Lemma \ref{lem:univalentDescendentIsUnivalent}, we know that $N^{FD_j}$ is $v$-valent and $N^{\prime FD_j}$  is $(1-v)$-valent. Let $b$ be a fair branch of $\mathcal{R}^G$ that contains nodes $N$, $N^l$ and $N^{FD_j}$.

Since $N$ is bivalent, from Lemma \ref{prop:bivalentNoDecisionValue}, we know that $exe(N)$ does not have a decision value. Since $l$ is $FD_i$, we know that $exe(N^l)$ and $exe(N^{FD_j})$ do not have decision values. From Lemma \ref{prop:fairBranchOneDecision} we know that $exe(b)$ has exactly one decision value, and since $N^{FD_j}$ is $v$-valent, the decision value is $v$. 
 That is, there exists an edge $E^v$ and a node $N^v$ such that $E^v$ occurs in $b$ after $N^{FD_j}$, $a_{E^v}$ is $decide(v)_j$, and $N^v$ is the node that precedes $E^v$ in $b$. 
 
  Since $N^{FD_j}$ and $N^{\prime FD_J}$ are post-$crash_i$ nodes, $N^{FD_j} \sim_i N^{\prime FD_j}$, and $N^v$ is a descendant of $N^{FD_j}$, we apply Theorem \ref{thm:similarModuloDescendant} to conclude that there exists a descendant $N^{(1-v)}$ of $N^{\prime FD_j}$ such that $N^v \sim_i N^{(1-v)}$. From the definition of $\sim_i$ we know that the state of the process automaton at $j$ is the same in $c_{N^v}$ and $c_{N^{(1-v)}}$. Since the action $a_{E^v} = decide(v)_j$ is enabled at the process automaton at $j$ in state $c_{N^v}$, we know that action $decide(v)_j$ is enabled in state $c_{N^{1-v}}$. Therefore, the $Proc_j$-child $N^{\prime (1-v)}$ of $N^{(1-v)}$ has a decision value $v$. However, since $N^{\prime FD_j}$  is $(1-v)$-valent and $N^{(1-v)}$ is a descendant of $N^{\prime FD_j}$, by Lemma \ref{lem:univalentDescendentIsUnivalent}, we know that $N^{(1-v)}$ is $(1-v)$-valent. Thus, we have contradiction.
\end{proof}

\subsubsection{Hooks}
\label{subsubsec:HookProperties}

In the tree $\mathcal{R}^{G}$, a \emph{hook} is a gadget $(N,l,r,E^l,E^r,E^{rl})$ such that the following is true.
\begin{enumerate}
 \item $N$ is bivalent.
 \item For some $v \in \set{0,1}$, the lower endpoint $N^l$ of $E^l$ is $v$-valent and the lower endpoint $N^{rl}$ of $E^{rl}$ is $(1-v)$-valent.  
 \item $a_{E^r} \neq \bot$.
\end{enumerate}

Any hook $(N,l,r,E^l,E^r,E^{rl})$ in $\mathcal{R}^G$ satisfies three properties. (1) the action tags of $a_{E^l}$ and $a_{E^r}$ cannot be $\bot$, (2) the locations of the action tags $a_{E^l}$ and $a_{E^r}$ must be the same location (say) $i$, and (3) location $i$, called the \emph{critical location} of the hook, must be live in $G$. We prove each property separately.

For the remainder of this subsection, fix a hook $(N,l,r,E^l,E^r,E^{rl})$ in $\mathcal{R}^G$; we use the following convention from the definition of a hook: $N^l$ denotes the $l$-child of $N$ connected by the edge $E^l$, $N^r$ denotes the $r$-child of $N$ connected by the edge $E^r$, and $N^{rl}$ denotes the $l$-child of $N^r$ connected by the edge $E^{rl}$.

\begin{lemma}\label{thm:hookEventTagsNotBot}
The action tags $a_{E^l}$ and $a_{E^r}$ are not $\bot$.
\end{lemma}
\begin{proof} 
From the definition of a hook, we know that $a_{E^r} \neq \bot$. It remains to show that $a_{E^l} \neq \bot$.

For contradiction, assume $a_{E^l}$ is $\bot$.
Then, by construction, $c_N=c_{N^l}$ and $v_N = v_{N'}$. Recall that $N$ is bivalent and its descendant $N^{rl}$ is $(1-v)$-valent. From the definition of valence, we know there exists a descendant $N_{(1-v)}$ of $N^{rl}$ (and therefore a descendant of $N$) such that the decision value of $exe(N_{(1-v)})$ is $1-v$.

 Applying Lemma \ref{lem:sameConfigSameExtensionLength} to $N$ and $N^l$, we know that there exists a descendant $N^{l}_{(1-v)}$ of $N^l$ such that the suffix of $exe(N^{l}_{(1-v)})$ following $exe(N^l)$ is identical to the suffix of $exe(N_{(1-v)})$ following $exe(N)$. Since  $exe(N)$ is bivalent, by Lemma \ref{prop:bivalentNoDecisionValue} it does not have a decision value; it follows that some event in the suffix of $exe(N_{(1-v)})$ following $exe(N)$ must be of the form $decide(1-v)_i$ (where $i \in \Pi$). Therefore, the decision value of $exe(N^{l}_{(1-v)})$ is $1-v$. But since $N^l$ is $v$-valent, we have a contradiction. 
\end{proof}

\begin{lemma}\label{thm:sameCriticalProcess}
The locations of the action tags $a_{E^l}$ and $a_{E^r}$ are the same.
\end{lemma}

\begin{proof}
For the purpose of contradiction, we assume that the location $i$ of the action tag $a_{E^l}$ is different from the location $j$ of the action tag $a_{E^r}$; that is, $i \neq j$.
This assumption implies that $l \in \set{FD_i, Proc_i} \cup \set{Chan_{k,i} | k \in \Pi\setminus \{i\}} \cup \set{Env_{v,i}| v \in \{0,1\}}$ and $r \in \set{FD_j, Proc_j} \cup \set{Chan_{k,j}|k \in \Pi\setminus \{j\}} \cup \set{Env_{v,j} | v \in \{0,1\}}$. From Lemma \ref{thm:hookEventTagsNotBot}, we know that $a_{E^l}$ and $a_{E^r}$ are both enabled actions in state $c_N$. 

A simple case analysis for all possible values of $l$ and $r$ (while noting that $i\neq j$) establishes the following. Extending $exe(N)$ by applying $a_{E^l}$ followed by $a_{E^r}$ will yield the same final state as applying $a_{E^r}$, followed by $a_{E^l}$, to $exe(N)$. Intuitively, the reason is that $a_{E^l}$ and $a_{E^r}$ occur at different locations, and therefore, may be applied in either order to $exe(N)$ and result in the same final state. The above observation implies that $N^l$ has an $r$-edge $E^{lr}$ whose action tag $a_{E^{lr}}$ is the action $a_{E^r}$; let $N^{lr}$ be the $r$-child of $N^l$ connected by $E^{lr}$.
Observe that $c_{N^{lr}} = c_{N^{rl}}$ and $v_{N^{lr}} = v_{N^{rl}}$.

Recall that since $(N,l,r,E^l,E^r,E^{rl})$ is a hook, $N^l$ is $v$-valent and $N^{rl}$ is $(1-v)$-valent for some $v \in \set{0,1}$. Since $N^{lr}$ is a descendant of $N^l$, by Lemma \ref{lem:univalentDescendentIsUnivalent}, $N^{lr}$ is also $v$-valent. Let $N^{lr}_{v}$ be a descendant of $N^{lr}$ such that $exe(N^{lr}_{v})$ has a decision value $v$. 
Applying Lemma \ref{lem:sameConfigSameExtensionLength},  we know that there exists a descendant $N^{rl}_{v}$ of $N^{rl}$ such that $c_{N^{lr}_{v}} = c_{N^{rl}_{v}}$ and the suffix of $exe(N^{lr}_{v})$ following $exe(N^{lr})$ is identical to the suffix of $exe(N^{rl}_{v})$ following $exe(N^{rl})$.

Note that since $N$ is bivalent, by Lemma \ref{prop:bivalentNoDecisionValue}, $exe(N)$ has no decision value.

\emph{Claim 1.} $a_{E^l}$ is not a $decide$ action.
\begin{proof}
For contradiction, assume $a_{E^l}$ is a $decide$ action. Since $exe(N^l)$ contains the event $a_{E^l}$ and $exe(N^l)$ is $v$-valent, it follows that $a_{E^l}$ is a $decide(v)$ action. However, recall that $a_{E^{rl}} = a_{E^l}$, $exe(N^{rl})$ contains the event $a_{E^{rl}}$; therefore, $exe(N^{rl})$ contains a $decide(v)$ event. However, $exe(N^{rl})$ is $(1-v)$-valent. Thus, we have a contradiction.
\end{proof}

\emph{Claim 2.} $a_{E^r}$ is not a decide action.
\begin{proof}
Similar to the proof of Claim 1.
\end{proof}

 From Claims 1 and 2, we know that for each of  $N$'s $l$-edge, $N$'s $r$-edge, $N^l$'s $r$-edge, and $N^r$'s $l$-edge, their action tags cannot be a $decide$. Therefore, since $exe(N^{lr}_{v})$ has a decision value $v$, the suffix of $exe(N^{lr}_{v})$ following $exe(N^{lr})$ contains an event of the form $decide(v)$. In other words, the suffix of $exe(N^{rl}_{v})$ following $exe(N^{rl})$ contains an event of the form $decide(v)$. However, this is impossible because $N^{rl}$ is $(1-v)$-valent.
\end{proof}

Next, we present the third property of a hook. Before stating this property, we have to define a \emph{critical location} of a hook. Given the hook $(N,l,r,E^l,E^r,E^{rl})$, the \emph{critical location} of the hook is the location of $a_{E^l}$ and $a_{E^r}$; from Lemma \ref{thm:sameCriticalProcess}, we know that this is well-defined. 

\begin{lemma}\label{thm:criticalProcessCorrect}
The critical location of $(N,l,r,E^l,E^r,E^{rl})$ is in $live(G)$.
\end{lemma}

\begin{proof}
Note that $N^l$ is $v$-valent for some $v\in\{0,1\}$ and $N^{rl}$ is $(1-v)$-valent. 
Let $i$ be the critical location of the hook $(N,l,r,E^l,E^r,E^{rl})$. 
In order to show that $i$ is in $live(G)$, we have to show that $G$ contains infinitely many vertices whose location is $i$.

For the purpose of contradiction, we assume that $G$ contains only finitely many vertices whose location is $i$. Recall that $G$ is a viable observation of $D$ such that at most $f$ locations are not live in $G$. Since $f<n$, we conclude that least one location is live in $G$. Fix such a location $j$.


From Lemma \ref{prop:outgoingEdgesToCrashedVertices} we know that there exists a positive integer $k$ such that for each positive integers $k' \geq k$, there is no edge from any vertex of the form $(j,k',*)$ to any vertex whose location is $i$. Fix such a positive integer $k$, and fix the corresponding vertex $(j,k,*)$.

Next we fix a vertex $v_1$ in $G$ such that, roughly speaking, the event $e''$ of $v_1$ is an event at $j$ and ``occurs'' after the events of $v_N$, $v_{N^l}$, $v_{N^{rl}}$ and after location $i$ is ``crashed''; precisely, $v_1$ is fixed as follows. Let $V'$ be $V \cap \set{v_N, v_{N^l}, v_{N^{rl}}}$; that is, $V'$ is the maximal subset of $\set{v_N, v_{N^l}, v_{N^{rl}}}$ such that each vertex in $V'$ is a vertex  of $G$. If $V'$ is non-empty, then from Lemma \ref{prop:outgoingEdgesToLiveVertices}, we know that there exists a positive integer $k' \geq k$ such that there are outgoing edges from each vertex in $V'$ to a vertex $(j,k',*)$; fix $v_1$ to be such a vertex $(j,k',e')$. If $V'$ is empty, then fix $v_1$ to be any vertex in $V$ of the form $(j,k',e')$, where $k' \geq k$.

From the construction of $\mathcal{R}^G$, we know that there exist $FD_j$-edges $E^{FD_j}$, $E^{l\cdot FD_j}$, and $E^{rl \cdot FD_j}$ from $N$, $N^l$ and $N^{rl}$, respectively, whose action tag is $e'$ and vertex tag is $v_1$. Let $N^{FD_j}$, $N^{l\cdot FD_j}$, and $N^{rl\cdot FD_j}$ be the $FD_j$-children of $N$, $N^l$ and $N^{rl}$, respectively, connected to their parent by edges $E^{FD_j}$, $E^{l \cdot FD_j}$ and $E^{rl \cdot FD_j}$, respectively. By construction, $v_{N^{FD_j}} = v_{N^{l \cdot FD_j}} = v_{N^{rl \cdot FD_j}} = v_1$. See Fig. \ref{fig:extendPathsByFDEdges} for reference. 

Also recall that in $G$ there is no edge from the vertex of the form $(j,k,*)$ to any vertex whose location is $i$, and since $k' \geq k$, we know that is no edge from $v_1$ to any vertex whose location is $i$. Therefore, $N^{FD_j}$, $N^{l\cdot FD_j}$, and $N^{rl\cdot FD_j}$ are \emph{post-$crash_i$} nodes\footnote{Recall from Section \ref{subsec:PropertiesOfSimilarModuloNodes} that a node $N$ is a \emph{post-$crash_i$} node if the following property is satisfied. If $v_N = (\bot, 0 ,\bot)$, then there are no vertices in $G$ whose location is $i$. Otherwise, there are no outgoing edges in $G$ from $v_N$ to any vertex whose location is $i$.}.

\begin{figure}[hptb]
 \centering
\includegraphics[scale=0.4,page=4]{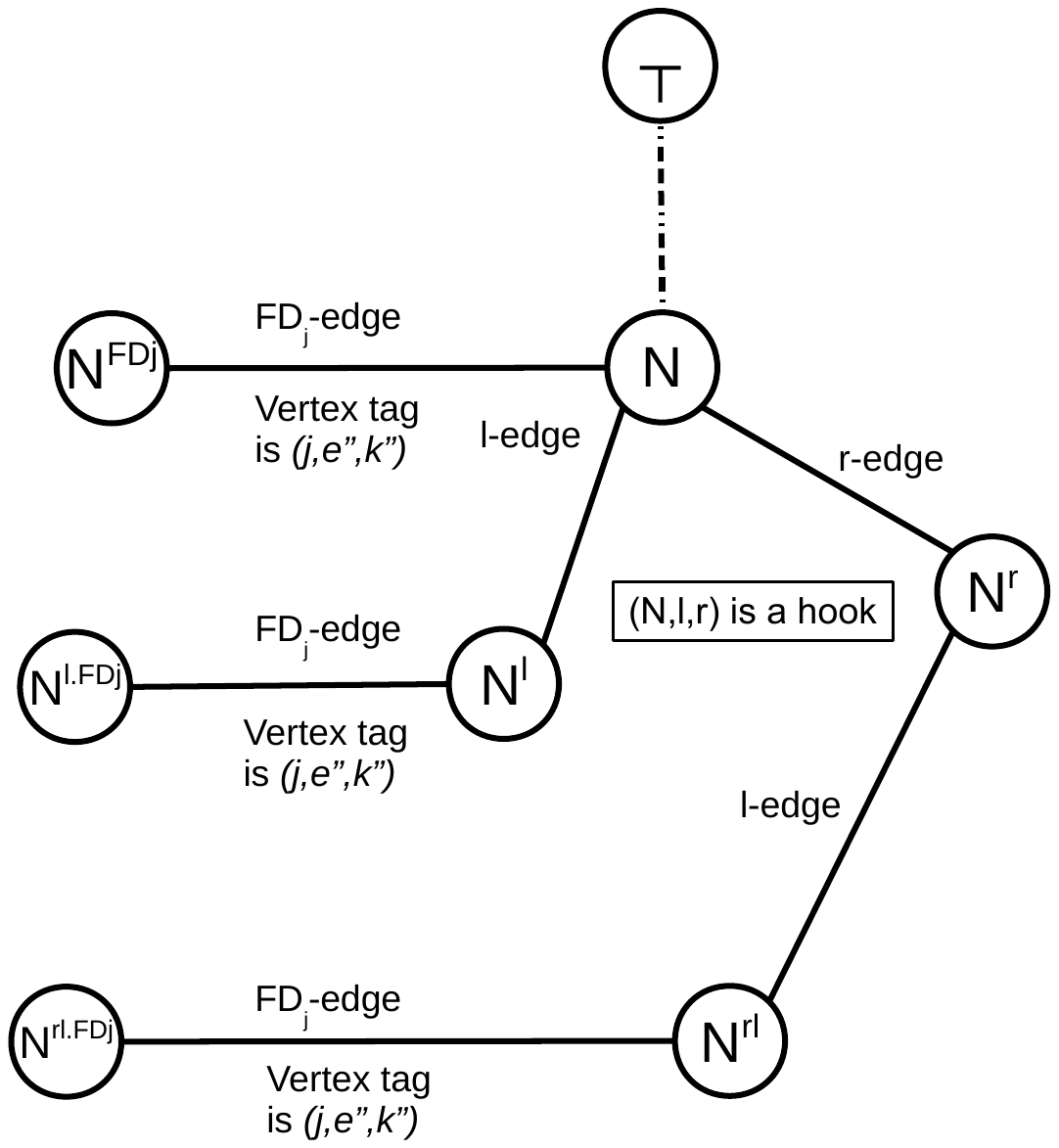}
 \caption{This figure shows how the nodes $N^{FD_j}$, $N^{l\cdot FD_j}$, and $N^{rl\cdot FD_j}$ are determined in the proof of Lemma \ref{thm:criticalProcessCorrect}.}
 \label{fig:extendPathsByFDEdges}
\end{figure}

Note that by construction, the following is true of states of automata in system $\mathcal{S}$. For each location $x \in \Pi \setminus \set{i}$, the state of the process automaton $A_x$ is the same in states $c_{N^{FD_j}}$, $c_{N^{l \cdot FD_j}}$, and $c_{N^{rl \cdot FD_j}}$; similarly, the state of the environment automaton $\mathcal{E}_{C,x}$ is the same in states $c_{N^{FD_j}}$, $c_{N^{l \cdot FD_j}}$, and $c_{N^{rl \cdot FD_j}}$. For every pair of distinct locations $x,y \in \Pi \setminus \set{i}$, the state of the channel automaton $Chan_{x,y}$ is the same in states $c_{N^{FD_j}}$, $c_{N^{l \cdot FD_j}}$, and $c_{N^{rl \cdot FD_j}}$. Finally, for every location $x \in \Pi \setminus \set{i}$, the messages in transit in the channel automaton $Chan_{i,x}$ from $i$ to $x$ in state $c_{N^{FD_j}}$ is a prefix of the messages in transit in $Chan_{i,x}$ in state $c_{N^{l \cdot FD_j}}$ and in state $c_{N^{rl \cdot FD_j}}$. Therefore, we conclude that $N^{FD_j} \sim_i N^{l \cdot FD_j}$ and $N^{FD_j} \sim_i N^{rl \cdot FD_j}$.

Recall that $N^l$ is $v$-valent and $N^{rl}$ is $(1-v)$-valent. Therefore, applying Lemma \ref{lem:univalentDescendentIsUnivalent}, we know that $N^{l\cdot FD_j}$ is $v$-valent and $N^{rl \cdot FD_j}$  is $(1-v)$-valent. Also recall that $N$ is bivalent. 

Let $b$ be a fair branch of $\mathcal{R}^G$ that contains nodes $N$ and $N^{FD_j}$. 
By Lemma \ref{prop:fairBranchOneDecision}, we know that $exe(b)$ has exactly one decision value (say) $v'$; note that either $v'=v$ or $v' = 1-v$. We consider each case.

\emph{Case 1.} $v'=v$. There exists an edge $E_v$ in $b$ such that, the action tag of $E_v$ is $decide(v)_j$. Let $N_v$ be the node preceding $E_v$ in $b$. Note that $N_v$ is descendant of $N^{FD_j}$. 
Recall that $N^{FD_j}$ and $N^{rl\cdot FD_j}$ are \emph{post-$crash_i$} nodes.
By Theorem \ref{thm:similarModuloDescendant}, we know that there exists a descendant $N^{rl}_{v}$ of $N^{rl\cdot FD_j}$ such that $N^v \sim_i N^{rl_v}$.

From the definition of $\sim_i$ we know that the state of the process automaton at $j$ is the same in $c_{N^v}$ and $c_{N^{rl}_{v}}$. Since the action $a_{E_v} = decide(v)_j$ is enabled at the process automaton at $j$ in state $c_{N_v}$, we know that action $decide(v)_j$ is enabled in state $c_{N^{rl}_{v}}$. Therefore, the $Proc_j$-child $N^{\prime rl}_{v}$ of $N^{rl}_{v}$ has a decision value $v$. However, since $N^{rl}$  is $(1-v)$-valent and $N^{rl}_{v}$ is a descendant of $N^{rl}$, by Lemma \ref{lem:univalentDescendentIsUnivalent}, we know that $N^{rl}_{v}$ is $(1-v)$-valent. Thus, we have a contradiction.

\emph{Case 2.} $v' = 1-v$. This is analogous to Case 1 except that we replace  $N^{rl\cdot FD_j}$ with $N^{l\cdot FD_j}$.
\end{proof}

\subsubsection{Decision Gadgets}
Recall that a decision gadget is a gadget that is either a fork or a hook. We have seen that both forks and hooks contain a critical location that must be live in $G$. 
Thus, we have seen that if a tree $\mathcal{R}^G$ contains a decision gadget, then we know that the critical location of that decision gadget must be live in $G$. 

\subsection{Existence of a Decision Gadget} 
\label{subsubsec:existanceOfHook}
The previous subsection demonstrated interesting properties of decision gadgets in $\mathcal{R}^G$. However, it did not demonstrate that $\mathcal{R}^G$, in fact, does contain decision gadgets. We address this here.
Recall that $G$ is viable for $D$, and at most $f$ locations are not live in $G$.  
 

\begin{lemma}\label{lem:bivalentToMonovalent}
There exists a bivalent node $N$ in tree $\mathcal{R}^{G}$ and a label $l$ such that for every descendant $\hat{N}$ of $N$ (including $N$), every $l$-child of $\hat{N}$ is univalent.
\end{lemma}

\begin{proof}
 For contradiction, assume that for every bivalent node $N$ in the tree $\mathcal{R}^{G}$, and every label $l \in L$, there exists a descendant $\hat{N}$ of $N$, such that some $l$-child of $\hat{N}$ is bivalent. 
Therefore, from any bivalent node $N$ in the tree $\mathcal{R}^{G}$, we can choose any label $l$ and find a descendant $\hat{N}'$ of $N$ such that (1) $\hat{N}'$ is bivalent, and (2) the path between $N$ and $\hat{N}'$ contains an edge with label $l$. 

Recall that the $\top$ node is bivalent (Lemma \ref{prop:initialBivalent}).  Applying Lemma \ref{prop:allLabelsExist}, we know that each node in $\mathcal{R}^G$ has an $l$-edge for each label $l \in T \cup \set{FD_i | i \in \Pi}$. Thus, by choosing labels in a round-robin fashion, we can construct a fair branch $b$ starting from the $\top$ node such that every node in that branch is bivalent. Fix such a $b$. We will use $b$ to get a contradiction to the fact that the distributed algorithm $A$ solves $f$-crash-tolerant consensus. using $D$.

By Theorem \ref{thm:fairBranchFairExec}, we know that there exists a fair trace $t_{\mathcal{S}}$ of $\mathcal{S}$ such that $trace(b) = t_{\mathcal{S}}|_{act(S) \setminus \hat{I}}$ and $t_{\mathcal{S}}|_{\hat{I} \cup O_D} \in T_D$. Since $trace(b)|_{O_D} = t_{\mathcal{S}}|_{O_D}$ and $t_{\mathcal{S}}|_{\hat{I} \cup O_D} \in T_D$, we know that at most $f$ locations are not live in $t_{\mathcal{S}}$; therefore, $t_{\mathcal{S}}|_{I_P \cup O_P}$ satisfies $f$-crash limitation. Let $\alpha$ be a fair execution of $\mathcal{S}$ whose trace is $t_{\mathcal{S}}$. Since $t_{\mathcal{S}}|_{I_P \cup O_P}$ satisfies $f$-crash limitation, $\alpha|_{I_P \cup O_P}$ also satisfies $f$-crash limitation. Invoking Lemma \ref{lem:fairExecExactlyOneDecision}, we know that $\alpha$ has exactly one decision value. Since $trace(b) = t_{\mathcal{S}}|_{act(S) \setminus \hat{I}}$, and $t_{\mathcal{S}}$ is the trace of $\alpha$, we know that $trace(b)$ has exactly one decision value. In other words, $exe(b)$ has exactly one decision value.  Therefore, there exists a node $N$ in $b$ such that $exe(N)$ has a decision value. However, this contradicts our conclude that every node in $b$ is bivalent.
\end{proof}

\begin{lemma}\label{lem:bivalentTo1ValentAnd0Valent}
 There exists a bivalent node $N$ in tree $\mathcal{R}^{G}$, a descendant $\hat{N}$ of $N$ (possibly $N$ itself), a label $l$, and $v \in \set{0,1}$ such that (1) 
 for every descendant $\hat{N}'$ of $N$, each $l$-child of $\hat{N}'$ is univalent, (2) some $l$-child of $N$ is $v$-valent, and (3) some $l$-child of $\hat{N}$ is $(1-v)$-valent.
\end{lemma}
\begin{proof}
Invoking Lemma \ref{lem:bivalentToMonovalent}, we fix a pair $(N,l)$ of node $N$ and label $l$ such that (1) $N$ is bivalent, and (2) for every descendant $\hat{N}$ of $N$ (including $N$), every $l$-child of $\hat{N}$ is univalent. Let an $l$-child of $N$ be $v$-valent for some $v \in \set{0,1}$. Since $N$ is bivalent, there must exist some descendant $\hat{N}$ of $N$ such that $exe(\hat{N})$ has a decision value $(1-v)$; that is, $\hat{N}$ is $(1-v)$-valent. By Lemma \ref{lem:univalentDescendentIsUnivalent}, it follows that any $l$-child of $\hat{N}$ is $(1-v)$-valent.
\end{proof}

\begin{lemma}\label{lem:hookExists}
 There exists a bivalent node $N$ such that at least one of the following holds true. (1) There exists a label $l$ and a pair of edges $E^l$ and $E^{\prime \ell}$ such that $(N,l,E^l,E^{\prime \ell})$ is a fork. (2) There exist a pair of labels $l,r$ and edges $E^l$,$E^r$, and $E^{rl}$ such that $(N,l,r,E^l,E^r,E^{rl})$ is a hook. 
\end{lemma}
\begin{proof}
 Applying Lemma \ref{lem:bivalentTo1ValentAnd0Valent}, we know that there exists some node $\tilde{N}$ in tree $\mathcal{R}^{G}$, a descendant $\widehat{\tilde{N}}$ of $\tilde{N}$, and a label $l$ such that (1) $\tilde{N}$ is bivalent, (2) for every descendant $\tilde{N}'$ of $\tilde{N}$, every $l$-child of $\tilde{N}'$ is univalent, (3) some $l$-child of $\tilde{N}$ (denoted $uni(\tilde{N})$)
 is $v$-valent, where $v \in \{0,1\}$, and (4) some $l$-child of $\widehat{\tilde{N}}$  (denoted $uni(\widehat{\tilde{N}})$)
 is $(1-v)$-valent.

Extend the path from $\tilde{N}$ to $\widehat{\tilde{N}}$ to $uni(\widehat{\tilde{N}})$ yielding a path $w$. 
Let $E$ be the first $l$-edge on $w$, let $M$ be the upper endpoint of $E$, and let $M^l$ be the lower endpoint of $E$. Thus, the path from $\tilde{N}$ to $M$ does not contain any $l$-edge.
Note that following:  (1) $uni(\widehat{\tilde{N}})$ is a descendant of $\widehat{\tilde{N}}$ and is $(1-v)$-valent, (2) $\widehat{\tilde{N}}$ is either $M$ or a descendant of $M^l$, and (3) by Lemma \ref{lem:bivalentTo1ValentAnd0Valent}, $M^l$ is univalent. Thus, we conclude that $M^l$ is $(1-v)$-valent. See Figure \ref{fig:hookExistence} for reference.

\begin{figure}[htpb]
 \centering
\includegraphics[scale=0.4,keepaspectratio=true,page=3]{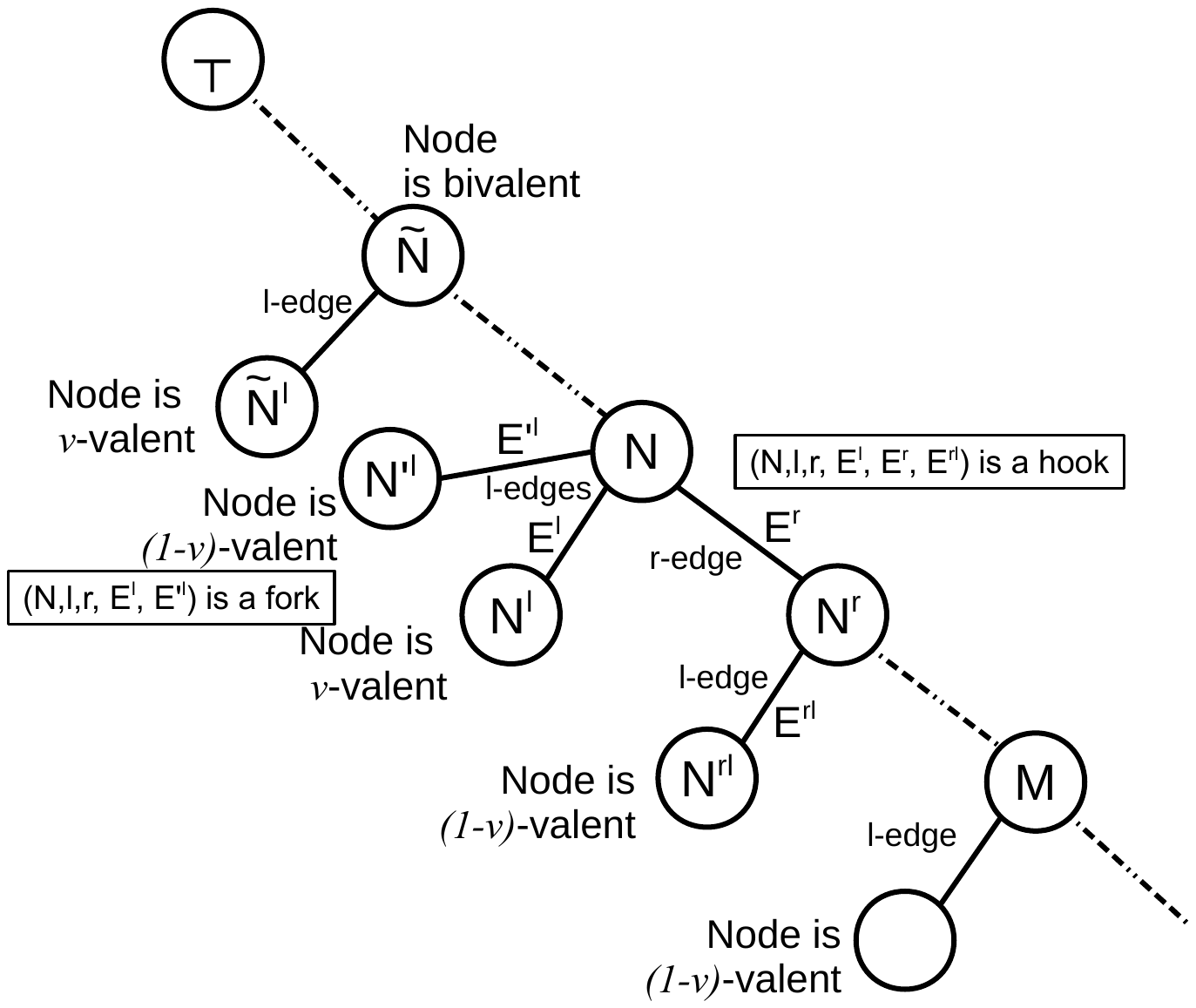}
 \caption{Construction that shows the existence of a ``fork'' or a ``hook'' in the proof for Lemma \ref{lem:hookExists}. 
 }
 \label{fig:hookExistence}
\end{figure}



Note that for each node $N'$ from $\tilde{N}$ to $M$, each $l$-child $N^{\prime l}$ of $N'$ is univalent. 
Recall that $uni(\tilde{N})$, which is an $l$-child of $\tilde{N}$ is $v$-valent and $M^l$, which is an $l$-child of $M$, is $(1-v)$-valent. Therefore, there exists a label $r$ and an $r$-edge $E^r$ from a node $N$ to a node $N^r$ in the path from $\tilde{N}$ to $M$ (inclusive) such that some $l$-child $N^l$ of $N$ is $v$-valent and some $l$-child $N^{rl}$ of $N^r$ is $(1-v)$-valent. Let $E^{rl}$ denote the edge connecting $N^r$ and $N^{rl}$. (See Figure \ref{fig:hookExistence}.)

%
%
%
%


We consider two cases: (1) $a_{E^r} \neq \bot$, and (2)  $a_{E^r} = \bot$.

(1) If $a_{E^r} \neq \bot$, then by definition, $(N,l,r,E^l,E^r,E^{rl})$ is a hook. 

(2) Otherwise, $a_{E^r} = \bot$; therefore, $c_N = c_{N^r}$ and $v_N = v_{N^r}$. Applying Lemma \ref{lem:sameConfigSameChild}, we know that there exists an $l$-child  $N^{\prime \ell}$ of $N$ such that $c_{N^{\prime \ell}} = c_{N^{rl}}$ and $v_{N^{\prime \ell}} = v_{N^{rl}}$. Since $N^{rl}$ is $(1-v)$-valent, $N^{\prime \ell}$ is also $(1-v)$-valent. In other words, $N$ has two $l$-children $N^l$ and $N^{\prime \ell}$, and $N^l$ is $v$-valent and $N^{\prime \ell}$ is $(1-v)$-valent. 
Thus, $(N,l,E^l, E^{\prime \ell})$ is a fork, where $E^{\prime \ell}$ is the edge from $N$ to $N^{\prime \ell}$.
\end{proof}

Thus, we arrive at the main result of this section.

\begin{theorem}\label{thm:ViableGHasDecisionGadgets}
 For every observation $G$ that is viable for $D$ such that $live(G)$ contains at least $n-f$ locations, the directed tree $\mathcal{R}^G$ contains at least one decision gadget. For each decision  gadget in $\mathcal{R}^G$, the critical location of the decision gadget is live in $G$.
\end{theorem}
\begin{proof}
 Fix $G$. From Lemma \ref{lem:hookExists}, we know that $\mathcal{R}^G$ has at least one decision gadget. For each decision gadget that is a fork, from Lemma \ref{thm:forkCriticalProcessCorrect} we know that the critical location of that decision gadget is live in $G$, and for each decision gadget that is a hook, from Lemma \ref{thm:criticalProcessCorrect} we know that the critical location of that decision gadget is live in $G$.
\end{proof}

\begin{theorem}\label{thm:ViableGHasNonBotDecisionGadgets}
For every observation $G$ that is viable for $D$ such that $live(G)$ contains at least $n-f$ locations, the directed tree $\mathcal{R}^G$ contains at least one non-$\bot$ decision gadget.
\end{theorem}
\begin{proof}
Fix $G$. From Theorem \ref{thm:ViableGHasDecisionGadgets} we know that $\mathcal{R}^G$ contains at least one decision gadget. Fix $Y$ to be such a decision gadget. Let node $N$ be the first element in the tuple $Y$. Applying Corollary \ref{cor:NonBotNodesExist}, we know that there exists a non-$\bot$ node $N_{\not\bot}$ such that $exe(N) = exe(N')$, $v_N = v_{N'}$. Applying Lemma \ref{lem:everyDescendantNoBot} to the descendants of $N$ and $N_{\not\bot}$, we know that there exists a non-$\bot$ decision gadget $Y'$ whose first element is $N_{\not\bot}$.
\end{proof}

Theorem \ref{thm:ViableGHasNonBotDecisionGadgets} establishes an important property of any strong-sampling AFD that is sufficient to solve consensus. It demonstrates that in any fair execution of a system that solves consensus using an AFD, some prefix of the execution is bivalent whereas eventually, a longer prefix becomes univalent. The transition from a bivalent to a univalent execution must be the consequence of an event at a \emph{correct} location.

\subsection{Decision gadgets for execution trees in a convergent sequence of observations}
Recall that $G$ is a viable observation of $D$ such that at most $f$ locations are not
live in $G$; $t_D \in T_D$ is a a trace that is compatible with $D$. Finally, $G_1, G_2, G_3, \ldots$ is a sequence of observations that converge to $G$. Next we show the ``persistence'' of non-$\bot$ decision gadgets across the sequence of execution trees $\mathcal{R}^{G_1}, \mathcal{R}^{G_2}, \mathcal{R}^{G_3}, \ldots$.

\begin{lemma}\label{lem:decisionGadgetStableInAllGs}
Let $Y$ be a non-$\bot$ decision gadget in $\mathcal{R}^{G}$. There exists a positive integer $x$ such that for all positive integers $x' \geq x$, $Y$ is a non-$\bot$ decision gadget in $\mathcal{R}^{G_{x'}}$.
\end{lemma}
\begin{proof}
Fix $Y$ to be a non-$\bot$ decision gadget in $\mathcal{R}^{G}$. We consider two cases: (1) $Y$ is a fork, and (2) $Y$ is a hook.

\emph{Case 1.} Let $Y$ be a fork $(N,\ell,E^\ell,E^{\prime \ell})$ in $\mathcal{R}^{G^\infty}$. Let $N^\ell$ be the $\ell$-child of $N$ whose incoming edge is $E^\ell$, and let $N^{\prime \ell}$ be the $\ell$-child of $N$ whose incoming edge is $E^{\prime \ell}$. Let $N^\ell$ be $c$-valent, and let $N^{\prime \ell}$ be $(1-c)$-valent for some $c \in \set{0,1}$. 

Invoking Corollary \ref{cor:bivalentNodeFiniteTimeInAllGs}, we know that there exists a positive integer $x_b$ such that for all $x' \geq x_b$, $N$ is a non-$\bot$ bivalent node in $\mathcal{R}^{G_{x'}}$.
Invoking Lemma \ref{lem:univalentNodeFiniteTimeInAllGs}, we know that there exists a positive integer $x_u$ such that for all $x' \geq x_u$, $N^\ell$ is $c$-valent and  $N^{\prime \ell}$ is $(1-c)$-valent in $\mathcal{R}^{G_{x'}}$.
Let $x = \max(x_b,x_u)$. By construction, for each $x' \geq x$, $Y$ is a non-$\bot$ fork in $\mathcal{R}^{G_{x'}}$.

\emph{Case 2.} Let $Y$ be a hook $(N,\ell,r,E^\ell,E^r,E^{r \ell})$ in $\mathcal{R}^{G^\infty}$. Let $N^\ell$ be the $\ell$-child of $N$ whose incoming edge is $E^\ell$. Let $N^r$ be the $r$-child of $N$ whose incoming edge is $E^r$. Let $N^{r \ell}$ be the $\ell$-child of $N^r$ whose incoming edge is $E^{r \ell}$. Let $N^\ell$ be $c$-valent, and let $N^{r \ell}$ be $(1-c)$-valent for some $c \in \set{0,1}$. 

Invoking Corollary \ref{cor:bivalentNodeFiniteTimeInAllGs}, we know that there exists a positive integer $x_b$ such that for all $x' \geq x_b$, $N$ isa non-$\bot$  bivalent node in $\mathcal{R}^{G_{x'}}$. Invoking Lemma \ref{lem:univalentNodeFiniteTimeInAllGs}, we know that there exists a positive integer $x_u$ such that for all $x' \geq x_u$, $N^\ell$ is $c$-valent and  $N^{r \ell}$ is $(1-c)$-valent in $\mathcal{R}^{G_{x'}}$.
Let $x = \max(x_b,x_u)$. By construction, for each $x' \geq x$, $Y$ is a non-$\bot$ hook in $\mathcal{R}^{G_{x'}}$.
\end{proof}

\begin{lemma}\label{lem:nonDecisionGadgetStableInAllGs}
For each gadget $Y$ in $\mathcal{R}^{G}$ that is not a non-$\bot$ decision gadget, the following is true.
 There exists a positive integer $x$ such that for all positive integers $x' \geq x$, $Y$ is a gadget in $\mathcal{R}^{G_{x'}}$, but $Y$ is not a non-$\bot$ decision gadget in $\mathcal{R}^{G_{x'}}$.
\end{lemma}
\begin{proof}
Fix $Y$ as in the hypotheses of the lemma. Since $Y$ is a gadget in $\mathcal{R}^{G}$, by construction, there exists a positive integer $x_N$ such that for all positive integers $x'_N \geq x_N$, $Y$ is a gadget in $\mathcal{R}^{G_{x'_N}}$. 

We consider two cases: (1) $Y$ is a tuple $(N,\ell,E^\ell,E^{\prime \ell})$, and (2) $Y$ is a tuple $(N,\ell,r,E^\ell,E^r,E^{r\ell})$.

\emph{Case 1.} $Y$ is a tuple $(N,\ell,E^\ell,E^{\prime \ell})$. 
Let $N^\ell$ and $N^{\prime \ell}$ be the nodes to which $E^\ell$ and $N^{\prime \ell}$ are the incoming edges, respectively.
Since $Y$ is not a non-$\bot$ decision gadget, one of the following is true: (1) the path from root to $N$ contains an edge with $\bot$ action tag, (2) $N$ is univalent,  or (3) at least one of $N^\ell$ and $N^{\prime \ell}$ is bivalent in $\mathcal{R}^{G}$. 

If the path from root to $N$ contains an edge with $\bot$ action tag, then by Lemma \ref{lem:superObsYieldsSuperTree}, we know that  exists a positive integer $x_N$ such that for every positive integer $x'_N \geq x_N$, the path from root to $N$ contains an edge with $\bot$ action tag in $\mathcal{R}^{G_{x'_N}}$. Therefore, $Y$ cannot be a non-$\bot$ decision gadget in $\mathcal{R}^{G_{x'_N}}$.

If $N$ is univalent in $\mathcal{R}^{G}$, then by Lemma \ref{lem:univalentNodeFiniteTimeInAllGs}, we know that there exists a positive integer $x_N$ such that for every positive integer $x'_N \geq x_N$, $N$ is univalent in $\mathcal{R}^{G_{x'_N}}$. Therefore, for any positive integer $x'_N \geq x_N$, $Y$ cannot be a decision gadget in $\mathcal{R}^{G_{x'_N}}$.

If $N^\ell$ (or $N^{\prime \ell}$, respectively) is bivalent in $\mathcal{R}^{G^\infty}$, then by Corollary \ref{cor:bivalentNodeFiniteTimeInAllGs}, we know that there a positive integer $x \geq x_N$ such that for all positive integers $x' \geq x$, node $N^\ell$ (or $N^{\prime \ell}$, respectively) is bivalent in $\mathcal{R}^{G_{x'}}$, and consequently, $Y$ is not a decision gadget in $\mathcal{R}^{G_{x'}}$. 

Thus, if $Y$ is a tuple $(N,\ell,E^\ell,E^{\prime \ell})$, then there exists a positive integer $x$ such that for all positive integers $x' \geq x$, $Y$ is a gadget in $\mathcal{R}^{G_{x'}}$, but $Y$ is not a non-$\bot$ decision gadget in $\mathcal{R}^{G_{x'}}$.

\emph{Case 2.} $Y$ is a tuple $(N,\ell,r,E^\ell,E^r,E^{r\ell})$. Let $N^\ell$ be the node to which $E^\ell$ is the incoming edge.
Let $N^{r\ell}$ be the node to which $E^{rl}$ is the incoming edge. Since $Y$ is not a decision gadget, one of the following is true: (1) the path from root to $N$ contains an edge with $\bot$ action tag, (2) $N$ is univalent,  or (3) at least one of $N^\ell$ and $N^{\prime \ell}$ is bivalent in $\mathcal{R}^{G}$.

If the path from root to $N$ contains an edge with $\bot$ action tag, then by Lemma \ref{lem:superObsYieldsSuperTree}, we know that  exists a positive integer $x_N$ such that for every positive integer $x'_N \geq x_N$, the path from root to $N$ contains an edge with $\bot$ action tag in $\mathcal{R}^{G_{x'_N}}$. Therefore, $Y$ cannot be a non-$\bot$ decision gadget in $\mathcal{R}^{G_{x'_N}}$.

If $N$ is univalent in $\mathcal{R}^{G}$, then as in Case 1, by Lemma \ref{lem:univalentNodeFiniteTimeInAllGs}, we know that there exists a positive integer $x_N$ such that for every positive integer $x'_N \geq x_N$, $N$ is univalent in $\mathcal{R}^{G_{x'_N}}$. Therefore, for any positive integer $x'_N \geq x_N$, $Y$ cannot be a decision gadget in $\mathcal{R}^{G_{x'_N}}$.

Similarly, if $N^\ell$ (or $N^{r \ell}$, respectively) is bivalent in $\mathcal{R}^{G}$, then as in Case 1,  there exists a positive integer $x \geq x_N$ such that for all positive integers $x' \geq x$, $N^\ell$ (or $N^{r \ell}$, respectively) is bivalent in $\mathcal{R}^{G_{x'}}$, and $Y$ is not a decision gadget in $\mathcal{R}^{G_{x'}}$. 

Thus, if $Y$ is a tuple $(N,\ell,r,E^\ell,E^r,E^{r\ell})$, then there exists a positive integer $x$ such that for all positive integers $x' \geq x$, $Y$ is a gadget in $\mathcal{R}^{G_{x'}}$, but $Y$ is not a non-$\bot$ decision gadget in $\mathcal{R}^{G_{x'}}$.
\end{proof}

\subsection{Ordering the Decision Gadgets.}


In this subsection, we show
that a ``first'' decision gadget exists in $\mathcal{R}^G$. However, to define the ``first'' decision gadget, we first define a metric function 
%
in four steps: (1) We order the elements in each of the following sets: $\Pi \cup \set{\bot}$, and
$T \cup \set{FD_i | i \in \Pi}$. (2) We order the vertices in $G$. (3) We use the aforementioned orders to define a metric function for each node $N$ in $\mathcal{R}^G$ and for each edge outgoing from $N$. (4) Finally, we define the metric function for each gadget.

\paragraph{Ordering the elements in $\Pi \cup \set{\bot}$.}
Recall that the locations in $\Pi$ are totally ordered by the $<_\Pi$ relation. For simplicity, we assume that $\Pi$ is the set of integers in $[1,n]$ and $\bot = 0$. Thus, $\Pi \cup \set{\bot}$ is totally ordered by the $<$ relation.


\paragraph{Ordering the elements in $T \cup \set{FD_i | i \in \Pi}$.}
Informally, we order $T \cup \set{FD_i | i \in \Pi}$ as follows. $Proc_1,Proc_2,\ldots,$ $Proc_n,Env_{1,0},Env_{1,1},\ldots,$ $Env_{n,0},Env_{n,1}, Chan_{1,2},Chan_{1,3},\ldots,$ \\ $Chan_{1,n},Chan_{2,1},Chan_{2,3},\ldots,$ $Chan_{n,n-1},FD_1,FD_2,\ldots,FD_n$.

Formally, we define  $m: T \cup \set{FD_i | i \in \Pi} \rightarrow [1, n^2 + 3n]$ to be a mapping from all the labels in $\mathcal{R}^G$ to the set of integers in $[1, n^2 + 3n]$ as follows. For each element $l$ in $T$ and each element $l'$ in $\set{FD_i | i \in \Pi}$, $m(l) < m(l')$. Note that $T$ consists of $n$ $Proc_*$ tasks, $2n$ $Env_{*,*}$ tasks, and $n(n-1)$ $Chan_{*,*}$ tasks. For each $Proc_*$ task $l$, each $Env_{*,*}$ task $l'$ and each $Chan_{*,*}$ task $l''$, $m(l) < m(l') < m(l'')$.

For each location $i$, recall that we assume $i \in [1,n]$. For a $Proc_i$ task, $m(Proc_i) = i$. For an $Env_{i,0}$ task, $m(Env_{i,0}) = n + 2i - 1$, and for an $Env_{i,1}$ task, $m(Env_{i,0}) = n + 2i$.
For a $Chan_{i,j}$ task, $m(Chan_{i,j}) = 3n + n(i-1) + j$. It is easy to see that $m$ is a bijection from $T$ to $[1, n^2 + 2n]$. 
We define the mapping from $\set{FD_i | i \in \Pi}$ as follows. $m(FD_i) = n^2+2n + i$. Therefore, $m$ is a bijection from $T \cup \set{FD_i | i \in \Pi}$ to $[1, n^2 + 3n]$.
Thus, the tasks in $T \cup \set{FD_i | i \in \Pi}$ are totally ordered by the range of $m$ and the $<$ relation on integers.


Based on ordering the elements in $T$, we can order any pair of distinct sequences of labels by their lexicographic ordering.

\paragraph{Ordering vertices in $G$.}
We order the vertices $(i,k,e)$ in $G$ first by their index $k$, and break the ties among vertices with the same index by their location $i$.
We define a mapping $m : V \cup \set{(\bot,0,\bot)} \rightarrow \mathbb{N}$, where $G = (V,Z)$ as follows. Note that for any vertex $v = (i,k,e)$, there are potentially infinitely many vertices in $G$ with the same location $i$ and at most $n$ vertices in $G$ whose index is $k$. Based on the above observation, we order all the vertices $G$ by defining $m(v) = k\times n + i$, where $v = (i,k,e)$; note that by this definition, $m((\bot, 0, \bot) = 0$ and for any $v \in V$, $m(v) > 0$. Thus, the vertices in $V$ are totally ordered by the range of $m$ and the $<$ relation on integers.

\paragraph{Ordering outgoing edges from each node in $\mathcal{R}^G$.}
Fix any node $N$ in $\mathcal{R}^G$. We define a total order over the set of edges outgoing from $N$ as follows. Note that $N$ has exactly one outgoing edge for each label in $T$, and potentially infinitely many outgoing edges for each label in $\set{FD_i | i \in \Pi}$. Also note that $|\set{FD_i | i \in \Pi}| = n$. By Lemma \ref{lem:childNodeUniqueByLabelAndVertexTag}, we know that for each outgoing $FD_i$-edge, where $i$ is a location, its vertex tag is distinct from the vertex tag of all other  $FD_i$ edges. Therefore, for a given vertex tag, there can be only finitely many outgoing edges from $N$: there is at most one outgoing $FD_i$ edge for each location $i$ with a given vertex tag, and there is at most one $l$-edge outgoing from $N$ for any non-$FD$ label $l$. It follows that there is at most one outgoing edge from $N$ for a given vertex tag and task label.

Thus, we first order all the edges by their vertex tags, and for a given vertex tag, we order all edges with the same vertex tag by their task label. Formally, this ordered is captured by the metric function $m$ for the outgoing edges $E$ from any node $N$: $m(E) = (m(v_E), m(l_E))$. 

Note that the range of $m$ is $\mathbb{N} \times \mathbb{N}$. The lexicographic ordering of the range of $m$ induces a total order on outgoing edges from each node in $\mathcal{R}^G$.



\paragraph{Ordering all the non-$\bot$ nodes in $\mathcal{R}^G$.}
Recall that each non-$\bot$ node $N$ in $\mathcal{R}^G$ can be uniquely identified by the sequence of labels from $\top$ to $N$ and the sequence of distinct vertex tags in the path from $\top$ to $N$. Also, note that nodes that contain a $\bot$ action tag in the path from $\top$ to $N$ cannot be uniquely identified using the above information. However, for our purposes, it is sufficient to order non-$\bot$ nodes.

Fix a non-$\bot$ node $N'$ in $\mathcal{R}^G$. Let $d_{N'}$ denote the depth of the node, and let $k_{N'}$ denote the index of $v_{N'}$; that is, $v_{N'}=(*,k_{N'}, *)$, where $k_{N'} \in \mathbb{N}$. Let $E^0_{N'}, E^1_{N'},\ldots,E^{d_{N'}}_{N'}$ denote the sequence of edges in the path from $\top$ to $N'$. We define the metric function for each node $N$ in $\mathcal{R}^G$ as follows: $m(N) = (d_N + k_N, m(E^0_N), m(E^1_N),\ldots,m(E^{d_N}_N))$.

Thus, given two nodes $N$ and $N'$ in $\mathcal{R}^G$, we say that $N$ is ordered before $N'$ if either of the following is true.
\begin{itemize}
\item $d_N + k_N < d_{N'} + k_{N'}$.
\item Assuming $d_N + k_N = d_{N'} + k_{N'}$, let $x$ be the smallest integer such that at least one of $E^x_{N}$ and $E^x_{N'}$ exists, and if $E^{x}_{N'}$ also exists, then  $E^x_{N} \neq E^x_{N'}$. Then, $m(E^x_{N}) < m(E^x_N)$. Informally, $N$ is ordered before $N'$ if the sequence of edges from $\top$ to $N$ is lexicographically less than the sequence of sequence of edges from $\top$ to $N'$. 
\end{itemize}

Next, we show that the metric function $m$ imposes a total order on the non-$\bot$ nodes in $\mathcal{R}^G$, and there exists a node with the minimum metric value among all the nodes in $\mathcal{R}^G$.
In Lemma \ref{lem:distinctMetricValueDistinctNodes}, we show that distinct non-$\bot$ nodes must have distinct metric values, which implies that the metric function $m$ establishes a total order over all the non-$\bot$ nodes in $\mathcal{R}^G$ (Lemma \ref{lem:totallyOrderedNodesInTree}).
By implication, $m$ establishes a total order over any non-empty subset of non-$\bot$ nodes in $\mathcal{R}^G$ (Corollary \ref{cor:totallyOrderedNodesInSubset}).
In Lemma \ref{lem:finitelyManyNodesWithSmallerMetric}, we show that for any non-$\bot$ node $N$ there are only finitely many nodes whose metric value is lexicographically smaller than the metric value of $N$ (we use Lemma \ref{lem:OnlyFinitelyManyNodesWithDPlusK} as a helper lemma to prove this). 
Corollary \ref{cor:finitelyManyNodesWithSmallerMetricInSubset} immediately follows from Lemma \ref{lem:finitelyManyNodesWithSmallerMetric};  Corollary \ref{cor:finitelyManyNodesWithSmallerMetricInSubset} states that in any non-empty subset of non-$\bot$ nodes in $\mathcal{R}^G$, for each node $N$, there are only finitely many nodes with a smaller metric value. Lemma \ref{lem:totallyOrderedNodesInTree} and Corollary \ref{cor:finitelyManyNodesWithSmallerMetricInSubset} together imply Corollary \ref{cor:NodesWithSmallestMetricInSubset}, which states that any non-empty subset $\mathcal{N}$ of non-$\bot$ nodes in $\mathcal{R}^G$ contains a unique node with the minimum metric value.

\begin{lemma}\label{lem:distinctMetricValueDistinctNodes}
For any pair $N$, $N'$ of distinct non-$\bot$ nodes in $\mathcal{R}^G$, $m(N) \neq m(N')$.
\end{lemma}
 \begin{proof}
Fix $N$ and $N'$ as in the hypothesis of the lemma. For contradiction, assume $m(N) = m(N')$. Therefore, the sequence of labels in the path from $\top$ to $N$ and from $\top$ to $N'$ are identical (consequently, both $N$ and $N'$ are at the same depth), and $v_N = v_{N'}$. Invoking Lemma \ref{lem:nonBotNodeUniqueExe}, we know that $N = N'$. This contradicts the hypothesis that $N$ and $N'$ are distinct.
 \end{proof}

\begin{lemma}\label{lem:totallyOrderedNodesInTree}
The non-$\bot$ nodes in $\mathcal{R}^G$ are totally ordered by their metric function $m$.
\end{lemma}
\begin{proof}
By Lemma \ref{lem:distinctMetricValueDistinctNodes}, we know that each non-$\bot$ node in $\mathcal{R}^G$ has a distinct metric value. By definition the range of the metric function $m$ of nodes in $\mathcal{R}^G$ are totally ordered (by lexicographic ordering).
Therefore, the non-$\bot$ nodes in $\mathcal{R}^G$ are totally ordered by their metric value.
\end{proof}

\begin{corollary}\label{cor:totallyOrderedNodesInSubset}
For any non-empty subset $\mathcal{N}$ of non-$\bot$ nodes in $\mathcal{R}^G$, the nodes in $\mathcal{N}$ are totally ordered by their metric function $m$.
\end{corollary}
\begin{proof}
Follows from Lemma \ref{lem:totallyOrderedNodesInTree}.
\end{proof}

\begin{lemma}\label{lem:OnlyFinitelyManyNodesWithDPlusK}
For any non-$\bot$ node $N$ in $\mathcal{R}^G$, there are only finitely many nodes $N'$ such that $d_{N'} + k_{N'} \leq d_{N} + k_{N}$.
\end{lemma}
\begin{proof}
We use the following two claims to prove the main lemma.

\emph{Claim 1.} For any vertex $v$ in $G$, there are only finitely many paths in $G$ that end with $v$.
\begin{proof}
Fix a vertex $v = (i,k,e)$ in $G$. 
For contradiction, assume that $G$ contains infinitely many paths ending in $v$.
Therefore, there are infinitely many vertices $v'$ in $G$ such that there is a path from $v'$ to $v$.
By the transitive closure property of $G$, it implies that there are infinitely many vertices $v'$ such that there is an edge in $G$ from $v'$ to $v$.
This contradicts Lemma \ref{prop:finiteIncomingEdges}.
\end{proof}

\emph{Claim 2.} For any pair of positive integers $d$ and $k$, there are only finitely many nodes $N''$ such that $d_{N''} = d$ and $k_{N''} = k$.
\begin{proof}
Fix $d$ and $k$. By construction of $G$, there are at most $n$ vertices $v$ of the form $(*,k,*)$ in $G$; let $\tilde{V}$ be the set of all such vertices. For each $v \in \tilde{V}$, by Claim 1, there are only finitely many paths $p$ in $G$ that end with $v$; let $\tilde{P}$ denote all the paths in $G$ that end with some vertex in $\tilde{V}$. For each $p \in \tilde{P}$, there are only finitely many sequences $p'$ of length $d$ consisting of only the vertices in $p$; let $\tilde{P'}$ denote the set of all sequences over the vertices in some $p \in \tilde{P}$. Note that $\tilde{P'}$ is finite.

Let $\tilde{L}$ be the set of all sequences of length $d$ over $T \cup \set{FD_i | i \in \Pi}$. Note that $\tilde{L}$ is finite.

For each non-$\bot$ node $N''$ in $\mathcal{R}^G$ such that $d_{N''} = d$ and $k_{N''} = k$; let $e_{N''}$ denote the sequence of edges from $\top$ to $N''$. By Lemma \ref{lem:LabelsAndVertexAgsDenoteUniqueNode}, we know that the projection of $e_{N''}$ on the sequence of vertex tags and labels is unique, and by construction, this projection is an element of $\tilde{P'} \times \tilde{L}$. Since $\tilde{P'}$ and $\tilde{L}$ are finite, we conclude that there are only finitely many nodes $N''$ such that $d_{N''} = d$ and $k_{N''} = k$.
\end{proof}

Fix a non-$\bot$ node $N$ in $\mathcal{R}^G$. Let $dk = d_{N} + k_{N}$. We apply Claim 2 for all values of $d$ and $k$, where $d$ is in $[0,dk]$ and $k$ is in $[0, dk-d]$, respectively, to conclude that there are only finitely many nodes $N'$ such that $d_{N'} + k_{N'} \leq d_{N} + k_{N}$.
\end{proof}

\begin{lemma}\label{lem:finitelyManyNodesWithSmallerMetric}
For any non-$\bot$ node $N$ in $\mathcal{R}^G$, there are only finitely many non-$\bot$ nodes $N'$ such that $m(N') \leq m(N)$.
\end{lemma}
\begin{proof}
Fix $N$ as in the hypothesis of the lemma. 
Recall that the first element in $m(N')$ of any node $N'$ is $d_{N'} + k_{N'}$. Therefore, for any non-$\bot$ node $N'$ such that $m(N') < m(N)$, $d_{N'} + k_{N'} \leq d_{N} + k_{N}$.
Invoking Lemma \ref{lem:OnlyFinitelyManyNodesWithDPlusK}, we know that there are only finitely many nodes $N'$ such that $d_{N'} + k_{N'} \leq d_{N} + k_{N}$. Therefore, there are only finitely many non-$\bot$ nodes $N'$ such that $m(N') \leq m(N)$.
\end{proof}

\begin{corollary}\label{cor:finitelyManyNodesWithSmallerMetricInSubset}
For any non-empty subset $\mathcal{N}$ of non-$\bot$ nodes in $\mathcal{R}^G$, for any non-$\bot$ node $N$ in $\mathcal{N}$, there are only finitely many non-$\bot$ nodes $N' \in \mathcal{N}$ such that $m(N') \leq m(N)$.
\end{corollary}
\begin{proof}
Follows from Lemma \ref{lem:finitelyManyNodesWithSmallerMetric}.
\end{proof}

\begin{corollary}\label{cor:NodesWithSmallestMetricInSubset}
For any non-empty subset $\mathcal{N}$ of non-$\bot$ nodes in $\mathcal{R}^G$, there exists a unique non-$\bot$ node $N \in \mathcal{N}$ such that for all $N' \in \mathcal{N}\setminus \set{N}$, $m(N) < m(N')$.
\end{corollary}
\begin{proof}
Fix $\mathcal{N}$ as in the hypothesis of the corollary. For contradiction, assume that for every node $N \in \mathcal{N}$, there exists a node $N' \in \mathcal{N}$ such that $m(N') < m(N)$. By Corollary \ref{cor:totallyOrderedNodesInSubset}, we know that the nodes in $\mathcal{N}$ are totally ordered by their metric value. Therefore, for any node $N \in \mathcal{N}$, there must exist an infinite number of nodes $N'\in \mathcal{N}$ such that $m(N') < m(N)$. However, this contradicts  Corollary \ref{cor:finitelyManyNodesWithSmallerMetricInSubset}.
\end{proof}

\paragraph{Ranking non-$\bot$ nodes in $\mathcal{R}^G$.}
From Lemma \ref{lem:totallyOrderedNodesInTree}, we know that the metric function $m$ for non-$\bot$ nodes establishes a total order over the set of non-$\bot$ nodes in $\mathcal{R}^G$. By Corollaries \ref{cor:finitelyManyNodesWithSmallerMetricInSubset} and  \ref{cor:NodesWithSmallestMetricInSubset}, we map the non-$\bot$ nodes to the set of natural numbers by a function $rank$ defined as follows.

Let $\mathcal{N}^G$ be the set of all non-$\bot$ nodes in $\mathcal{R}^G$.
For any non-negative integer $x$, if $N_x$ is node with the $x$-th smallest metric value among the nodes in $\mathcal{N}^G$, then $rank(N_x) = x$. 

This notion of ``rank'' is used to define the metric value of non-$\bot$ gadgets.

\paragraph{Metric value of non-$\bot$ gadgets.} 
Given a non-$\bot$ gadget of the form $(N,l,r,E^l,E^r,E^{rl})$, it can be uniquely identified by $N$, $N^l$ and $N^{rl}$, where $N^l$ is the lower endpoint of $E^{l}$ and $N^{rl}$ is the lower endpoint of $E^{rl}$. Similarly, given a non-$\bot$ gadget of the from $(N,l,E^l,E^{\prime \ell})$, it can be uniquely identified by $N$, $N^l$ and $N^{\prime \ell}$, where $N^l$ is the lower endpoint of $E^l$ and $N^{\prime \ell}$ is the lower endpoint of $E^{\prime \ell}$.

For a non-$\bot$ decision gadget $(N,l,r,E^l,E^r,E^{rl})$, the metric value of the gadget is defined as $m((N,l,r,E^l,E^r,E^{rl})) = \langle rank(N), \langle rank(N^l),  rank(N^{rl}) \rangle \rangle$, where $\langle \cdot, \cdot \rangle$ is the Cantor pairing function \cite{cantor:ueed}\footnote{Recall that Cantor pairing function $\pi$ is a bijection from $\mathbb{N} \times \mathbb{N}$ to $\mathbb{N}$ and is defined by $\pi(n_1,n_2) = \frac{1}{2} (n_1 + n_2)(n_1 + n_2 + 1) + n_2$.}.
Similarly, for a non-$\bot$ decision gadget $(N,l,E^l,E^{\prime \ell})$, the metric value of the gadget is defined as $m((N,l,E^l,E^{\prime \ell})) = \langle rank(N), \langle rank(N^l),  rank(N^{\prime \ell}) \rangle \rangle$.


\begin{lemma}\label{lem:distinctGadgetsDistinctMetricValue}
For any pair of distinct non-$\bot$ gadgets $Y_1$ and $Y_2$, $m(Y_1) \neq m(Y_2)$.
\end{lemma}
\begin{proof}
Follows from the properties of the Cantor pairing function.
\end{proof}

The \emph{first non-$\bot$ decision gadget} in $\mathcal{R}^G$ is the non-$\bot$ decision gadget with the smallest metric value among all non-$\bot$ decision gadgets in $\mathcal{R}^G$. Next, we show that such a decision gadget exists.

\begin{lemma}\label{lem:finitelyManySmallerGadgets}
For any any non-$\bot$ gadget $Y$ in $\mathcal{R}^G$, there are only finitely many non-$\bot$ gadgets $Y'$ in $\mathcal{R}^G$ such that $m(Y) > m(Y')$.
\end{lemma}
\begin{proof}
The lemma follows directly from the properties of the Cantor pairing function.
\end{proof}

Next, we show that $\mathcal{R}^G$ has a first non-$\bot$ decision gadget.
\begin{theorem}\label{thm:smallestGadget}
$\mathcal{R}^G$ contains a non-$\bot$ decision gadget $Y$ such that the metric
value of any other non-$\bot$ decision gadget $Y'$ is strictly greater than the
metric value of $Y$.
\end{theorem}
\begin{proof}
Let $\mathcal{Y}$ be the set of all non-$\bot$ decision gadgets in $\mathcal{R}^G$.
Fix an arbitrary $Y' \in \mathcal{Y}$. By Lemma \ref{lem:finitelyManySmallerGadgets}, we know that there are only finitely many $Y'' \in \mathcal{Y}$ such that $m(Y'') < m(Y')$. Let $\mathcal{Y}' = \set{ Y'' | Y'' \in \mathcal{Y} \wedge m(Y'') \leq m(Y')}$. Since $\mathcal{Y}'$ is a finite set, let $Y = \argmin_{Y \in \mathcal{Y}'}\set{m(Y)}$. By construction, $Y$ is a non-$\bot$ decision gadget such that the metric
value of any other non-$\bot$ decision gadget $Y'$ is strictly greater than the
metric value of $Y$.
\end{proof}


Given an observation $G$ that is viable for $D$, let
$first(\mathcal{R}^G)$ denote the first non-$\bot$ decision gadget in $\mathcal{R}^G$.

%
%
%
Recall that at most $f$ locations are not
live in $G$; $t_D \in T_D$ is a a trace that is compatible with $D$, and $G_1, G_2, G_3, \ldots$ is a sequence of observations that converge to $G$. Next we show the `persistence' of non-$\bot$ decision gadgets across the sequence of execution trees $\mathcal{R}^{G_1}, \mathcal{R}^{G_2}, \mathcal{R}^{G_3}, \ldots$.

\begin{lemma}
\label{lem:finitelyManySmallerGadgetsInPrefix}
For any $G' \in \set{G, G_1,G_2,\ldots}$, for any any non-$\bot$ gadget $Y$ in $\mathcal{R}^{G'}$, there are only finitely many non-$\bot$ gadgets $Y'$ in $\mathcal{R}^{G'}$ such that $m(Y) > m(Y')$.
\end{lemma}
\begin{proof}
The lemma follows directly from the properties of the Cantor pairing function.
\end{proof}

Let $Y_{min}$ denote $first(\mathcal{R}^G)$: this first non-$\bot$ decision gadget in $\mathcal{R}^G$.

\begin{lemma}\label{lem:minGadgetFoundInAllGs}
There exists a positive integer $x$ such that for all positive integers $x' \geq x$, $Y_{min}$ is the first non-$\bot$ decision gadget in $\mathcal{R}^{G_{x'}}$.
\end{lemma}
\begin{proof}
Applying Lemma \ref{lem:decisionGadgetStableInAllGs}, we know that there exists a positive integer $x_{Y}$ such that for all positive integers $x'_{Y} \geq x_{Y}$, $Y_{min}$ is a non-$\bot$ decision gadget in $\mathcal{R}^{G_{x'_Y}}$. Fix $x_Y$.

By Lemma \ref{lem:finitelyManySmallerGadgetsInPrefix}, we know that there are only finitely many non-$\bot$ gadgets $Y'$ in $\mathcal{R}^{G_{x_Y}}$ such that $m(Y') < m(Y_{min})$.
Let $\hat{\mathcal{Y}}$ denote the set of all such gadgets $Y'$.
By construction all the gadgets in $\hat{\mathcal{Y}}$ are not non-$\bot$ decision gadgets (that is, they are either $\bot$ decision gadgets, or not decision gadgets at all) in $\mathcal{R}^{G}$. 
By Lemma \ref{lem:nonDecisionGadgetStableInAllGs}, we know that for each $Y' \in \hat{\mathcal{Y}}$  there exists a positive integer $x_{Y'}$ such that for all positive integers $x'_{Y'} \geq x_{Y'}$, $Y'$ is not a non-$\bot$ decision gadget in $\mathcal{R}^{G_{x'_{Y'}}}$; fix an $x_{Y'}$ for each such $Y'$. Let $x$ denote the largest such $x_{Y'}$; since $\hat{Y}$ is a finite set, we know that $x$ is exists.

Thus, for all $x' \geq x$, $Y_{min}$ is the first non-$\bot$ decision gadget in $\mathcal{R}^{G_{x'}}$.
\end{proof}

\section{A Weakest AFD for Consensus}
\label{sec: wfd}

%

In \cite{chan:twfdf}, Chandra et al. showed that $\Omega$ is a weakest
failure detector for solving $(n-1)$-crash-tolerant consensus. We use
similar arguments to show that AFD  $\Omega_f$ (defined in Section
\ref{subset:omegaDef}), which is a generalization of the $\Omega$ AFD,
is a weakest strong-sampling AFD to solve
$f$-crash-tolerant consensus in all well-formed environments. Although
the assumption about strong-sampling AFDs seems to weaken our result
with respect to the result in \cite{chan:twfdf}, in fact, a similar
assumption was implicitly used in \cite{chan:twfdf}.


Recall that $\Omega_f$, where $0 \leq f <n$, denotes the AFD that
behaves exactly like $\Omega$ in traces that have at most $f$ faulty
locations, and in traces that have more than $f$ faulty locations, the outputs by $\Omega$ are unconstrained.
In order to show that $\Omega_f$ is weakest to solve
$f$-crash-tolerant consensus, first we have to show that
$f$-crash-tolerant consensus can be solved using $\Omega_f$ in any
well-formed environment. Since $\Omega_f$ behaves exactly like
$\Omega$ in executions where at most $f$ locations crash, we see that
the algorithm in \cite{chan:ufdfr} can be modified trivially to solve
$f$-crash-tolerant consensus using $\Omega_f$ in any well-formed
environment. It remains to show that, for every strong-sampling AFD
$\FD$, if $\FD$ is sufficient to solve $f$-crash-tolerant consensus in
any well-formed environment,  then $\FD$ is stronger than $\Omega_f$. 

For the remainder of this section, fix $f$ to be a positive integer not exceeding $n$.

In Section \ref{subsec:omegaExtraction}, we present an algorithm that
solves $\Omega_f$ using any arbitrary strong-sampling AFD that solves
$f$-crash-tolerant consensus, and in Section
\ref{subsec:omegaExtractionCorrectness}, we present the proof of
correctness. Since we know that $\Omega_f$ is sufficient to solve
$f$-crash-tolerant consensus, we thus establish that $\Omega_f$ is a
weakest AFD to solve $f$-crash-tolerant consensus.

\subsection{Algorithm for Solving $\Omega_f$}\label{subsec:omegaExtraction}

Let $\FD$ be an AFD sufficient to solve $f$-crash-tolerant consensus,
where $0 \leq f < n$, in the well-formed environment $\mathcal{E}_C$
from Section \ref{subsec:conEnvDef}. By definition, there exists a
distributed algorithm $A_D^P$ that uses $\FD$ to solve
$f$-crash-tolerant consensus in $\mathcal{E}_C$. Using $A_D^P$, we
construct an algorithm $A^\Omega$ that uses $\FD$ to solve
$\Omega_f$.

In $A^\Omega$, each process automaton keeps track of the outputs
provided by AFD $\FD$ and exchanges this information with all other
process automata (at other locations).
Each process uses this information to maintain an observation $G$ (a local variable), and sends this observation to the other process automata. Initially, the observation $G$
at each process automaton is empty, and the local variable $k$, which
counts the number of AFD events that have occurred at that location,
is $0$. Each process also maintains a local variable $fdout$ which may be viewed as the automaton's current estimate of the output of the AFD $\Omega_f$ that it implements; initially, at each process $i$, the value of $fdout$ is $i$. Next, we describe the actions of the process automaton at a location
(say) $i$. 

When an AFD output $d$ occurs at location $i$, the input action
$d$ occurs in $A^\Omega_i$; in this action, the process automaton does 
the following.
It increments $k$ by $1$ (which updates the number of AFD events that
have occurred at $i$) and inserts a new vertex $(i,k,d)$ into its local variable $G$; the insert operation is defined in Section \ref{subsec:OperationsOnObservations}.
A copy of the updated observation $G$ is appended to $sendq$ for every
other location to be sent out to all other locations. The process automaton
constructs the directed tree $\mathcal{R}^G$ for the current value of $G$ (as described in Section \ref{subsec:treeOfExec}). If $\mathcal{R}^G$ contains a non-$\bot$ gadget, then it determines the first non-$\bot$ decision gadget in 
$\mathcal{R}^G$ and
updates $fdout$ to the critical location of that decision gadget. Finally, the automaton adds
$(fdout,i)$ to $sendq$.

If the front element of $sendq$ is a pair consisting of an observation $observe$
and  location $j$, then the output action $send(observe,j)$ is enabled. When this
action occurs, the front element of $sendq$ is deleted (and a message is
send to $j$ that contains the observation $observe$).

When the process automaton at $i$ receives a message from another
location $j$ with the observation $observe$, the input event $receive(observe,j)$
occurs, and the process automaton updates $G$ with the union of $G$
and $observe$; the union operation is defined in Section \ref{subsec:OperationsOnObservations}.

If the front element of $sendq$ is a pair $(j,i)$, where $j$ is a location,
the output action $FD_\Omega(j)$ is enabled. When this action occurs, 
the front element of $sendq$ is deleted.

Note that $sendq$ contains both the observations that are sent to other locations and the value of the $\Omega_f$ AFD output events. This is because we model process automata as having a single task. Alternatively, we could have modeled process automata as having multiple tasks and used separate data structures to store the AFD outputs and the observations to be sent to other locations.

 
 The pseudocode for the algorithm is given in Algorithm \ref{alg:exchangeFD}.

\begin{algorithm}\footnotesize
\caption{Algorithm for solving $\Omega$}
\label{alg:exchangeFD}
The automaton $A^{\Omega}_i$ at each location $i$.


\textbf{Signature:}

\tab input $crash_i$

\tab input $d$: $O_{D,i}$

\tab input $receive(obs: \text{Observation}, j:\Pi\setminus\set{i})_i$

\tab output $send(obs: \text{Observation}, j:\Pi\setminus\set{i})_i$

\tab output $FD_{\Omega}(j: \Pi)$

\tab 

\textbf{Variables:}

\tab $G$: a finite observation, initially empty \hfill \emph{Finite observation maintained at all locations}

\tab $k$: integer, initially $0$ \hfill \emph{Denotes the number of AFD outputs occurred so far}

\tab $sendq$: queue of pairs $(o,j)$, where $o$ is either an observation or a location, and $j$ is a location, initially empty. 

\tab $fdout$: $\Pi$, initially $i$ \hfill \emph{Location ID output by the $\Omega_f$ AFD output actions}

\tab $faulty$: Boolean, initially $false$ \hfill \emph{When true, the process automaton is crashed}

\tab

\textbf{Actions:}

\tab input $crash$

\tab effect

\tab \tab $faulty$ := $true$

\tab

\tab input $d$: $d \in O_{D,i}$

\tab effect

\tab \tab if not faulty, then

\tab \tab \tab $k$ :=  $k+1$


\tab \tab \tab insert vertex $(i,k,d)$ into $G$ \hfill \emph{The insert operation is defined in Section \ref{subsec:OperationsOnObservations}.}

\tab \tab \tab foreach $j \in \Pi \setminus \set{i}$

\tab \tab \tab \tab append $(G,j)$ to $sendq$

\tab \tab \tab if $\mathcal{R}^G$ contains a non-$\bot$ decision gadget, then

\tab \tab \tab \tab $H$ := $first(\mathcal{R}^G)$ \hfill \emph{Recall that $first(\mathcal{R}^G)$ is the first non-$\bot$ decision gadget in $\mathcal{R}^G$}


\tab \tab \tab \tab $fdout$ := critical location of $H$

\tab \tab \tab append $(fdout, i)$ to $sendq$

\tab

\tab input $receive(obs, j)$

\tab effect

\tab \tab if not faulty, then

\tab \tab \tab $G$ := $G \cup obs$\hfill \emph{The union operation is defined in Section \ref{subsec:OperationsOnObservations}.}

\tab

\tab output $send(obs,j)$

\tab precondition

\tab \tab $(\neg faulty \wedge ((obs,j) = head(sendq))) $

\tab effect

\tab \tab delete head of $sendq$

\tab

\tab output $FD_\Omega(j)$

\tab precondition

\tab \tab $(\neg faulty \wedge ((j,i) = head(sendq)))$ 

\tab effect

\tab \tab delete head of $sendq$

\end{algorithm}

\subsection{Correctness}
\label{subsec:omegaExtractionCorrectness}


Fix an arbitrary fair execution $\alpha$ of the system consisting of
$A^\Omega$, the channel automata, the crash automaton, and the
well-formed environment $\mathcal{E}_C$ such that $\alpha|_{\hat{I}
  \cup O_D} \in T_D$ and at most $f$ locations crash in $\alpha$. Let
$\alpha|_{\hat{I} \cup O_D} = t_D$. Recall that AFD $\Omega_f$ behaves
exactly like $\Omega$ if at most $f$ locations crash.
Thus, it remains to show that $\alpha|_{\hat{I} \cup O_{\Omega}} \in T_{\Omega}$.

 The remainder of this section uses the following notation. Recall
 that an execution is a sequence of alternating states and actions. In
 execution $\alpha$, $\alpha_s[k]$ denotes the $k$-th state in
 $\alpha$, and $\alpha_s[k].G_i$ denotes the value of the observation
 $G_i$ in state $\alpha_s[k]$. We assume that the initial state of $\alpha$, denoted $\alpha_s[0]$, is the $0$-th state in $\alpha$.
 
 The proof is divided into three parts. 
In Section~\ref{subsubsec:proofPartOne}, we prove some basic properties of
the graphs $G_i$, where $i$ is a location, that are used in the
remainder of the proof. 
%
%
In Section~\ref{subsubsec:proofGisAreObs}, we show that each
$\alpha[k].G_i$, where $k$ is a positive integer and $i$ is a
location, is a viable observation for $D$.
In Section~\ref{subsubsec:proofLimitObsIsViable}, we show that for all
live locations $i$, the limits $G^\infty_i$ of $\alpha[k].G_i$, as $k$
approaches $\infty$, are identical and a viable observation for
$D$; therefore, we denote all $G^\infty_i$ (for all locations $i$) as $G^\infty$. 
Finally, in Section~\ref{subsubsec:proofSmallestDecisionGadget}, we
identify the ``first'' non-$\bot$ decision gadget $Y$ in $G^\infty$ and show that
for each live location $i$, eventually, $Y$ is also the first
non-$\bot$ decision gadget for $G_i$. 
Since each live process eventually detects the same decision gadget as
the ``first'', each live process eventually and permanently outputs
the same live location as the output of
$\Omega_f$. This completes the proof. 
  
\subsubsection{Properties of the graphs $G_i$ at each location $i$}
\label{subsubsec:proofPartOne}


Here we present some basic properties of the $G_i$ graphs\footnote{Although $G_i$ for each location $i$ is an observation, we have not yet shown this to be the case. Consequently, we refer to them merely as ``graphs''. We prove that the $G_i$s are observations in Theorem \ref{thm:GiIsObs}.}.
Lemma \ref{prop:subgraphNext} states that the value of $G_i$ in any state is a subgraph of its value in any later state.
For a triple $v = (i,\hat{k},e)$ that exists in some $\alpha_s[x'].G_{j'}$, let $x$ be the smallest positive integer such that $\alpha_s[x].G_j$ contains the vertex $v$ for some location $j$; then, vertex $v$ said to ``appear'' in $\alpha$ at index $x$.
Lemma \ref{lem:vertexFirstAtHomeLocation} establishes that when a new vertex $v = (i,\hat{k},e)$ ``appears'' in $\alpha$ at index $x$, $v$ is inserted into $G_i$; that is, $\alpha_s[x].G_i$ contains $v$. 
Lemma \ref{lem:vertexInsertionProp} establishes that  when $v = (i,\hat{k},e)$ first ``appears'' in $\alpha$ at index $x$ (1) $e$ precedes the state $\alpha_s[x]$ in $\alpha$, (2) the value of $k_i$ is $\hat{k} -1$, (3) $e$ is the $\hat{k}$-th $O_{D,i}$ event in $\alpha$, (4) $G_i$ does not contain any other vertex of the form $(i,\hat{k},*)$, and (5) $G_i$ contains vertices of the form $(i,k',*)$ for all $k' < \hat{k}$. Lemma \ref{lem:edgeFirstAtHomeLocation} establish that when a vertex $v$ ``appears'' in $\alpha$, all the incoming edges to $v$ are fixed and do not change thereafter. Lemma \ref{lem:containsVertexContainsSubgraph} establishes that if $v$ ``appears'' in $\alpha$ at index $x$, then for all $x' \geq x$, $\alpha_s[x].G_i$ is a subgraph of $\alpha_s[x'].G_j$. Finally, Lemma \ref{lem:edgeHappensBefore} establishes that if an edge $(v_1,v_2)$ occurs in any graph $G_i$, then the event of $v_1$ precedes the event of $v_2$ in $\alpha$.

\begin{lemma}
\label{prop:subgraphNext}
For each positive integer $x$ and each location $i$, $\alpha_s[x].G_i$ is a subgraph of $\alpha_s[x+1].G_i$.
\end{lemma}  

\begin{proof}
Fix $i$ and $x$ as in the hypotheses of the Lemma. The proof follows from the observation that no vertex and no edge in $\alpha_s[x].G_i$ is deleted in $\alpha_s[x+1].G_i$, by any action.
\end{proof}
\begin{corollary}\label{cor:subgraphFuture}
For each positive integer $x$, each location $i$, for all positive integers $x'>x$, $\alpha_s[x].G_i$ is a subgraph of $\alpha_s[x'].G_i$.
\end{corollary}


\begin{lemma}
\label{lem:vertexFirstAtHomeLocation}
For any vertex $(i,\hat{k},e)$, let $x$ be the smallest integer such that
for some location $j$, $\alpha_s[x].G_j$ contains the vertex
$(i,\hat{k},e)$. Then (1) $j=i$ and (2) event $e$ immediately precedes $\alpha_s[x]$ in $\alpha$.
\end{lemma} 

\begin{proof}
Fix $(i,\hat{k},e)$, $x$, and $j$ as in the hypotheses of the lemma. Therefore,
$\alpha_s[x-1].G_j$ does not contain the vertex $(i,\hat{k},e)$ and
$\alpha_s[x].G_j$ contains the vertex $(i,\hat{k},e)$. Let $a$ be the action
that occurs between states $\alpha_s[x-1]$ and $\alpha_s[x]$ in
$\alpha$.

First, we prove part 1 of the lemma.
From the pseudocode, we know that $a$ is either an action in $O_{D,j}$
or an action of the form $receive(*,*)_{j}$. In the former case, we
see that $j=i$. We show that the latter case is impossible.

For contradiction, assume that $a$ is an action of the form
$receive(observe,j')_j$. From the pseudocode, we see that $observe$
contains vertex $(i,\hat{k},e)$. However, from the reliable FIFO behavior of
the channel automata, we know the process automaton at $j'$ must have
sent the message $observe$ containing vertex $(i,\hat{k},e)$ before state
$\alpha_s[x-1]$ in $\alpha$. Let this occur in state $\alpha_s[x^-]$,
where $x^- < x$. Therefore, $\alpha_s[x^-].G_{j'}$ contains vertex
$(i,\hat{k},e)$, which contradicts our assumption that $x$ is the smallest
integer such that for some location $j$, $\alpha_s[x].G_j$ contains
the vertex $(i,\hat{k},e)$; this establishes part 1 of the lemma. 

Also, we see 
that $a$ must be an action in $O_{D,j}$, and from the pseudocode,
we conclude that $a = e$; this establishes part 2 of the lemma.
\end{proof}

 \begin{lemma}\label{lem:vertexInsertionProp}
 For any vertex $(i,\hat{k},e)$, let $x$ be the smallest integer such that $\alpha_s[x].G_i$ contains the vertex $(i,\hat{k},e)$.
The following are true.
\begin{enumerate}
\item $\alpha_s[x-1].k_i = \hat{k}-1$.
\item $e = \alpha|_{O_{D,i}}[\hat{k}]$
\item $\alpha_s[x-1].G_i$ does not contain any other vertex of the form $(i,\hat{k},*)$.
\item For each positive integer $k' < \hat{k}$, $\alpha_s[x-1].G_i$ contains one vertex of the form $(i,k',*)$.
\end{enumerate}
\end{lemma}
\begin{proof}
Fix $i$, $v=(i,\hat{k},e)$ and $x$ as in the hypotheses of the lemma. We prove the lemma by induction on $\hat{k}$.

\emph{Base Case.} Let $\hat{k}=1$. When the first event $e$ from $O_{D,i}$
occurs in $\alpha$, from the pseudocode, we see that the vertex
$(i,1,e)$ is added to $G_i$. Therefore, for vertex $(i,1,e)$, let $x$
be the smallest integer such that $\alpha_s[x].G_i$ contains the
vertex $(i,1,e)$. From the pseudocode, we see that (1) $\alpha_s[x-1].k_i = 0$.
Since $e$ is the first event from $O_{D,i}$, (2) $e = \alpha|_{O_{D,i}}[1]$. Note that (3) $\alpha_s[x-1].G_i$ does not contain any vertex of the form $(i,1,*)$. Property 4 is satisfied vacuously.

\emph{Inductive Hypothesis.}  For any vertex $(i,\hat{k},e)$, let $x$ be the
smallest integer such that $\alpha_s[x].G_i$ contains the vertex
$(i,\hat{k},e)$. Then the following is true.
\begin{enumerate}
\item $\alpha_s[x-1].k_i = \hat{k}-1$.
\item $e = \alpha|_{O_{D,i}}[\hat{k}]$.
\item $\alpha_s[x-1].G_i$ does not contain any other vertex of the form $(i,\hat{k},*)$.
\item For each positive integer $k' < \hat{k}$, $\alpha_s[x-1].G_i$ contains one vertex of the form $(i,k',*)$.
\end{enumerate}

\emph{Inductive Step.} 
Let $x'$ be the smallest integer such that $\alpha_s[x'].G_i$ contains the vertex
$(i,\hat{k}+1,e')$ for some $e'$. Applying Lemma \ref{lem:vertexFirstAtHomeLocation}, we know that for every other location $j$ and all $x'' \leq x'$, $\alpha_s[x''].G_j$ does not contain the vertex $(i,\hat{k}+1,e')$ and the event preceding $\alpha_s[x']$ is event $e'$.
From the pseudocode, we see that $e' \in O_{D,i}$, and since any action from $O_{D,i}$ increments $k_i$ by $1$, we conclude that (1) $\alpha_s[x'-1].k_i = \hat{k}$. Also, since $k_i$ is updated only when an action from $O_{D,i}$ occurs, $e = \alpha|_{O_{D,i}}[\hat{k}]$, and when $e'$ occurs, vertex $(i,\hat{k}+1,e')$ is inserted to $G_i$, we conclude that (2) $e' = \alpha|_{O_{D,i}}[\hat{k}+1]$. 

From the inductive hypothesis we know that $\alpha_s[x-1].k_i = \hat{k}-1$. Since $e \in O_{D,i}$, and any action from $O_{D,i}$ increments $k_i$, we know that $\alpha_s[x].k_i = \hat{k}$. We have already established that $\alpha_s[x'-1].k_i = \hat{k}$. Therefore, $e'$ is the earliest event from $O_{D,i}$ that follows $e$. That is, (3) $e = \alpha|_{O_{D,i}}[\hat{k}+1]$.

By the inductive hypothesis, we know that each positive integer $k' < \hat{k}$, $\alpha_s[x-1].G_i$ contains one vertex of the form $(i,k',*)$. We have established that $e'$ is the earliest event from $O_{D,i}$ that follows $e$. Therefore, $\alpha_s[x'-1].G_i$ contains exactly one event of the form $(i,\hat{k},*)$, which is $(i,\hat{k},e)$. Therefore, (4) for each positive integer $k' < \hat{k}+1$, $\alpha_s[x'-1].G_i$ contains one vertex of the form $(i,k',*)$.
\end{proof}


\begin{lemma}\label{lem:edgeFirstAtHomeLocation}
For any location $j$, any positive integer $x$, and any pair of
vertices $u$ and $v = (i,k,e)$ such that $\alpha_s[x].G_j$ contains
the edge $(u,v)$, the following is true. Let $x'$ be the smallest
positive integer such that $\alpha_s[x'].G_i$ contains the vertex
$v$. 
Then $\alpha_s[x'].G_i$ contains the edge $(u,v)$.
\end{lemma}


\begin{proof}
Fix $j$, $x$, $u$, $v = (i,k,e)$, and $x'$ as in the hypotheses of the
lemma. Let $x_{min}$ be the smallest positive integer such that for
some location $j'$, $\alpha_s[x_{min}].G_{j'}$ contains the edge
$(u,v)$. Applying Lemma \ref{lem:vertexFirstAtHomeLocation}, 
we know that $x_{min} \geq x'$. If $x_{min} > x'$, then note that the edge
$(u,v)$ is added to $G_{j'}$ by an action of the form
$receive(observe,j'')_{j'}$, where $observe$ contains the edge
$(u,v)$. However, this implies that for some $x_{prev} < x_{min}$,
$\alpha_s[x_{prev}].G_{j''}$ contains the edge $(u,v)$, and this
contradicts the definition of $x_{min}$. Therefore, $x_{min} =
x'$. Applying Lemma \ref{lem:vertexFirstAtHomeLocation}, 
we know that
$j'=i$. Therefore, $\alpha_s[x'].G_i$ contains the edge $(u,v)$.
\end{proof}

\begin{lemma}\label{lem:containsVertexContainsSubgraph}
For any vertex $(i,k,e)$, let $x$ be the smallest integer such that
for some location $j$, $\alpha_s[x].G_j$ contains the vertex
$(i,k,e)$. For any location $j'$ and any positive integer $x'$ such
that $\alpha_s[x'].G_{j'}$ contains the vertex $(i,k,e)$,
$\alpha_s[x].G_j$ is a subgraph of $\alpha_s[x'].G_{j'}$.
\end{lemma}

\begin{proof}
Fix $(i,k,e)$, $x$, and $j$ as in the hypotheses of the lemma. Applying Lemma \ref{lem:vertexFirstAtHomeLocation} we know that $j=i$.

For contradiction, assume there exists a location $j'$ and a positive integer $x'$ such that $\alpha_s[x'].G_{j'}$ contains the vertex $(i,k,e)$, but $\alpha_s[x].G_j$ is not a subgraph of $\alpha_s[x'].G_{j'}$. Fix the smallest such $x'$ and the corresponding location $j'$ such that $\alpha_s[x'].G_{j'}$ contains the vertex $(i,k,e)$. 

From the definition of $x$ we know that $x' \geq x$. Applying Corollary \ref{cor:subgraphFuture}, we know that $\alpha_s[x].G_{i}$ is a subgraph of $\alpha_s[x'].G_{i}$, and therefore  $j' \neq i$. 

Since $x'$ is the smallest integer such that $\alpha_s[x'].G_{j'}$ contains the vertex $(i,k,e)$ and $j' \neq i$, we conclude that the action preceding $\alpha_s[x'].G_{j'}$ in $\alpha$ is an action of the form $receive(observe, j'')_{j'}$, where $observe$ contains the vertex $(i,k,e)$ and $\alpha_s[x].G_{i}$ is not a subgraph of $observe$. Fix the location $j''$. Therefore, there exists a positive integer $x'' < x'$ such that $\alpha_s[x''].G_{j''}$ contains the vertex $(i,k,e)$ and $\alpha_s[x].G_{i}$ is not a subgraph of $\alpha_s[x''].G_{j''}$. This contradicts the definition of $x'$.
\end{proof}

\begin{lemma}
\label{lem:edgeHappensBefore}
For any edge $(v_1,v_2)$ in $\alpha_s[x].G_i$, the event $e_1$ occurs
before event $e_2$ in $\alpha$, where $v_1= (i_1,k_1,e_1)$ and $v_2 =
(i_2,k_2,e_2)$.
\end{lemma}

\begin{proof}
Fix $v_1= (i_1,k_1,e_1)$ and $v_2 = (i_2,k_2,e_2)$, as in the hypotheses of the lemma.

Applying Lemma \ref{lem:vertexFirstAtHomeLocation}, we know that there
exists a positive integer $x_1$ such that (1) $\alpha_s[x_1].G_{i_1}$
contains vertex $v_1$, (2) for each positive integer $x'_1 < x_1$,
$\alpha_s[x'_1].G_{i_1}$ does not contain $v_1$, and (3) for each
positive integer $x'_1 \leq x_1$ and every other location $j$,
$\alpha_s[x'_1].G_{j}$ does not contain the vertex $v_1$.

Similarly, applying Lemma \ref{lem:vertexFirstAtHomeLocation}, we know
that there exists a positive integer $x_2$ such that (1)
$\alpha_s[x_2].G_{i_2}$ contains vertex $v_2$, (2) for each positive
integer $x'_2 < x_2$, $\alpha_s[x'_2].G_{i_2}$ does not contain the
vertex $v_2$, and (3) for each positive integer $x'_2 \leq x_2$ and
every other location $j$, $\alpha_s[x'_2].G_{j}$ does not contain the
vertex $v_2$. From Lemma \ref{lem:edgeFirstAtHomeLocation} we know
that $\alpha_s[x_2].G_{i_2}$ also contains the edge $(v_1,v_2)$.

Therefore, $\alpha_s[x_1].G_{i_1}$ contains vertex $v_1$ and does not
contain $v_2$, whereas $\alpha_s[x_2].G_{i_2}$ contains vertices $v_1$
and $v_2$. Applying Lemma \ref{lem:containsVertexContainsSubgraph}, we
know that  $\alpha_s[x_1].G_{i_1}$ is a subgraph of
$\alpha_s[x_2].G_{i_2}$. From the definition of $x_1$ and $x_2$, we
know that $x_1 < x_2$. Note that $v_1$ is added to
$\alpha_s[x_1].G_{i_1}$ when event $e_1$ occurs in $\alpha$ after
state $\alpha_s[x_1-1]$, and similarly, $v_2$ is added to
$\alpha_s[x_2].G_{i_1}$ when event $e_2$ occurs in $\alpha$ after
state $\alpha_s[x_2-1]$. Therefore, $e_1$ occurs before $e_2$ in
$\alpha$.
\end{proof}

%


%


\subsubsection{For each location $i$, $G_i$ is an observation}
\label{subsubsec:proofGisAreObs}

In this subsection, we prove in Theorem \ref{thm:GiIsObs} that for each location $i$ and each positive
integer $x$, $\alpha_s[x].G_i$ is an observation for $D$. We use two three lemmas to prove the main result. In Lemma \ref{lem:insertYeildsObservation}, we prove that for any location $i$, if the graph $G_i$ is an observation and an event from $O_{D,i}$ occurs, then in the resulting state, $G_i$ is an observation. In Lemma \ref{lem:sameIndexSameEvent}, we show that for any two graphs $\alpha_s[x].G_j$ and $\alpha_s[x'].G_{j'}$, and for every vertex $v = (i,k,e)$ from $\alpha_s[x].G_j$, either $\alpha_s[x'].G_{j'}$ also contains $v$, or $\alpha_s[x'].G_{j'}$ does not contain any vertex of the form $(i,k,*)$. In Lemma \ref{lem:sameVertexSameIncomingEdges}, we show that for any two graphs $\alpha_s[x].G_j$ and $\alpha_s[x'].G_{j'}$, for any vertex $v$ that is in both $\alpha_s[x].G_j$ and $\alpha_s[x'].G_{j'}$, $v$ has the same set of incoming edges in both $\alpha_s[x].G_j$ and $\alpha_s[x'].G_{j'}$.

\begin{lemma}\label{lem:insertYeildsObservation}
For any location $i$ and a positive integer $x$, if $\alpha_s[x].G_i$ is an observation and the event $e$ between $\alpha_s[x]$ and $\alpha_s[x+1]$ in $\alpha$ is an event from $O_{D,i}$  then $\alpha_s[x+1].G_i$ is an observation.
\end{lemma}
\begin{proof}
Fix $i$, $x$, and $e$ from the hypothesis of the lemma.
From the pseudocode, we know that when $e$ occurs, a vertex $v$ of the form $(i,e,\hat{k})$ is added to $G_i$, and for each vertex $u$ in $\alpha_s[x].G_i$, the edge $(u,v)$ is added to $G_i$ as well.
From Lemma \ref{lem:vertexInsertionProp}, we know that $\alpha_s[x].k_i = \hat{k}-1$ and $\alpha_s[x].G_i$ does not contain any vertex of the form $(i,*,\hat{k})$. 
Therefore, $\alpha_s[x+1].G_i = insert(\alpha_s[x].G_i,v)$; invoking Lemma \ref{prop:insertVertexYeildsObservation}, we conclude that $\alpha_s[x+1].G_i$ is an observation.
\end{proof}

\begin{lemma} \label{lem:sameIndexSameEvent}
For any pair of positive integers $x$ and $x'$, and any pair of locations $j$ and $j'$, if $\alpha_s[x].G_j$ contains a vertex $v = (i,k,e)$, then it is not the case that $\alpha_s[x'].G_{j'}$ contains a vertex $v' = (i,k,e')$ where $e \neq e'$.
\end{lemma}
\begin{proof}
Fix a pair of positive integers $x$ and $x'$, and a pair of locations $j$ and $j'$, such that $\alpha_s[x].G_j$ contains a vertex $v = (i,k,e)$, $\alpha_s[x'].G_{j'}$ contains a vertex $v' = (i,k,e')$. We complete the proof by showing that $e = e'$.

Let $x_1$ be the smallest positive integer such that for some location $i_1$, $\alpha_s[x_1].G_{i_1}$ contains the vertex $v$, and let $x_2$ be the smallest positive integer such that for some location $i_2$, $\alpha_s[x_2].G_{i_2}$ contains the vertex $v'$. Invoking Lemma \ref{lem:vertexFirstAtHomeLocation}, we know that $i_1 = i_2 = i$. Invoking Lemma \ref{lem:vertexInsertionProp}, we know that $e = \alpha|_{O_{D,i}}[k]$ and $e'= \alpha|_{O_{D,i}}[k]$; that is, $e=e'$.
\end{proof}

\begin{lemma}\label{lem:sameVertexSameIncomingEdges}
For any pair of positive integers $x$ and $x'$, and any pair of locations $j$ and $j'$, for every vertex $v$ in $\alpha_s[x].G_j$ and $\alpha_s[x'].G_{j'}$, If an edge $(u,v)$ is in $\alpha_s[x].G_j$, then the edge $(u,v)$ exists in $\alpha_s[x'].G_{j'}$.
\end{lemma}
\begin{proof}
Fix a pair of positive integers $x$ and $x'$, and a pair of locations $j$ and $j'$. If the set of vertices of $\alpha_s[x].G_j$ and $\alpha_s[x'].G_{j'}$ are disjoint, then the lemma is satisfied vacuously. For the remainder of the proof, assume that there exists at least one vertex in both $\alpha_s[x].G_j$ and $\alpha_s[x'].G_{j'}$. Fix such a vertex $v = (i,k,e)$. Fix $u$ to be any vertex in $\alpha_s[x].G_j$ such that $(u,v)$ is an edge in $\alpha_s[x].G_j$. We show that the edge $(u,v)$ exists in $\alpha_s[x'].G_{j'}$.

Let $\hat{x}$ be the smallest positive integer such that $\alpha_s[\hat{x}].G_i$ contains the vertex $v$. Invoking Lemma \ref{lem:edgeFirstAtHomeLocation}, we know that $\alpha_s[\hat{x}].G_i$ contains the edge $(u,v)$. Invoking Lemma \ref{lem:containsVertexContainsSubgraph}, we know that $\alpha_s[\hat{x}].G_i$ is a subgraph of $\alpha_s[x'].G_{j'}$, and therefore, $\alpha_s[x'].G_{j'}$ contains the edge $(u,v)$.
\end{proof}

\begin{theorem}
\label{thm:GiIsObs}
For each location $i$, for each positive integer $x$, $\alpha_s[x].G_i$ is an observation.
\end{theorem}
\begin{proof}
We prove the theorem by strong induction on $x$.

\emph{Inductive Hypothesis.} For each location $i$, and each positive integer $x' < x$, $\alpha_s[x'].G_i$ is an observation.

\emph{Inductive Step.} Fix a location $i$. We know that for $x = 0$, $\alpha_s[x].G_i$ is the empty graph, and is therefore an observation. The remainder of the proof assumes $x >1$. We know from Lemma \ref{prop:subgraphNext} that $\alpha_s[x-1].G_i$ is a subgraph of $\alpha_s[x].G_i$. Therefore, either $\alpha_s[x-1].G_i = \alpha_s[x].G_i$, or $\alpha_s[x-1].G_i \neq \alpha_s[x].G_i$. In the former case, we apply the inductive hypothesis to conclude that $\alpha_s[x].G_i$ is an observation. In the latter case, the following argument holds.

From the pseudocode, we know that the event $e$ between $\alpha_s[x-1]$ and $\alpha_s[x]$ in $\alpha$ is either (1) an event from $O_{D,i}$ or (2) an event of the form $receive(observe, j)_i$ for some $j \neq i$.

\emph{Case 1.} Let $e$ be an event from $O_{D,i}$. Recall that by the inductive hypothesis, $\alpha_s[x-1].G_i$ is an observation. Invoking Lemma \ref{lem:insertYeildsObservation}, we conclude that $\alpha_s[x].G_i$ is an observation.

\emph{Case 2.} Let $e$ be an event of the form $receive(observe, j)_i$ for some $j \neq i$. From the FIFO property of the channels, we know that an event $send(observe,i)_j$ occurred in $\alpha$ before event $e$, From the pseudocode, we know that for some $x' < x$, $observe = \alpha_s[x'].G_j$. By the inductive hypothesis, we conclude that $observe$ and $\alpha_s[x-1].G_i$ are observations. Also, from the pseudocode, we know that when even $e$ occurs, $G_i$ is updated to $G_i \cup observe$. Therefore,  $\alpha_s[x].G_i = \alpha_s[x-1].G_i \cup observe$. 
By Lemma \ref{lem:sameIndexSameEvent}, we know that for each vertex $v = (i',k',e')$ in $observe$, it is not the case that $\alpha_s[x-1].G_i$ contains a vertex $v' = (i',k',e'')$ where $e'' \neq e'$, and invoking Lemma \ref{lem:sameVertexSameIncomingEdges}, we know that for every vertex $v$ in both $observe$ and $\alpha_s[x-1].G_i$, $v$ has the same set of incoming edges in both $observe$ and $\alpha_s[x-1].G_i$.
Therefore, we can invoke Lemma \ref{prop:obs:union} to conclude that $\alpha_s[x].G_i$ is an observation. This completes the induction.
\end{proof}

\subsubsection{The limit of the $G_i$'s is a viable observation}
\label{subsubsec:proofLimitObsIsViable}


For each location $i$, we define $G_i^\infty$ to be the limit of $\alpha_s[k].G_i$ as $k$ tends to $\infty$. In this subsection, we show that
for each pair of live locations $i$ and $j$, $G_i^\infty = G_j^\infty$, and this limiting observation is viable for $D$.

Recall that the limit $G^\infty_i = (V^\infty_i,Z^\infty_i)$ is defined as follows. Let $\alpha_s[k].G_i = (V^k_i,Z^k_i)$ for each natural number $k$. Then, $V^\infty_i = \bigcup_{k \in \mathbb{N}} V^k_i$ and $Z^\infty_i = \bigcup_{k \in \mathbb{N}} Z^k_i$.


\begin{lemma}
\label{lem:prefixOfFutureObs}
For each location $i$, for every pair of integers $x$,$x'$, where
$x'>x$,  $\alpha_s[x].G_i$ is a prefix of $\alpha_s[x'].G_i$.
\end{lemma} 

\begin{proof}
Fix $i$, $x$, and $x'$, as in the hypotheses of the lemma. Applying Theorem
\ref{thm:GiIsObs}, we know that $\alpha_s[x].G_i$ and $\alpha_s[x'].G_i$ are
observations.
From Corollary \ref{cor:subgraphFuture} we know that $\alpha_s[x].G_i$ is a
subgraph of $\alpha_s[x'].G_i$. Applying Lemma
\ref{lem:edgeFirstAtHomeLocation}, we conclude that for each vertex in
$\alpha_s[x].G_i$, the set of incoming edges of $v$ is the same in
$\alpha_s[x].G_i$ and $\alpha_s[x'].G_i$. Therefore, $\alpha_s[x].G_i$
is a prefix of $\alpha_s[x'].G_i$.
\end{proof}



\begin{corollary}\label{cor:ObsPrefixOfLimit}
For each location $i$ and each positive integer $x$, $\alpha_s[x].G_i$
is a prefix of $G_i^\infty$.
\end{corollary}

Next, Lemma \ref{lem:limitGraphs} shows that for any pair $i,j$ of live locations $G_i^\infty = G_j^\infty$. We use Lemma \ref{lem:subgraphFutureLiveLocation}, which shows at any given point in the execution, the value of $G_i$ is a prefix of the value of $G_j$ at some later point in the execution, as a helper.


\begin{lemma}\label{lem:subgraphFutureLiveLocation}
For each positive integer $k$, every pair of locations $i$ and $j$
that are live in $t_D$, there exists a positive integer $k' \geq k$ such that
$\alpha_s[k].G_i$ is a prefix of $\alpha_s[k'].G_j$.
\end{lemma}

\begin{proof}
Fix $k$, $i$, and $j$ as in the hypotheses of the lemma.
Since $i$ is live, there exist a positive integers $k_1 \geq k$ and
$k_2 \geq k$ such that $\alpha_s[k_2].sendq_i$ contains
$(\alpha_s[k_1].G_i,j)$, and therefore, eventually the event
$send(\alpha_s[k_1].G_i,j)_i$ occurs which sends $\alpha_s[k_1].G_i$
to $j$. By Lemma \ref{lem:prefixOfFutureObs},
we know that
$\alpha_s[k].G_i$ is a prefix of $\alpha_s[k_1].G_i$. From the
properties of the channel automata we know that eventually event
$receive(\alpha_s[k_1].G_i,i)_j$ occurs in state (say)
$\alpha_s[k_3]$, where $k_3 > k_2$, and from the pseudocode, we know
that $\alpha_s[k_3+1].G_j$ is $\alpha_s[k_3].G_j \cup
\alpha_s[k-1].G_i$. Invoking Theorem \ref{thm:GiIsObs}, we know that $\alpha_s[k_3+1].G_j$ is an observation. 
Since we have already established that $\alpha_s[k].G_i$ is a prefix
of $\alpha_s[k_1].G_i$, we conclude that $\alpha_s[k].G_i$ is a
prefix of $\alpha_s[k_3+1].G_j$. Thus the lemma is satisfied for
$k' = k_3+1$.
\end{proof}

\begin{lemma}
\label{lem:limitGraphs}
For every pair of locations $i$ and $j$ that are live in $t_D$,
$G_i^\infty = G_j^\infty$.
\end{lemma}

\begin{proof}
Fix $i$ and $j$ as in the hypotheses of the lemma. Fix $z$ to be
either an edge or a vertex in $G_i^\infty$. By definition, there
exists a positive integer $k$ such that $\alpha_s[k].G_i$ contains
$z$. By Lemma \ref{lem:subgraphFutureLiveLocation}, we know that there
exists a positive integer $k'$ such that $\alpha_s[k'].G_j$ contains
$z$; applying Corollary \ref{cor:subgraphFuture}, we conclude that for
all $k'' \geq k'$, $\alpha_s[k''].G_j$ contains $z$. In other words,
$G_j^\infty$ contains $z$. Therefore, $G_i^\infty$ is a subgraph of 
$G_j^\infty$.

Reversing the roles of $i$ and $j$, we see that $G_j^\infty$ is a subgraph of
$G_k$. Therefore, $G_i^\infty = G_j^\infty$.
\end{proof}

Lemma \ref{lem:limitGraphs} allows us to define $G^\infty$ to be the
graph $G^\infty_i$ for any location $i$ that is live in $t_D$.
 
 
\begin{lemma}
\label{lem:GinftyEdges}
For every location $i$ such that $G^\infty$ contains an infinite number of
vertices whose location is $i$, for each vertex $v$ in $G^\infty$, there exists
a vertex $v'$ in $G^\infty$ whose location is $i$ and the edge $(v,v')$ is in
$G^\infty$.
\end{lemma} 

\begin{proof}
Fix $i$ and $v$ as in the hypotheses of the lemma. 
Since $G^\infty$ contains an infinite number of vertices whose location is $i$, we know that $i$ is live in $\alpha$, and therefore, an infinite number of events from $O_{D,i}$ occur in $\alpha$.

Since $v$ is in $G^\infty$, we know that there exists a positive integer $x_i$ such that $v$ is a vertex $\alpha_s[x_i].G_i$. Fix $e$ to be the first event from $O_{D,i}$ following $\alpha_s[x_i]$ in $\alpha$. Let the state preceding $e$ in $\alpha$ be $\alpha_s[x]$. From the pseudocode, we know that when $e$ occurs, a vertex of the form $(i,*,e)$ is inserted in $G_i$. Let this vertex be $v'$. From the insertion operation, we know that an edge $(v,v')$ is added to $G_i$. Therefore, $\alpha_s[x+1].G_i$ contains the edge $(v,v')$. From Corollary \ref{cor:ObsPrefixOfLimit}, we know that $\alpha_s[x+1].G_i$ is a prefix of $G^\infty$. Therefore, there exists
a vertex $v'$ in $G^\infty$ whose location is $i$ and the edge $(v,v')$ is in
$G^\infty$.
\end{proof}


 
 Finally, in Theorem \ref{thm:GinftyIsObs}, we establish that $G^\infty$ is an observation, and in Theorem \ref{thm:GinftyIsViable}, we establish that it is a viable observation.
 
 \begin{theorem}\label{thm:GinftyIsObs}
 $G^\infty$ is an observation.
 \end{theorem}
 \begin{proof}
 For any live location $j$, we know from Lemma \ref{lem:prefixOfFutureObs} that $\alpha_s[0].G_j,\alpha_s[1].G_j,\alpha_s[2].G_i,\ldots$ is an infinite sequence of finite observations, where $\alpha_s[x].G_j$ is a prefix of $\alpha_s[x+1].G_j$ for each natural number $x$. By definition, 
 we know that  $G^\infty_i$ is the limit of the infinite sequence $\alpha_s[0].G_j,\alpha_s[1].G_j,\alpha_s[2].G_i,\ldots$, and we know that $G^\infty = G^\infty_i$.
 
 By Lemma \ref{lem:GinftyEdges}, we know that for every vertex $v \in V$ and any location $i \in live(G^\infty)$,
there exists a vertex $v'$ with location
$i$ and $G^\infty$ contains an edge from $v$ to $v'$.
 Therefore, invoking Lemma \ref{lem:obs:limitObs}, we conclude that $G^\infty$ is an observation.
 \end{proof}


%

 Next we establish that $G^\infty$ is a viable observation. Intuitively, the proof is as follows. Recall that $\alpha|_{O_D \cup \hat{I}} = t_D$. For any live location $i$, $G^\infty$ contains all the AFD output events from $t_D$ that occur at $i$ and in the same order in which they occur at $i$. For any non-live location $i$, $G^\infty$ contains some prefix of all the AFD output events from $t_D$ that occur at $i$ and in the same order in which they occur at $i$. Also, there is an edge from a vertex $v_1 = (i_1,k_1,e_1)$ to another vertex $v_2 = (i_2,k_2,e_2)$ in $G^\infty$ only if $e_1$ occurs before $e_2$ in $t_D$. Therefore, there must exist some sampling $t'_D$ of $t_D$ such that $t_D|_{O_D}$ is a topological sort of $G^\infty$. Invoking closure under sampling, we conclude that $t'_D$ must be in $T_D$, and therefore $G^\infty$ is viable. The formal theorem statement and proof follows.
 
 \begin{theorem}\label{thm:GinftyIsViable}
 $G^\infty$ is a viable observation for $D$.
 \end{theorem}

\begin{proof}
Recall that $\alpha|_{O_D \cup \hat{I}} = t_D \in T_D$. We complete the proof by showing that there exists a trace $t'_D \in T_D$ that is compatible with $G^\infty$; specifically, we show that there exists a topological sort $\nu$ of the vertices of $G^\infty$ and a sampling $t'_D$ of $t_D$ such that $t'_D|_{O_D} = \epsilon|_{O_D}$, where $\epsilon$ is the event-sequence of $\nu$.

Let $\widehat{\nu}$ be the set of all topological sorts of the vertices of $G^\infty$, and let $\widehat{\epsilon}$ be the set of all event-sequences such that each $\widehat{\epsilon} = \set{ \epsilon' | \epsilon' \text{ is the event-sequence of some } \nu' \in \widehat{\nu}}$. From the pseudocode, we see that $\widehat{\epsilon}$ is the set of all $\epsilon'$ such that (1) for each location $j$, $\epsilon'|_{O_{D,j}}$ is a prefix of $\alpha|_{O_{D,j}} = t_D|_{O_{D,j}}$, and (2) for each location $j$ that is live in $t_D$, $\epsilon'|_{O_{D,j}} = \alpha|_{O_{D,j}} = t_D|_{O_{D,j}}$.

For any edge $(v_1,v_2)$ in $G^\infty$ we know that there exists a location $i$ and a positive integer $x$ such that $\alpha_s[x].G_i$ contains the edge $(v_1,v_2)$; applying Lemma \ref{lem:edgeHappensBefore}, the event of $v_1$ occurs before the event of
$v_2$. Therefore, in $G^\infty$, for every edge $(v_1,v_2)$, the event
of $v_1$ occurs before the event of $v_2$. Therefore, (3) for every pair of vertices $v_1,v_2$ of $G^\infty$, it is not the case that the event of $v_1$ occurs before the event of $v_2$ in the event sequence $\epsilon'$ of every topological sort $\nu'$ of the vertices of $G^\infty$, and the event of $v_1$ does not occur before the event of $v_2$ in $t_D$.

From (1), (2), and (3), we conclude that there must exist an event-sequence $\epsilon \in \widehat{\epsilon}$, such that (1) for each location $j$, $\epsilon|_{O_{D,j}}$ is a prefix of $\alpha_{O_{D,j}} = t_D|_{O_{D,j}}$, (2) for each location $j$ that is live in $t_D$, $\epsilon|_{O_{D,j}} = \alpha_{O_{D,j}} = t_D|_{O_{D,j}}$, and (3) for every pair of events $e_1, e_2$ in $\epsilon$, if $e_1$ occurs before $e_2$ in $t_D$, then $e_1$ occurs before $e_2$ in $\epsilon$. Therefore, there exists a sampling $t'_D$ of $t_D$ such that  $t'_D|_{O_D} = \epsilon|_{O_D}$.
By closure under sampling we know that $t'_D \in T_D$. 
Thus, by definition, $G^\infty$ is viable for $D$.
\end{proof}


We have seen so far that in any fair execution $\alpha$ of the system, at
each live location $i$, $G_i$ evolves as an ever growing observation such that
the limit $G^\infty$ of $G_i$ in $\alpha$ is a viable observation for $D$. 

%

\subsubsection{Identifying the smallest decision gadget}\label{subsubsec:proofSmallestDecisionGadget}
Next, we show that $\mathcal{R}^{G^\infty}$ has at least one non-$\bot$ decision gadget. Let $Y$ be the first non-$\bot$ decision gadget in $\mathcal{R}^{G^\infty}$. We show that at
each live location $i$, eventually, $\mathcal{R}^{G_i}$ will contain the decision gadget $Y$, and importantly, eventually forever, $Y$ remains the first non-$\bot$ decision gadget of $\mathcal{R}^{G_i}$. By Theorem \ref{thm:ViableGHasDecisionGadgets}, we know that the critical location of $Y$ is a live location. However, since for all the live processes $i$, the first non-$\bot$ decision gadget of $\mathcal{R}^{G_i}$ converges to $Y$, we know that all the live locations converge to the same live location, which is output of the $A^\Omega$. Thus, $A^\Omega$ solves $\Omega_f$ using AFD $D$.

\begin{corollary}\label{cor:GinftyHasDecisionGadget}
$\mathcal{R}^{G^\infty}$ contains at least one decision gadget.
\end{corollary}
\begin{proof}
Follows from Theorems \ref{thm:ViableGHasDecisionGadgets} and \ref{thm:GinftyIsViable}.
\end{proof}

Applying the above Corollary, we know that $\mathcal{R}^{G^\infty}$ contains a decision gadget. Applying  Theorem \ref{thm:smallestGadget}, let $Y_{min}$ be $first(\mathcal{R}^{G^\infty})$ (the first non-$\bot$ decision gadget in $\mathcal{R}^{G^\infty}$).

\begin{lemma}\label{lem:minGadgetFound}
For each location $i$ that is live in $t_D$, there exists a positive integer $x$ such that for all positive integers $x' \geq x$, $Y_{min}$ is the first non-$\bot$ decision gadget in $\alpha_s[x'].\mathcal{R}^G_i$.
\end{lemma}
\begin{proof}
Fix a location $i$ that is live in $t_D$.
Invoking Theorem \ref{thm:GiIsObs}, we know that for each positive integer $x$, $\alpha_s[x].G_i$ is an observation.
Since $G^\infty = \lim_{x \rightarrow \infty} \alpha_s[x].G_i$ is a viable observation for $D$ and $t_D$ is compatible with $G^\infty$, we invoke Lemma \ref{lem:prefixOfFutureObs} to conclude that $\alpha_s[1].\mathcal{R}^G_i, \alpha_s[2].\mathcal{R}^G_i, \ldots$ is an infinite sequence of finite observations that converge to $G^\infty$.

Thus, the conclusion follows immediately from the application of Lemma \ref{lem:minGadgetFoundInAllGs}.

\end{proof}

\begin{theorem}
The algorithm $A^\Omega$ solves $\Omega_f$ using AFD $D$, where $f < n$.
\end{theorem}
\begin{proof}
Fix a fair execution $\alpha$ of the system consisting of $A^\Omega$, the channel automata, and the crash automaton such that $\alpha|_{O_D \cup \hat{I}} \in T_D$ and at most $f$ locations crash in $\alpha$. Denote $\alpha|_{O_D \cup \hat{I}}$ as $t_D$. For each location $i$ that is live in $t_D$, let $G_i^\infty$ denote $\lim_{x \rightarrow \infty} \alpha_s[x].G_i$. Applying Lemma \ref{lem:limitGraphs}, we know that for each location $j$ that is live in $t_D$, $G_i^\infty= G_j^\infty = G^\infty$. By Theorem \ref{thm:GinftyIsViable}, we know that $G^\infty$ is a viable observation for $D$. By Corollary \ref{cor:GinftyHasDecisionGadget}, we know that $\mathcal{R}^{G^\infty}$ contains at least one decision gadget. Applying  Theorem \ref{thm:smallestGadget} let $Y_{min}$ be the first non-$\bot$ decision gadget in $\mathcal{R}^{G^\infty}$. Applying Lemma \ref{lem:minGadgetFound}, we know that for each location $i$ that is live in $t_D$, eventually and permanently, $Y_{min}$ is the first non-$\bot$ decision gadget of $\mathcal{R}^G_i$. Thus, for each location $i$ that is live in $t_D$, eventually and permanently, when an event from $O_{D,i}$ occurs in $\alpha$, $(fdout,i)$ is appended to $sendq_i$, where $fdout$ is the critical location of $Y_{min}$. Therefore, for each location $i$ that is live in $t_D$, some suffix of $\alpha|_{FD_{\Omega,i}}$ is the infinite sequence over $FD_\Omega(fdout)_i$. Applying Theorem \ref{thm:ViableGHasDecisionGadgets}, we know that $fdout$ is live in $G^\infty$, and therefore, $fdout$ is live in $\alpha$. In other words, $\alpha|_{O_\Omega \cup \hat{I}} \in T_{\Omega_f}$.
\end{proof}

\section{Conclusion}\label{sec:conclusion}

We have shown that for any strong sampling AFD sufficient to solve consensus, the executions of the system that solves consensus using this AFD must satisfy the following property. For any fair execution, the events responsible for the transition from a bivalent execution to a univalent execution must occur at location that does not crash.
Using the above result, we have shown that $\Omega$ is a weakest strong-sampling AFD to solve consensus. The proof is along the lines similar to the original proof from \cite{chan:twfdf}. However, our proof is much more rigorous and does not make any implicit assumptions or assertions. Furthermore, the notion of observations and tree of executions introduced in Sections \ref{sec:observations}  and \ref{subsec:treeOfExec} and their properties may be of independent interest themselves.

\bibliographystyle{plain}
\bibliography{ref}

\end{document}